\documentclass[sigconf]{acmart}
\usepackage[utf8]{inputenc}
\usepackage[T1]{fontenc}

\AtBeginDocument{%
  }

\usepackage{xspace}
\usepackage{colonequals}
\usepackage{amsmath}
\usepackage{scalerel}
\usepackage{mathtools}
\usepackage{balance}

\usepackage{xcolor}
\usepackage{xcolor-solarized}
\usepackage{tikz}
\usetikzlibrary{graphs}
\usetikzlibrary{graphs.standard}
\usetikzlibrary{arrows,arrows.meta,calc}
\usetikzlibrary{shadows}
\usetikzlibrary{shapes.multipart}
\usetikzlibrary{positioning}
\usetikzlibrary{shadows}
\usetikzlibrary{shapes.multipart}

\usepackage{circuitikz}

\DeclareMathOperator*{\bigplus}{\scalerel*{+}{\sum}}

\renewcommand{\paragraph}[1]{\par\medskip\noindent\textbf{#1.}}

\theoremstyle{plain}
\newtheorem{theorem}{Theorem}[section]

\newtheorem{proposition}[theorem]{Proposition}
\newtheorem{lemma}[theorem]{Lemma}
\newtheorem{corollary}[theorem]{Corollary}
\newtheorem*{claim}{Claim}
\theoremstyle{definition}
\newtheorem{example}[theorem]{Example}

\makeatletter
\renewenvironment{proof}[1][\proofname]{\par
  \pushQED{\qed}%
  \normalfont \topsep6\p@\@plus6\p@\relax
  \trivlist
\item\relax
  {\scshape
    #1\@addpunct{.}}\hspace\labelsep\ignorespaces
}{%
\popQED\endtrivlist\@endpefalse
}
\makeatother

\usepackage{xspace}

\newif\ifremark
\remarktrue

\newcommand{\cameraready}[2]{#2}

\def\iff{\quad\text{\it iff}\quad}
\def\minus{\mathop{\setminus}}
\newcommand{\dom}{\mathit{dom}}

\def\clap#1{\hbox to 0pt{\hss#1\hss}}

\def\mathrlap{\mathpalette\mathrlapinternal}

\def\mathrlapinternal#1#2{%
\rlap{$\mathsurround=0pt#1{#2}$}}

\newcommand{\PTIME}{\textnormal{PTIME}\xspace}

\newcommand{\ASPACE}{\textnormal{ASPACE}\xspace}

\newcommand{\EXPTIME}{\textnormal{EXPTIME}\xspace}

\newcommand{\EXPSPACE}{\textnormal{EXPSPACE}\xspace}
\newcommand{\TWOEXPTIME}{\textnormal{2EXPTIME}\xspace}

\newcommand{\ALCIF}{\ensuremath{\mathcal{A\hspace{-0.8pt}L\hspace{-0.8pt}C\hspace{-1.5pt}I\hspace{-1.8pt}F}}\xspace}
\newcommand{\ALCI}{\ensuremath{\mathcal{A\hspace{-0.8pt}L\hspace{-0.8pt}C\hspace{-1.5pt}I}\xspace}}
\newcommand{\HornALCIF}{\ensuremath{\text{Horn-}\mathcal{A\hspace{-0.8pt}L\hspace{-0.8pt}C\hspace{-1.5pt}I\hspace{-1.8pt}F}}\xspace}
\newcommand{\ALC}{\ensuremath{\mathcal{A\hspace{-0.8pt}L\hspace{-0.8pt}C}\xspace}}

\def\P{\ensuremath{\mathcal{P}\xspace}}
\def\L{\ensuremath{\mathcal{L}\xspace}}
\def\A{\ensuremath{\mathcal{A}\xspace}}

\def\N{\ensuremath{\mathcal{N}\xspace}}
\def\Q{\ensuremath{\mathcal{Q}\xspace}}
\def\F{\ensuremath{\mathcal{F}\xspace}}
\def\M{\ensuremath{\mathcal{M}\xspace}}

\def\G{\ensuremath{\mathcal{G}\xspace}}
\def\T{\ensuremath{\mathcal{T}\xspace}}

\newcommand{\interval}[1]{{\ensuremath{\mathord{\text{\normalfont\fontfamily{lmtt}\selectfont{}#1}}}}}
\newcommand{\NONE}{\interval{0}\xspace}
\newcommand{\ONE}{\interval{1}\xspace}
\newcommand{\MAYBE}{\interval{?}\xspace}
\newcommand{\MANY}{\interval{*}\xspace}
\newcommand{\PLUS}{\interval{+}\xspace}

\newcommand{\PCOne}{{\mathrm{PC}\kern-0.1pt1}}

\newcommand{\trim}{\mathit{trim}}

\newcommand{\Vaccine}{\textsl{Vaccine}\xspace}
\newcommand{\Antigen}{\textsl{Antigen}\xspace}
\newcommand{\Pathogen}{\textsl{Pathogen}\xspace}
\newcommand{\exhibits}{\textsl{exhibits}\xspace}
\newcommand{\targets}{\textsl{targets}\xspace}
\newcommand{\designTarget}{\textsl{designTarget}\xspace}
\newcommand{\crossReacting}{\textsl{crossReacting}\xspace}

\copyrightyear{2023} 
\acmYear{2023} 
\setcopyright{rightsretained} 
\acmConference[PODS '23]{Proceedings of the 42nd ACM SIGMOD-SIGACT-SIGAI Symposium on Principles of Database Systems}{June 18--23, 2023}{Seattle, WA, USA}
\acmBooktitle{Proceedings of the 42nd ACM SIGMOD-SIGACT-SIGAI Symposium on Principles of Database Systems (PODS '23), June 18--23, 2023, Seattle, WA, USA}
\acmDOI{10.1145/3584372.3588654}
\acmISBN{979-8-4007-0127-6/23/06}

\makeatletter
\gdef\@copyrightpermission{
  \begin{minipage}{0.2\columnwidth}
   \href{https://creativecommons.org/licenses/by/4.0/}{\includegraphics[width=0.90\textwidth]{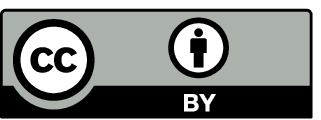}}
  \end{minipage}\hfill
  \begin{minipage}{0.8\columnwidth}
   \href{https://creativecommons.org/licenses/by/4.0/}{This work is licensed under a Creative Commons Attribution International 4.0 License.}
  \end{minipage}
  \vspace{5pt}
}
\makeatother

\begin{document}

\title{Static Analysis of Graph Database Transformations}

\author{Iovka Boneva}
\email{iovka.boneva@univ-lille.fr}
\orcid{0000-0002-2696-7303}
\affiliation{%
  \institution{Univ. Lille, CNRS, UMR 9189 CRIStAL}
  \city{F-59000 Lille}
  \country{France}
}

\author{Benoît Groz}
\email{groz@lri.fr}
\orcid{0000-0001-7292-6409}
\affiliation{%
  \institution{Univ. Paris Saclay, CNRS, UMR 9015 LISN}
  \city{91405 Orsay}
  \country{France}
}

\author{Jan Hidders}
\email{j.hidders@bbk.ac.uk}
\orcid{0000-0002-8865-4329}
\affiliation{%
  \institution{Birkbeck, University of London}
  \city{London}
  \country{United Kingdom}
}

\author{Filip Murlak}
\email{f.murlak@uw.edu.pl}
\orcid{0000-0003-0989-3717}
\affiliation{%
  \institution{University of Warsaw}
  \city{Warsaw}
  \country{Poland}
}

\author{Sławek Staworko}
\email{slawek.staworko@relational.ai}
\orcid{0000-0003-3684-3395}
\affiliation{%
   \institution{RelationalAI}
   \city{Berkeley}
   \country{USA}
}
\affiliation{%
   \institution{Univ. Lille, CNRS, UMR 9189 CRIStAL}
   \city{F-59000 Lille}
   \country{France}  
}

\begin{abstract}
  We investigate graph transformations, defined using Datalog-like rules based
  on acyclic conjunctive two-way regular path queries (acyclic C2RPQs), and we
  study two fundamental static analysis problems: \emph{type checking} and
  \emph{equivalence} of transformations in the presence of graph
  schemas. Additionally, we investigate the problem of \emph{target schema
    elicitation}, which aims to construct a schema that closely captures all
  outputs of a transformation over graphs conforming to the input schema. We
  show all these problems are in \EXPTIME by reducing them to C2RPQ containment
  modulo schema; we also provide matching lower bounds. We use \emph{cycle
    reversing} to reduce query containment to the problem of unrestricted
  (finite or infinite) satisfiability of C2RPQs modulo a theory expressed in a
  description logic. 
\end{abstract}



\maketitle

\section{Introduction}
\label{sec:introduction}

The growing adoption of graph databases calls for suitable data processing
methods. Query languages for graph databases typically define their semantics as
a set of tuples, which alone is inadequate for scenarios such as (materialized)
graph database views and data migration in the context of schema
evolution~\cite{BFGHOV19}, with the schema describing the expected structure of the graph. A more adequate mechanism is that of a
\emph{transformation}, which takes a graph as input and produces a graph on the
output.
\begin{example}
  \label{ex:medical-database}
  Consider a scenario where the schema of a medical knowledge graph undergoes
  changes due to advances in the understanding of biomolecular processes. The
  purpose of this knowledge graph is to catalog vaccines based on the antigen
  they are designed to target and to identify the pathogens that exhibit the
  antigens, each antigen being exhibited by at least one pathogen. Additionally,
  some pairs of antigens are known to be \emph{cross reacting}: if a vaccine $v$
  targets an antigen $x$ that is cross reacting with an antigen $y$, then $v$
  also targets $y$. Thus, the set of all antigens targeted by a vaccine is
  represented implicitly.

  The schema $S_0$ of the original knowledge graph is presented in
  Figure~\ref{fig:medical-database-schema} as a graph itself.
  \begin{figure}[htb]
  \begin{tikzpicture}[
    >=stealth',
    punkt/.style={circle,minimum size=0.1cm,draw,fill,inner sep=0pt, outer sep=0.125cm},
    dot/.style={}
    ]
  \path[use as bounding box, red!75!black] (-1,0.75) rectangle (6.55,-1.5);
      
  \small
  \begin{scope}[]
    \node at (-1.25,0) {$S_0$:};
  \begin{scope}[xshift=0cm]
    \node[dot] (x1) at (2.35,0.45) {};
    \node[dot] (x2) at (3,0.125) {};
    \node[dot] (x3) at (3.65,0.45) {};
    \node[dot] (x4) at (3.75,1) {};
    \node[dot] (x5) at (2.25,1) {};
    \fill[red!20!white] plot[smooth cycle,tension=0.75] coordinates { (x1)  (x2)  (x3)  (x4) (x5) };
  \end{scope}
    
  \path (0,0) node[punkt] (vaccine) {} node[left] {\Vaccine};
  \path (3,0) node[punkt] (antigen) {} node[below] {\Antigen};
  \path (5.5,0) node[punkt] (pathogen) {} node[right]  {\Pathogen};

  \draw[->] (antigen)
  .. controls +(110:0.85) and +(70:0.85) ..
  node[above] {\crossReacting}
  node[left,pos=0.2] {\footnotesize\MANY}
  node[right,pos=0.8] {\footnotesize\MANY}
  (antigen);

  \draw (pathogen) edge[->]
  node[above=-1pt] {\exhibits}
  node[above,pos=0.15] {\footnotesize\PLUS}
  node[above,pos=0.85] {\footnotesize\MANY}
  (antigen);

  \draw (vaccine) edge[->]
  node[above=-1pt] {\designTarget}
  node[above,pos=0.15] {\footnotesize\MANY}
  node[above,pos=0.85] {\footnotesize\ONE}
  (antigen);

  \begin{scope}[yshift=-1cm]
  \node at (-1.25,0) {$S_1$:};

  \begin{scope}[yshift=-0.4cm]
    \node[dot] (x1) at (1.5,-0.275) {};
    \node[dot] (x2) at (2.5,-0.15) {};
    \node[dot] (x3) at (2.5,0.25) {};
    \node[dot] (x4) at (1.5,0.15) {};
    \node[dot] (x5) at (.5,0.25) {};
    \node[dot] (x6) at (.5,-0.15) {};
    \fill[green!20!white] plot[smooth cycle] coordinates { (x1) (x2) (x3) (x4) (x5) (x6)};
  \end{scope}

  \path (0,0) node[punkt] (vaccine) {} node[left] {\Vaccine};
  \path (3,0) node[punkt] (antigen) {} node[below right=0.0cm and -0.1cm] {\Antigen};
  \path (5.5,0) node[punkt] (pathogen) {} node[right]  {\Pathogen};

  \draw (pathogen) edge[->]
  node[above=-1pt] {\exhibits}
  node[above,pos=0.15] {\footnotesize\PLUS}
  node[above,pos=0.85] {\footnotesize\MANY}
  (antigen);

  \draw (vaccine) edge[->]
  node[above=-1pt] {\designTarget}
  node[above,pos=0.15] {\footnotesize\MANY}
  node[above,pos=0.85] {\footnotesize\ONE}
  (antigen);

  \draw[bend angle=20] (vaccine) edge[->,bend right]
  node[below=-1pt] {\targets}
  node[below,pos=0.15] {\footnotesize\MANY}
  node[below,pos=0.85] {\footnotesize\PLUS}
  (antigen);
  \end{scope}
  \end{scope}

  \end{tikzpicture}
  \caption{Evolving schema of a medical knowledge graph.}
  \label{fig:medical-database-schema}
\end{figure}
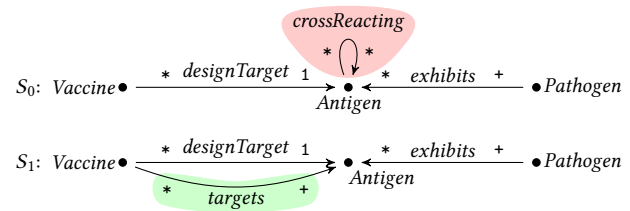
It specifies the allowed node and edge labels, and expresses participation
constraints on edges in a manner that is typical for data modeling languages, 
e.g.,
\tikz[]{
  \path[use as bounding box] (-0.1,-0.1) rectangle (1.1,0.125);
  \node (A) at (0,0) {\small \sl A};
  \node (B) at (1,0) {\small \sl B};
  \draw (0.125,-0.05) edge[-stealth']
  node[pos=0.45,above=-2pt] {\scriptsize \sl r}
  node[pos=0.2,above=-2.25pt]{\scriptsize \MANY}
  node[pos=0.75,above=-2.25pt] {\scriptsize \ONE}(0.9,-0.05);
}
indicates that every $A$-node has one outgoing $r$-edge to a $B$-node but a $B$-node may have arbitrarily many incoming $r$-edges from $A$-nodes.

Now, suppose that new findings refute the rule of cross-re\-ac\-ti\-vity of
antigens. The \textsl{cross-reacting} edges between antigens are no longer
adequate for representing information about the antigens that a vaccine targets,
and so, in the new schema $S_1$, this information is recorded explicitly with
\textsl{targets} edges. Since up to that point, the knowledge graph did not
contain any data points that contradicted the cross-reactivity rule, the logic
of the rule can be used to transform the old knowledge graph to one that
conforms to the new schema. Afterwards \textsl{cross-reacting} edges are
removed. \qed
\end{example}

In the present paper, we study two classical problems of static analysis on
graph transformations: \emph{type checking}, that verifies if for every graph
conforming to the source schema the transformation outputs a graph conforming to
the target schema, and \emph{equivalence}, that verifies if two transformations
produce the same output for every graph conforming to the source
schema. Additionally, when the target schema is not known, we investigate the
problem of \emph{target schema elicitation} that constructs the
containment-minimal target schema that captures the graphs produced by the
transformation.

We study \emph{executable} graph transformations defined with Datalog-like
rules. The rules specify how to construct the output graph from the results of
regular path queries evaluated over the input graph. To allow multiple copies of
the same input node the rules use \emph{node constructors}, essentially
explicit Skolem functions that create nodes. As an example, the
cross-reactivity rule from Example~\ref{ex:medical-database} gives rise to the
following graph transformation rule
\[
  \targets(f_V(x),f_A(y))\leftarrow
  (\designTarget\cdot\crossReacting^*)(x,y)\,,
\]
where $f_V(x)$ and $f_A(y)$ are constructors of \Vaccine and \Antigen nodes
respectively. 
The two constructors can, for instance, have the following  definitions 
$f_V(x)=(\Vaccine,x)$ and
$f_A(y)=(\Antigen,y)$; essentially, they take the identifiers of the original nodes and decorate them with their type.

We investigate transformations that use only \emph{acyclic two-way conjunctive
  regular path queries} (acyclic C2RPQs), which is argu\-ably of practical
relevance in the context of graph transformations. For instance, we have found
no cyclic queries in the transformations implementing graph data migration
between consecutive versions of the FHIR data format~\cite{FHIR,PrSoJi17} (Fast
Healthcare Interoperability Resources is an international standard for
interchange of medical healthcare data). Our constructions
rely on acyclicity of C2RPQs to obtain relatively low computational complexity. We argue that the acyclicity assumption cannot be lifted 
without a significant complexity increase (see Section~\ref{sec:conclusions}).

Node constructors are closely related to object creating
functions~\cite{HuYo90,HuYo91}. Our use of node constructors is inspired by
analogous constructions in transformation languages such as
R2RML~\cite{CMRRS16,Sequeda13,R2RML}, where node IRIs are typically obtained by
concatenation of a URL prefix and the key values of a object represented by the
constructed node. Our node constructors can have an arbitrary arity, thus
allowing for instance to create nodes in the target graph that represent
relationships (edges) between nodes in the source graph. To isolate the concern
of possible overlaps between node constructors, we make the natural assumption
that node constructors are injective, have pair-wise disjoint ranges, and for
every node kind (label) a single dedicated node constructor is used. These
assumptions allow us to remove the need to analyze the definitions of node
constructors, which is out of the scope of the present paper, and they are
consistent with how the analogous constructions are used in languages such as
R2RML and FHIR mapping language.

For schemas, we employ a natural formalism of \emph{graph schemas with
  participation constraints}, inspired by standard data modeling languages such
as Entity-Relationship diagrams~\cite{Chen75}, and already studied, for
instance, in the context of graph database evolution~\cite{BFGHOV19}. Such
schemas allow one to declare the available labels of nodes and edges and to express
participation constraints.  In contrast to more expressive languages as ShEx and
SHACL~\cite{SBLGHPS15,CoReSa18}, our formalism allows a \emph{single label per
  node}, which determines the node type. Thus, roughly speaking, our schema
formalism is to ShEx and SHACL what DTD is to XML Schema.

The key contributions of the present paper are as follows. 
\begin{enumerate}
\item 
We define graph database transformations and we reduce the problems of interest to \emph{containment} of C2RPQs in unions of acyclic C2RPQs \emph{modulo schemas}.
\item 
We reduce the query containment problem to the unrestricted (finite or infinite) satisfiability of a C2RPQ modulo a set of constraints expressed in the Horn fragment of a description logic known as \ALCIF.

The reduction involves an application of the \emph{cycle reversing} technique \cite{CosmadakisKV90,GarciaLS14}, carefully tailored to our needs.
\item 
The unrestricted satisfiability problem for \ALCIF can be solved in \EXPTIME owing to a simple model property~\cite{CalvaneseOS11},  but applying this result directly to the instance obtained via cycle reversing would lead to  doubly exponential complexity due to an exponential blow-up inherent to cycle reversing. 
We provide a new algorithm with \emph{improved complexity bounds}, 
which allows to accommodate the blow-up while keeping the overall complexity in \EXPTIME.

We also reformulate the simplicity of models in terms of a graph-theoretical notion of \emph{$(k,l)$-sparsity}~\cite{LeeS08}, which allows to streamline the reasoning.
\end{enumerate}
These reductions allow to solve all problems of interest in \EXPTIME and we also
establish the matching lower bounds.

The paper is organized as follows. In Section~\ref{sec:related} we discuss
related work. In Section~\ref{sec:preliminaries} we introduce basic notions. In
Section~\ref{sec:transformations} we define graph transformations and the
problems of interest, which we reduce to query containment modulo schema. In
Section~\ref{sec:containment} we reduce the latter to satisfiability of a query
modulo \HornALCIF theory, which we solve in Section~\ref{sec:satisfiability}. In
Section~\ref{sec:conclusions} we summarize our findings and identify directions
of future work. \cameraready{A technical report containing full proofs can be found in \cite{boneva:hal-03937274}.}{Full proofs and some standard definitions have been moved to Appendix.}

\section{Related Work}
\label{sec:related}

\textbf{Graph-based data models} have been proposed in various forms and shapes since the 1980s \cite{AngGut08}.

The proposals in the 1980s and 1990s included labeled graphs \cite{GPVdBvG94}, graphs where certain nodes represent complex values \cite{KupVar93, Hidders03}, graphs where nodes have associated complex values \cite{Abiteboul87, AbiKan89}, and graphs where nodes are associated with nested graphs \cite{LevPou90}. 
More recently the RDF data model \cite{RDF04} and the Property Graph data model \cite{Angles18} have become popular. RDF graphs are similar to labeled graphs except that nodes are unlabeled and participate in at least one edge, and the labels of edges can be nodes and participate in edges. Property Graphs are also similar to labeled graphs except that nodes and edges have multiple labels and properties, and edges have identity. In our work we assume one of the simplest models, namely, labeled graphs where nodes have multiple labels and edges have a single label; our schemas require exactly one label per node. Since we focus here on transformations of the graph structure, we have no explicit notion of value associated with nodes and edges, but there are straightforward ways of adding this, as is done for example in \cite{GPVdBvG94}.

The term \textbf{graph transformations} can refer to different formalisms \cite{Rozenberg1997handbook}:
the purpose of graph grammars is to define graph languages;
algebraic graph transformations are mainly used to mo\-del systems with infinite behavior and are not functional (they produce multiple outputs on single input).
Therefore, not only are these formalisms ill-suited for defining transformations of graph databases, but also the problems studied for them are unrelated to the problems we study here. Monadic second-order (MSO) graph transductions \cite{Courcelle1994} can capture our transformations only when restricted to unary node constructors; moreover, resorting to MSO logic typically incurs a prohibitive complexity overhead. 

\textbf{Transformation languages for graph databases} are often based on Datalog extended with node-creation syntax in the head of the rules. It could be just a variable that is not bound in the body of the rule, like in IQL \cite{AbiKan89} and G-Log \cite{PePeTa95}; this ensures a fresh node is created for each valuation that makes the body true. Another option is to replace the unbound variable with a term consisting of a constructor function (sometimes called a Skolem function) applied to bound variables, like in O-logic \cite{Maier86} and F-logic \cite{KifLau89}; the constructor creates a fresh node when called for the first time for certain arguments, and after that the same node for the same arguments. 
We adopt the idea of node constructors because we believe it provides a powerful and intuitive way to control the identity of new nodes.

A different proposal, based on structural recursion, is offered by UnQL \cite{Buneman2000}, but the underlying data model considers graphs equivalent if they are bisimilar, which makes the expressive power quite different.

Graph transformations can also be expressed using query languages such as SPARQL and Cypher.

Nevertheless, we believe that a rule-based transformation language is more convenient for defining transformations and it can co-exist with an expressive query language. 
For instance, in the XML world, XSLT \cite{xslt} (rule-based) focuses on transformations, while XQuery \cite{xquery} is mostly used for querying XML data.

In the context of data exchange, \textbf{schema mappings} provide a declarative way to define database transformations \cite{Fagin2005,Calvanese2011simplifying,Barcelo2013}.
Our transformations could be simulated by considering canonical solutions for plain SO-tgds \cite{Arenas2013}

extended to allow acyclic C2RPQs in rule bodies.
Note, however, that equivalence is undecidable for plain SO-tgds with keys \cite{Feinerer2015}, and open for plain SO-tgds~\cite{KolaitisPSS20}.

The \textbf{static type checking problem} originates in formal language theory and has been studied for finite state transducers on words and for various kinds of tree transducers, including some designed to capture XML transformation languages \cite{Milo2003,Maneth2005,Martens2007,MaNeGy08}.
Type checking has also been studied for graph transformations.
In \cite{Hidders03} labelled graphs are transformed using addition, deletion, and reduction operations, and type checking is investigated for schemas similar to ours but without participation constraints. 
The typing problem for UnQL is studied in \cite{Inaba2011}, but the approach relies on schemas specifying graphs up to bisimulation, which limits their power to express participation constraints.  
Regarding transformations defined by schema mappings, if the mapping does not define target constraints, then the target schema is simply a relational signature and type checking is reduced to trivial syntactic check, and as such it is irrelevant.
This is most often the case for graph schema mappings \cite{Calvanese2011simplifying,Barcelo2013}, with seldom exceptions such as \cite{Boneva2020} for mapping relational to graph-shaped data.
Their notion of consistency is related to type checking, but is studied for a simpler formalism without path queries.
In the context of XML schema mappings, absolute consistency can be seen as a counterpart of type checking for non-functional transformations \cite{Bojanczyk2013}.

\section{Preliminaries}
\label{sec:preliminaries}
\paragraph{Graphs}
We fix an enumerable set $\N$ of node identifiers, a recursively enumerable set $\Gamma$ of
node labels, and an recursively enumerable set $\Sigma$ of edge labels. We work with labeled
directed graphs, and in general, a node may have multiple labels while an edge
has precisely one label.
We allow, however, multiple edges between the same pair
of nodes, as long as these edges have different labels. We model graphs as
relational structures over unary relation symbols $\Gamma$ and binary relation
symbols $\Sigma$. That is, a graph $G$ is a pair $\big(\dom(G), \cdot^G\big)$
where $\dom(G) \subseteq \N$ is the set of nodes of $G$ and the function
$\cdot^G$ maps each $A\in\Gamma$ to a set $A^G \subseteq\dom(G)$ and each
$r\in\Sigma$ to a binary relation $r^G\subseteq\dom(G)\times\dom(G)$. A graph
$G$ is \emph{finite} if $\dom(G)$ is finite and $A^G$ and $r^G$ are empty for
all but finitely many $A\in\Gamma$ and $r\in\Sigma$. In the sequel, we use
$u,v,\ldots$ to range over node identifiers, $A,B,C,\ldots$ to range over node
labels, and $r,r',\ldots$ to range over edge labels. Also, we use $r^-$ for
inverse edges and let $(r^-)^G = \big\{(u,v) \mid (v,u)\in r^G\big\}$. We let
$\Sigma^\pm=\Sigma\cup\{r^-\mid r\in\Sigma\}$ and use $R,R',\ldots$ to range
over $\Sigma^\pm$.

\paragraph{Schemas}
We consider a class 
of schemas that constrain the number of edges between nodes of given labels and
we express these constraints with the usual symbols: $\MAYBE$ for at most one,
$\ONE$ for precisely one, $\PLUS$ for at least one, $\MANY$ for arbitrary many,
and $\NONE$ for none. We focus on these basic cardinality constraints that are most commonly used in practice; e.g., Chen's  original ER diagrams only used those~\cite{Chen75}. In fact, we were unable to find any non-basic cardinality constraints in the FHIR specifications~\cite{FHIR}, while in the SHACL schemas in Yago 4.0~\cite{YAGO} we found only one: a person may have at most two parents.

Now, a  \emph{schema} is a triple
$S=(\Gamma_S,\Sigma_S,\delta_S)$, where $\Gamma_S\subseteq\Gamma$ is a finite
set of allowed node labels, $\Sigma_S\subseteq\Sigma$ is a finite set of allowed
edge labels, and
$\delta_S:\Gamma_S\times\Sigma_S^\pm\times\Gamma_S\rightarrow\{\MAYBE,\ONE,\PLUS,\MANY,\NONE\}$.
Schemas can be presented as graphs themselves, interpreted as illustrated next. 
\begin{example}

  Take the schema $S_0$ in Figure~\ref{fig:medical-database-schema}
  and consider, for instance, the \designTarget edge. It indicates that every
  \Vaccine has a single design target \Antigen, in symbols
  \[
    \delta_{S_0}(\Vaccine,\designTarget,\Antigen)=\ONE\,,
  \]
  and that every \Antigen may be the design target of an arbitrary
  number of \Vaccine{}s, in symbols
  \[
    \delta_{S_0}(\Antigen,\designTarget^-,\Vaccine)=\MANY\,.
  \]
  Edges that are not present are implicitly forbidden, e.g., no \exhibits
  edge is allowed from \Vaccine to \Pathogen:
  \begin{gather*}
    \delta_{S_0}(\Vaccine,\exhibits,\Pathogen)=\NONE\,,\\
    \delta_{S_0}(\Pathogen,\exhibits^-,\Vaccine)=\NONE\,.
    \tag*{\qed}
  \end{gather*}
\end{example}
\noindent
Now, a graph $G$ \emph{conforms} to a schema $S$ if 1) every node in $G$ has a
single node label in $\Gamma_S$ and every edge has a label in $\Sigma_S$, and 2)
for all $A,B\in\Gamma_S$ and $R\in\Sigma_S^\pm$, for every node with label $A$
the number of its $R$-successors with label $B$ is as specified by
$\delta_S(A,R,B)$. By $L(S)$ we denote the set of all \emph{finite} graphs that
conform to $S$.

\paragraph{Queries} We work with conjunctive two-way regular path queries
(C2RPQs) that have the form
\[
  q(\bar{x}) = \exists \bar{y}.  \varphi_1 (z_1,z_1')\land\ldots\land \varphi_k
  (z_k,z_k')\,,
\]
where $\bar{x}=\{z_1,z_1',\ldots,z_k,z_k'\}\minus\bar{y}$ and for every
$i\in\{1,\ldots,k\}$,  $z_i$ and $z_i'$ are variables and the formula $\varphi_i$ is a regular expression that
follows the grammar
\[
  \varphi \coloncolonequals \varnothing \mid \epsilon \mid A \mid R \mid
  \varphi\cdot\varphi \mid \varphi + \varphi \mid \varphi^*\,,
\]
where $A\in\Gamma$ matches nodes, $R\in\Sigma^\pm$ matches edges, $\epsilon$
matches empty paths, and $\varnothing$ matches no path.
The semantics of C2RPQs is defined in the standard fashion~\cite{CGLV00} and we
denote the set of \emph{answers} to $q(\bar{x})$ in $G$ by $[q(\bar{x})]^G$.

\begin{example}
  \label{ex:queries}
  Recall the schema $S_0$ in Figure~\ref{fig:medical-database-schema}. The
  following query selects vaccines together with the antigens they are designed
  to target or target through cross-reaction.
  \begin{equation*}
    q(x,y) = (
    \Vaccine\cdot
    \designTarget\cdot
    \crossReacting^*\!\cdot
    \Antigen) (x,y). \ \ \square
  \end{equation*}
\end{example}

\noindent

\noindent
Trivial atoms are of the form $\varnothing(x,x)$, $\epsilon(x,x)$, and $A(x,x)$, and in the sequel, we abuse notation and write them as unary atoms: $\varnothing(x)$, $\epsilon(x)$, and $A(x)$, respectively.
The \emph{multigraph} of a C2RPQ $q$ has variables of $q$ as nodes and an edge  from $x$ to $y$ for every non-trivial atom $\varphi(x,y)$. The subclass of
\emph{acyclic} C2RPQs consists of queries whose multigraph is acyclic i.e., it does not have a path consisting of distinct edges that visits the same node twice.
Note that acyclicity for C2RPQs needs to be more restrictive than the classical acyclicity of conjunctive queries based on Gaifman graphs. Indeed, the Gaifman graph of a C2RPQ $\varphi(x,y)\land \psi(x,y)$ is acyclic but its matches may form nontrivial cycles in the input graph.

A \emph{Boolean} C2RPQ $q$ has all its variables existentially quantified, and
it may have only a single answer, the empty tuple, in which case, we say that
$q$ is \emph{satisfied} in $G$ and write $G\models q$.  We also use
\emph{unions} of C2RPQs (abbreviated as UC2RPQs) represented as sets of C2RPQs
$Q(\bar{x})=\{q_1(\bar{x}),\ldots, q_k(\bar{x})\}$ and extend the notions of
answers, satisfaction, and acyclicity to UC2RPQs in the natural fashion. Given
two UC2RPQs $P(\bar{x})$ and $Q(\bar{x})$, and a schema $S$, we say that
\emph{$P(\bar{x})$ is contained in $Q(\bar{x})$ modulo $S$}, in symbols
$P(\bar{x})\subseteq_S Q(\bar{x})$, if $[P(\bar{x})]^G\subseteq [Q(\bar{x})]^G$
for every $G\in L(S)$.
 
\paragraph{Description logics}
We operate on properties of graphs formulated in the
\emph{description logic} \ALCIF (and its fragments)~\cite{DLBook}. In
description logics, elements of $\Gamma$ and $\Sigma$ are called \emph{concept
  names} and \emph{role names}, respectively. \ALCIF allows to build more
complex concepts with the following grammar:
\[
  C \coloncolonequals
  \bot
  \mid
  A
  \mid
  C \sqcap C
  \mid
  \lnot C
  \mid
  \exists R.C \mid
  \exists^{\leq 1} R. C\,, 
\]
where $A\in\Gamma$ and $R\in\Sigma^\pm$.  We also use additional operators that
are redundant but useful when defining fragments; for brevity we introduce them
as syntactic sugar: $\top\colonequals \lnot \bot$,
$C_1\sqcup C_2\colonequals \lnot(\lnot C_1 \sqcap \lnot C_2)$,
$\forall R.C\colonequals \lnot \exists R.\lnot C$,
$\nexists R. C\colonequals \lnot \exists R.C$. We extend the interpretation function
$\cdot^G$ to complex concepts as follows:
\begin{gather*} 
\bot^G = \emptyset\,, \quad
      (C_1\sqcap C_2)^G = C_1^G\cap C_2^G\,, \quad
      (\lnot C)^G = \dom(G) \minus C^G\,, \\
(\exists R. C)^G = \big\{u\in \dom(G) \mid \exists v .\ (u,v)\in R^G \land v \in C^G\big\}\,,\\
(\exists^{\leq 1} R. C)^G = \big\{u\in \dom(G) \mid \exists^{\leq 1}v.\ (u,v)\in R^G \land v \in C^G\big\}\,.
\end{gather*}
\noindent
Statements in description logics have the form of \emph{concept inclusions},
\[
  C \sqsubseteq D
\]
where $C$ and $D$ are concepts. A graph $G$ \emph{satisfies} $C \sqsubseteq D$,
in symbols $G \models C \sqsubseteq D$, if $C^G \subseteq D^G$. A set $\T$ of
concept inclusions is traditionally called a \emph{TBox} and we extend
satisfaction to TBoxes in the canonical fashion: $G \models \T$ if
$G \models C \sqsubseteq D$ for each $C \sqsubseteq D \in \T$.

In the \emph{Horn fragment} of \ALCIF, written \HornALCIF, we only allow concept
inclusions in the following normal forms:
\begin{align*}
  &K \sqsubseteq A\,,&
  &K \sqsubseteq \bot\,,&
  &K \sqsubseteq \forall R. K'\,,\\[-2pt]
  &K \sqsubseteq \exists R. K'\,,&
  &K \sqsubseteq \nexists R. K'\,,& 
  &K \sqsubseteq \exists^{\leq 1} R. K'\,,
\end{align*}
where $A \in \Gamma$, $R\in\Sigma^\pm$, and $K,K'$ are intersections of concept
names (intersection of the empty set of concepts is $\top$).  
If statements of the form $K \sqsubseteq A_1 \sqcup A_2 \sqcup \dots \sqcup A_n$ are allowed too, then we recover the full power of $\ALCIF$ (up to introducing auxiliary concept names).

Participation constraints of schemas can be expressed with simple \HornALCIF statements as
illustrated in following example.
\begin{example}  
  For instance, the assertion in $S_0$
  (Figure~\ref{fig:medical-database-schema}) that \Pathogen manifests at least
  one \Antigen is expressed with the statement
  $\Pathogen \sqsubseteq \exists \exhibits.\Antigen$. The assertion that an
  \textsl{Antigen} may be exhibited by an arbitrary number of \emph{Pathogens}
  needs no \HornALCIF statement. However, statements are needed for implicitly
  forbidden edges, e.g.,
  $\Vaccine \sqsubseteq \nexists \exhibits. \Antigen$. \qed
\end{example}

\section{Graph Transformations}
\label{sec:transformations}
We propose transformations of graphs defined with Datalog-like rules that use
acyclic C2RPQs in their bodies. To allow multiple copies of the same source node
we use node constructors. Formally, a \emph{$k$-ary node constructor} is a
function $f:\N^k\rightarrow\N$ and we denote the set of node constructors by
$\F$. To remove the concern of overlapping node constructors, and the need to
analyze their definitions, we assume that for every node label $A\in\Gamma$ we
have precisely one node constructor $f_A$, all node constructors are injective,
and their ranges are pairwise disjoint.

We introduce two kinds of \emph{graph transformation rules}: node rules and edge
rules. A \emph{node rule} has the form
\[
  A\big(f_A(\bar{x})\big) \leftarrow q(\bar{x})\,,
\]
where $A\in\Gamma$, $f_A\in\F$, and $q$ is an acyclic C2RPQ. An \emph{edge rule}
has the form
\[
  r\big(f(\bar{x}),f'(\bar{y})\big) \leftarrow q(\bar{x},\bar{y})\,,
\]
where $r\in\Sigma$, $f,f'\in\F$, and $q$ is an acyclic C2RPQ. Note that an equality between variables $z=z'$ can be expressed as $\epsilon(z,z')$, and consequently, we can assume that $\bar{x}$ and $\bar{y}$ are disjoint.

Now, a
\emph{graph transformation} $T$ is a finite set of graph transformation rules.
By $\Gamma_T$ and $\Sigma_T$ we denote the finite sets of node and edge labels,
respectively, used in the heads of the rules of $T$.

\begin{example}
  \label{ex:transformation-definition}
  Below we present rules defining the transformation $T_0$ of the medical
  database, described in Example~\ref{ex:medical-database}. We use 3 unary node
  constructors $f_A(x)$ for $\Antigen$ nodes, $f_P(x)$ for $\Pathogen$ nodes,
  and $f_V(x)$ for $\Vaccine$ nodes.
  \begin{align*}
    \Vaccine(f_V(x))
    & \leftarrow (\Vaccine)(x)\,,\\[-2pt]    
    \Antigen(f_A(x))
    & \leftarrow (\Antigen)(x)\,,\\[-2pt]
    \designTarget(f_V(x),f_A(y))
    & \leftarrow (\designTarget)(x,y)\,,\\[-2pt]
    \targets(f_V(x),f_A(y))
    & \leftarrow  (\designTarget\cdot\crossReacting^*)(x,y)\,,\\[-2pt]
    \Pathogen(f_P(x))
    & \leftarrow (\Pathogen)(x)\,,\\[-2pt]
    \exhibits(f_P(x),f_A(y))
    &\leftarrow (\exhibits)(x,y)\,. \tag*{\qed}
  \end{align*}  
\end{example}
\noindent
Now, given a graph $G$ and a graph transformation $T$ the \emph{result of applying
  $T$ to $G$} is a graph $T(G)$ such that (for $A\in\Gamma$ and $r\in\Sigma$)
\begin{gather*}
  A^{T(G)} = \big\{
    f_A(t) \bigm{|} 
    A\big(f_A(\bar{x})\big) \leftarrow q(\bar{x})\in T,\ t\in [q(\bar{x})]^G
    \big\}\,,\\[-2pt]
    r^{T(G)} = \big\{
    \big(f(t),f'(t')\big) \bigm{|}
    r\big(f(\bar{x}),f'(\bar{y})\big) \leftarrow q(\bar{x},\bar{y}) \in T,\\[-2pt]
    \hspace{14em}(t,t')\in [q(\bar{x},\bar{y})]^G\big\}\,.
\end{gather*}
We are interested in the following two classical static analysis tasks.
\begin{list}{}{}
\item[\bf Type checking] Given a transformation $T$, a source schema $S$, and a
  target schema $S'$ check whether for every $G$ that conforms to $S$ the output
  of transformation $T(G)$ conforms to $S'$.
\end{list}
\begin{list}{}{}
\item[\bf Equivalence] Given a source schema $S$ and two transformations $T_1$
  and $T_2$ check whether $T_1$ and $T_2$ agree on every graph that conforms to $S$. 
\end{list}
In settings where the target schema is not known, it might be useful to
construct one. Naturally, we wish to preclude a trivial solution that produces the universal schema that accepts all graphs over a given set of node and edge labels. Instead, we propose to construct a schema that offers the tightest fit to the set of output graphs. To define formally this requirement, we define schema containment in the classical fashion: a schema $S$ is contained in $S'$ if and only if $L(S)\subseteq L(S')$.
\begin{list}{}{}
\item[\bf Schema elicitation] Given a transformation $T$ and a source schema
  $S$, construct the containment-minimal target schema $S'$ such that
  $T(G)\in L(S')$ for every $G\in L(S)$.
\end{list}
We observe that $T(G)$ may have nodes with no label, which may preclude it from
satisfying any schema, and consequently, schema elicitation may also return
error.

We prove the main result by reducing the problems of interest to query
containment modulo schema (and vice versa), which we later show to be
\EXPTIME-complete. Although schema elicitation is not a decision problem, we show \EXPTIME-completeness of deciding if the result of schema elicitation is equivalent to a given schema. Should schema elicitation have lesser complexity, so would have the corresponding decision problem since schema equivalence is easily decided in polynomial time.
\begin{theorem} \label{thm:main}
  Type checking, schema elicitation, and equivalence of graph transformations
  are \EXPTIME-complete.
\end{theorem}
\noindent
We outline the main ideas of the proof by illustrating how a transformation $T$
can be analyzed with a toolbox of methods based on query containment modulo
source schema $S$. We formulate these methods with an entailment relation:
\[
  (T,S)\models K \sqsubseteq K' \iff \text{$T(G)\models K \sqsubseteq K'$ for every
  $G\in L(S)$.}
\]
W.l.o.g.\ we assume that every rule of transformation $T$ is \emph{trim} i.e., it uses
in its body a query $q(\bar{x})$ that is satisfiable modulo $S$, in symbols
$\exists \bar{x}.q(\bar{x})\not\subseteq_S \varnothing$; otherwise, $T$ can be
trimmed.

First, we group queries from rules of $T$ based on the labels of nodes and edges
they create. For $A,B\in\Gamma_T$ and $r\in\smash{\Sigma_T}$ we define
\begin{align*}
  & Q_{A}(\bar{x})=\big\{q(\bar{x}) \bigm{|}
    A\big(f_A(\bar{x})\big)\leftarrow q(\bar{x})\in T\big\}\,, \\[-2pt]
  & Q_{A,r,B}(\bar{x},\bar{y})=\big\{q(\bar{x},\bar{y}) \bigm{|}
    r\big(f_A(\bar{x}),f_B(\bar{y})\big)\leftarrow q(\bar{x},\bar{y})\in T \big\}\,,\\[-2pt]
  & Q_{A,r^-,B}(\bar{x},\bar{y})=\big\{q(\bar{y},\bar{x}) \bigm{|}
    r\big(f_B(\bar{y}),f_A(\bar{x})\big)\leftarrow q(\bar{y},\bar{x})\in T \big\}\,.
\end{align*}
In essence, $Q_{A}(\bar{x})$ identifies tuples over the input graph that yield a
node constructed with $f_A$ and with label $A$ while
$Q_{A,R,B}(\bar{x},\bar{y})$ identifies tuples that yield $R$-edges from a node
created with $f_A$ to a node created with $f_B$.

\begin{example}
  A couple of examples of above queries for the transformation $T_0$ in Example
  \ref{ex:transformation-definition} follow.
  \begin{align*}
    & Q_\Vaccine(x) = (\Vaccine)(x)\,,\\[-2pt]
    & Q_{\Vaccine,\targets,\Antigen}(x,y) = (\designTarget\cdot\crossReacting^*)(x,y)\,,\\[-2pt]
    & Q_{\Vaccine,\designTarget,\Antigen}(x,y) = (\designTarget) (x,y)\,.
      \tag*{\qed}
  \end{align*}
\end{example}
\noindent
Since an edge rule does not assign labels to nodes it creates, the result of a
transformation may be a graph with nodes without a label. Such a situation
precludes type checking from passing and prevents schema elicitation from
producing meaningful output. Consequently, we first verify that every node in
every output graph has exactly one label, in symbols
$(T,S)\models\top\sqsubseteq\bigsqcup\Gamma_T$, where
$\bigsqcup \{A_1,\ldots, A_k\}$ is a shorthand for $A_1\sqcup\ldots\sqcup
A_k$. We prove the following\cameraready{.}{ (Lemma~\ref{lemma:transformations-consistency}).}
\begin{align*}
  & (T,S)\models\top\sqsubseteq\textstyle\bigsqcup\Gamma_T
  \iff {}\\[-2pt]
    &\hspace*{3em}
    \exists\bar{y}.Q_{A,R,B}(\bar{x}, \bar{y}) \subseteq_S Q_A(\bar{x})
    \quad\text{for all $A,B\in\Gamma_T$ and  $R\in\Sigma_T^\pm$\,.}
\end{align*}
We point out that the restriction of one node constructor per node label ensures
that each node of the output has at most one label.
\begin{example}
  Take $T_0$ from Example~\ref{ex:transformation-definition} and $S_0$ in
  Figure~\ref{fig:medical-database-schema}. Verifying that
  $(T_0,S_0)\models \top\sqsubseteq \bigsqcup \Gamma_{T_0}$ requires a number of
  containment tests including the following two. 
  \begin{align*}
    & \exists y. (\designTarget\cdot\crossReacting^*)(x,y) \subseteq_{S_0}
      (\Vaccine)(x)\,,\\[-2pt]
    & \exists y. (\designTarget)(x,y) \subseteq_{S_0} 
      (\Vaccine)(x)\,. \tag*{\qed}
  \end{align*}
\end{example}

\noindent
Now, to perform type checking against a given target schema $S'$, we verify that
$\Gamma_T\subseteq\Gamma_{S'}$ and $\Sigma_T\subseteq\Sigma_{S'}$. Then, we take
the TBox $\T_{S'}$ of concept inclusions that expresses participation
constraints of the target schema $S'$ and we verify that
$(T,S)\models\T_{S'}$. Type checking succeeds if and only if all the above tests
succeed\cameraready{.}{ (Lemma~\ref{lemma:transformations-type-checking}).}

The TBox $\T_{S'}$ consists of statements from a small fragment $\L_0$ of \HornALCIF
which allows only statements of the forms
\begin{align*}
  &A \sqsubseteq \exists R. B\,, &
  &A \sqsubseteq \nexists R. B\,, &
  &A \sqsubseteq \exists^{\leq 1} R. B\,,
\end{align*}
where $A,B \in \Gamma$ and $R\in\Sigma^\pm$. The entailment of such statements is
also reduced to query containment\cameraready{:}{ (Lemma~\ref{lemma:transformations-entailment}):} 
\begin{align*}
  & (T,S)\models A \sqsubseteq \exists R.B\iff 
    Q_A(\bar{x}) \subseteq_S \exists \bar{y}. Q_{A,R,B}(\bar{x},\bar{y})\,,\\[-2pt]
  &(T,S)\models A \sqsubseteq \nexists R.B \iff 
  \exists \bar{y}.Q_A(\bar{x}) \!\land\! Q_{A,R,B}(\bar{x},\bar{y})
  \subseteq_S \textstyle\bigwedge_i\varnothing(x_i)\,,\\[-2pt]
  &(T,S)\models A \sqsubseteq \exists^{\leq 1} R.B\iff\\[-2pt]
  & \hspace{7.75ex}\exists \bar{x}. Q_A(\bar{x}) \!\land\! Q_{A,R,B}(\bar{x},\bar{y}) \!\land\!
    Q_{A,R,B}(\bar{x},\bar{z}) \subseteq_S
    \textstyle\bigwedge_i \epsilon(y_i,z_i)\,.
\end{align*}
\begin{example}
  Take the transformation $T_0$ and the schemas $S_0$ and $S_1$ in
  Figure~\ref{fig:medical-database-schema}. The schema $S_1$ requires every
  vaccine to target at least one antigen, in symbols
  $\Vaccine\sqsubseteq\exists\targets.\Antigen$. This statement is entailed by
  $T_0$ and $S_0$ if and only if the following holds
  \begin{equation*}
    (\Vaccine) (x) \subseteq_{S_0} \exists y. (\designTarget\cdot\crossReacting^*)(x,y)\,.
    \quad\square
  \end{equation*}
\end{example}

\noindent
For schema elicitation, we use a close correspondence between schemas and $\L_0$
TBoxes. It is sufficient to construct the TBox $\T$ containing all $\L_0$
statements that are entailed by $T$ and $S$; $\T$ corresponds to the
containment-minimal target schema\cameraready{.}{ (Lemma~\ref{lemma:transformations-schema-elicitation}).}

Finally, the equivalence of two transformations $T_1$ and $T_2$ is essentially
the equivalence (modulo $S$) of the respective queries $Q_A$ and $Q_{A,R,B}$ of
both transformations\cameraready{.}{ (Lemma~\ref{lemma:transforamtions-equivalence}).} Naturally,
query equivalence is reduced to query containment, as usual.

We have shown that type checking, schema elicitation, and equivalence of graph
transformations are Turing-reducible in polynomial time to testing containment
of UC2RPQs in acyclic UC2RPQs modulo schema. We also show polynomial-time
reductions of containment of 2RPQs modulo schema to all above problems of
interest\cameraready{.}{ (Lemma~\ref{lemma:static-analysis-hardness}).}  With that,
Theorem~\ref{thm:main} follows from Theorem~\ref{thm:containment}.

\section{Query Containment modulo Schema}
\label{sec:containment}

The aim of this section is to show the following result. 

\begin{theorem}
\label{thm:containment}
Containment of UC2RPQs in acyclic UC2RPQs modulo schema is \EXPTIME-complete. 
\end{theorem}

The lower bound can be derived from the \EXPTIME-hardness of unrestricted containment of 2RPQs (using only edge labels) modulo very simple TBoxes. 
The latter is obtained by reduction from another reasoning task (satisfiability of $\ALCI$ TBoxes) and relies on the inner workings of its hardness proof. 
For completeness, we provide a direct reduction from the acceptance problem for polynomial-space alternating Turing machines\cameraready{.}{ (Theorem~\ref{theorem:query-containment-hardness}).}
The remainder of this section
is devoted to the upper bound. We show it by reduction to unrestricted (finite
or infinite) satisfiability of C2RPQs modulo a $\HornALCIF$ TBox, which we discuss
in Section~\ref{sec:satisfiability}. The principal technique applied in the reduction is \emph{cycle reversing} \cite{CosmadakisKV90}.

Let $S$ be a schema, $P$ a UC2RPQ, and $Q$ an acyclic UC2RPQ. Without loss of generality we may assume that $P$ and $Q$ are Boolean\cameraready{.}{ (see Lemma~\ref{lem:boolean-containment}).} The key idea is to pass from finite to possibly infinite graphs, thus making canonical witnesses for non-containment easier to find. However, as Example~\ref{ex:no-finite-controllability} shows, we cannot pass freely from finite to possibly infinite graphs, as this may affect the answer.

\begin{example} \label{ex:no-finite-controllability}
  Consider the schema $S$ in Figure~\ref{fig:query-containment}.
  \begin{figure}
  \begin{tikzpicture}[
    >=stealth',
    empty/.style={minimum size=0pt,inner sep=0pt, outer sep=0pt},
    punkt/.style={circle,minimum size=0.125cm,draw=solarized-red,fill=solarized-red,inner sep=0pt, outer sep=0.125cm},
    kwadrat/.style={rectangle,minimum size=0.125cm,draw=solarized-blue,fill=solarized-blue,inner sep=0pt, outer sep=0.125cm},
    punkcik/.style={circle,minimum size=0.1cm,draw=solarized-red,fill=solarized-red,inner sep=0pt, outer sep=0.075cm},
    kwadracik/.style={rectangle,minimum size=0.1cm,draw=solarized-blue,fill=solarized-blue,inner sep=0pt, outer sep=0.05cm},
    kropka/.style={rectangle,minimum size=0.0125cm,draw,fill,inner sep=0pt, outer sep=0.025cm},
    ]
    
    \path[use as bounding box] (-1.1,0.75) rectangle (7.25,-1.4);

  \begin{scope}
  \node[left] at (-0.6,0.4) {$S$:};
  \node[punkt] (A) at (0,0)   (A) {}; \node[above=1pt] {\sl A};
  \draw (A) edge[->,loop left] node[left] {\sl s} node[pos=0.2,below] {\MAYBE} node[pos=0.8,above] {\PLUS} (A);
  \draw (A) edge[->,loop right] node[right] {\sl r} node[pos=0.2,above] {\MANY} node[pos=0.8,below] {\MANY} (A);
  \end{scope}

  \begin{scope}[yshift=-1.25cm]
  \node[left] at (-0.6,0.4) {$S^*$:};
  \node[kwadrat] (A) at (0,0)   (A) {}; \node[above=1pt] {\sl A};
  \draw (A) edge[->,loop left] node[left] {\sl s} node[pos=0.2,below] {\ONE} node[pos=0.8,above] {\ONE} (A);
  \draw (A) edge[->,loop right] node[right] {\sl r} node[pos=0.2,above] {\MANY} node[pos=0.8,below] {\MANY} (A);
  \end{scope}

  \begin{scope}[xshift=2.25cm]
  \node[left] at (-0.5,0.4) {$G_0$:};
  
  \node[punkcik] (x0) at (90:0.5) {};
  \node[punkcik] (x1) at (162:0.5) {};
  \node[punkcik] (x2) at (234:0.5) {};
  \node[punkcik] (x3) at (306:0.5) {};
  \node[punkcik] (x4) at (378:0.5) {};
  
  \begin{scope}[bend angle=20]
  \draw (x0) edge[->,bend right] node[above,sloped,pos=0.4] {\sl s} (x1);
  \draw (x1) edge[->,bend right] node[below,sloped,pos=0.35] {\sl s} (x2);
  \draw (x2) edge[->,bend right] node[below,sloped,pos=0.3] {\sl s} (x3);
  \draw (x3) edge[->,bend right] node[below,sloped,pos=0.35] {\sl s} (x4);
  \draw (x4) edge[->,bend right] node[above,sloped,pos=0.4] {\sl s} (x0);
  \end{scope}

  \draw (x4) edge[->,loop right] node {\sl r} (x4);  
  \end{scope}

  \begin{scope}[xshift=4.7cm]
  \node[left] at (-.25,0.4) {$G_\infty$:};
  
  \node[punkcik] (x0) at (0,0) {};
  \node[punkcik] (x1) at (0.75,0) {};
  \node[punkcik] (x2) at (1.5,0) {};
  \node (x3) at (2.25,0) {\rlap{\ldots}};
  \draw (x0) edge[->] node[above] {\sl s} (x1);
  \draw (x1) edge[->] node[above] {\sl s} (x2);
  \draw (x2) edge[->] node[above] {\sl s} (x3);
  \draw (x0) edge[->,loop left] node {\sl r} (x0);  

  \node[punkcik] (y1) at (0.5,0.5) {};
  \node[punkcik] (y2) at (1.25,0.5) {};
  \node (y3) at (2,0.5) {\rlap{\ldots}};
  \draw[bend angle=15] (x0) edge[->,bend left] node[above,sloped,pos=0.4] {\sl s} (y1);
  \draw (y1) edge[->] node[above] {\sl s} (y2);
  \draw (y2) edge[->] node[above] {\sl s} (y3);

  \node[punkcik] (z2) at (1.25,-0.5) {};
  \node (z3) at (2,-0.5) {\rlap{\ldots}};
  \draw[bend angle=15] (x1) edge[->,bend right] node[below,sloped,pos=0.4] {\sl s} (z2);
  \draw (z2) edge[->] node[above] {\sl s} (z3);
  \end{scope}

  \begin{scope}[xshift=4.7cm,yshift=-1.25cm]
  \node[left] at (-0.5,0.5) {$G_\infty^*$:};
  
  \node[] (y3) at (-2.25,0) {\llap{\ldots}};
  \node[kwadracik] (y2) at (-1.5,0) {};
  \node[kwadracik] (y1) at (-0.75,0) {};
  \node[kwadracik] (x0) at (0,0) {};
  \node[kwadracik] (x1) at (0.75,0) {};
  \node[kwadracik] (x2) at (1.5,0) {};
  \node (x3) at (2.25,0) {\rlap{\ldots}};
  \draw (x0) edge[->] node[above] {\sl s} (x1);
  \draw (x1) edge[->] node[above] {\sl s} (x2);
  \draw (x2) edge[->] node[above] {\sl s} (x3);
  \draw (x0) edge[->,loop above] node {\sl r} (x0);  
  \draw (y3) edge[->] node[above] {\sl s} (y2);
  \draw (y2) edge[->] node[above] {\sl s} (y1);
  \draw (y1) edge[->] node[above] {\sl s} (x0);
  \end{scope}
\end{tikzpicture}
    \caption{\label{fig:query-containment}Query containment over finite and infinite graphs.}
  \end{figure}
Observe that $S$ allows infinite graphs that are essentially infinite trees when
restricted to $s$-edges, e.g. $G_\infty$ in Figure~\ref{fig:query-containment}. In
fact, every infinite graph satisfying $S$ that is connected when restricted to
$s$-edges is an infinite tree. On the other hand, every non-empty finite graph that conforms
to $S$ is a collection of disjoint cycles when restricted to
$s$-edges, e.g., $G_0$ in Figure~\ref{fig:query-containment}. Clearly, the topology
of finite and infinite graphs defined by the schema differs drastically.

Now, take the queries $P=\exists x. r(x,x)$,
$Q = \exists x,y. (r \cdot s^+ \cdot r)(x,y)$, and observe that
$P \subseteq_S Q$. However, the containment does not hold over infinite graphs:
$P$ is satisfied by $G_\infty$ while $Q$ is not. \qed
\end{example}

The reason why we cannot pass directly to infinite models is that finite graphs conforming to schema $S$ may display certain additional common properties, detectable by queries, but not shared by infinite graphs conforming to $S$. 
The cycle \emph{reversing technique} \cite{CosmadakisKV90} captures these properties in $S^*$ such that 
\[P \subseteq_S Q \quad \iff \quad  P \subseteq^\infty_{S^*} Q \]
where by $\subseteq^\infty_{S^*}$ we mean  containment over possibly infinite graphs conforming to $S^*$. However, as the following example shows, we cannot obtain $S^*$ by analysing $S$ alone.

\begin{example} 
\label{ex:schema-completion}
In Example~\ref{ex:no-finite-controllability} we saw that in a finite graph
conforming to $S$, each node has exactly one incoming and one outgoing
$s$-edge. We can use this observation to tighten the original schema $S$ to the
schema $S^*$ (Figure~\ref{fig:query-containment}). Alas, we still have
$P\not\subseteq_{S^*}^\infty Q$ because there is an infinite graph $G_\infty^*$
that satisfies $P$ but not $Q$. \qed
\end{example}

Instead, we first reduce containment modulo schema to finite satisfiability, fusing the schema $S$ and the query $Q$ into a single \HornALCIF TBox, and then pass from finite to unrestricted satisfiability by applying cycle reversing to the resulting TBox.
We follow closely the approach of \citeauthor{GarciaLS14}~\cite{GarciaLS14}, relying crucially on some of their results. 

Let $\T$ be a \HornALCIF TBox. A \emph{finmod cycle} is a sequence \[K_1, R_1, K_2, R_2, \dots, K_{n-1}, R_{n-1}, K_n\] where $R_1, \dots, R_{n-1} \in \Sigma^\pm$ and $K_1, \dots , K_n$ are conjunctions of concept names such that $K_n = K_1$ and
\[\T \models K_i \sqsubseteq \exists R_i.K_{i+1} \quad\text{and}\quad \T \models K_{i+1} \sqsubseteq \exists^{\leq 1} R_i^-.K_i\]
for $1 \leq i < n$.
By \emph{reversing} the finmod cycle  we mean extending $\T$ with concept inclusions
\[K_{i+1} \sqsubseteq \exists R^-_i.K_i \quad\text{and}\quad K_i \sqsubseteq \exists^{\leq 1}  R_i.K_{i+1}\] for $1 \leq i < n$.
The \emph{completion} $\T^*$ of a TBox $\T$ is obtained from $\T$ by exhaustively reversing finmod cycles. The following key result is stated in \cite{GarciaLS14} in terms of sets of ground facts (so-called ABoxes) rather than subgraphs, but our formulation is equivalent. 

\begin{theorem}[\citeauthor{GarciaLS14}, \citeyear{GarciaLS14}]
\label{thm:completion}
A \HornALCIF TBox $\T$ has a finite model containing a finite subgraph $H$  iff  its completion $\T^*$ has a possibly infinite model containing $H$.
\end{theorem}

\begin{example}
  Schema $S$ from Example~\ref{ex:no-finite-controllability} is equivalent to
  TBox $\T_S$ that consists of
  \[
    \top \sqsubseteq A\,,\quad
    A \sqsubseteq \exists s.A\,,\quad
    A \sqsubseteq \exists^{\leq 1} s^-.A\,.
  \]
  Non-satisfaction of $Q$ is captured by TBox $\T_{\lnot Q}$ that consists of 
  \[
    \top \sqsubseteq \forall r. B_r\,,\quad
    B_r \sqsubseteq \forall s. B_{r\cdot s^+}\,,\quad
    B_{r\cdot s^+} \sqsubseteq \forall s. B_{r\cdot s^+}\,,\quad
    \hfil B_{r\cdot s^+} \sqsubseteq \forall r. \bot\,.
  \]
  Let $\T=\T_S\cup\T_{\lnot Q}$ and observe that
  $A \sqcap B_{r\cdot s^+} \,,\, s \,,\, A \sqcap B_{r\cdot s^+}$ is a finmod
  cycle in $\T$. By reversing it, we obtain
  \[
    A \sqcap B_{r\cdot s^+} \sqsubseteq \exists s^-. A \sqcap B_{r\cdot s^+}
    \quad \text{and} \quad
    A \sqcap B_{r\cdot s^+} \sqsubseteq \exists^{\leq 1} s. A \sqcap B_{r\cdot s^+}\,.
  \]
  Now, suppose that there exists a (finite or infinite) model $G$ of $\T^*$ that
  satisfies $P$ (see Figure~\ref{fig:cycle-reversal}).
  \begin{figure}
  \begin{tikzpicture}[
    >=stealth',
    empty/.style={minimum size=0pt,inner sep=0pt, outer sep=0pt},
    punkt/.style={circle,minimum size=0.125cm,draw=solarized-red,fill=solarized-red,inner sep=0pt, outer sep=0.125cm},
    kwadrat/.style={rectangle,minimum size=0.125cm,draw=solarized-blue,fill=solarized-blue,inner sep=0pt, outer sep=0.125cm},
    punkcik/.style={circle,minimum size=0.1cm,draw=solarized-red,fill=solarized-red,inner sep=0pt, outer sep=0.075cm},
    kwadracik/.style={rectangle,minimum size=0.1cm,draw=solarized-blue,fill=solarized-blue,inner sep=0pt, outer sep=0.05cm},
    kropka/.style={rectangle,minimum size=0.0125cm,draw,fill,inner sep=0pt, outer sep=0.025cm},
    ]

    \path[use as bounding box] (-1.1,0.35) rectangle (6.5,-0.35);

  \begin{scope}
  \node[left] at (-0.85,0) {$S$:};
  \node[punkt] (A) at (0,0)   (A) {}; \node[above=1pt] {\sl A};
  \draw (A) edge[->,loop left] node[left] {\sl s} node[pos=0.2,below] {\MAYBE} node[pos=0.8,above] {\PLUS} (A);
  \draw (A) edge[->,loop right] node[right] {\sl r} node[pos=0.2,above] {\MANY} node[pos=0.8,below] {\MANY} (A);
  \end{scope}
  \begin{scope}[xshift=2.75cm]
  \node[left] at (-0.85,0) {$G$:};
  
  \node[punkcik] (x0) at (0,0) {};
  \node[punkcik] (x1) at (1.25,0) {};
  \node[punkcik] (x2) at (2.5,0) {};
  \node (x3) at (3.5,0) {\rlap{\ldots}};

  \node[above=-0.1cm of x0] {\small $u$};
  \node[above=-0.1cm of x1] {\small $u\mathrlap{'}$};
  \node[above=-0.1cm of x2] {\small $u\mathrlap{''}$};

  \node[below=-0.01cm of x0] {\small $A\sqcap B\mathrlap{{}_{r}}$};
  \node[below=-0.01cm of x1] {\small $A\sqcap B\mathrlap{{}_{r\cdot s^+}}$};
  \node[below=-0.01cm of x2] {\small $A\sqcap B\mathrlap{{}_{r\cdot s^+}}$};
  
  \draw (x0) edge[->] node[above] {\sl s} (x1);
  \draw (x1) edge[<-] node[above] {\sl s} (x2);
  \draw (x2) edge[<-] node[above] {\sl s} (x3);
  \draw (x0) edge[->,loop left] node {\sl r} (x0);

  \draw[-,double,very thin] (x0) .. controls +(45:0.85) and +(135:0.85) .. (x2);

  \end{scope}
\end{tikzpicture}

    \caption{\label{fig:cycle-reversal}Cycle reversal argument.}
  \end{figure}
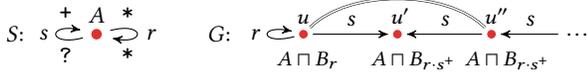
  $G$ must have a node $u$ with $(u,u)\in r^G$. It follows already from $\T$
  that $u\in (A\sqcap B_r)^G$ and that $u$ has an $s$-successor
  $u'\in (A\sqcap B_{r\cdot s^+})^G$. The statement 
  $A \sqcap B_{r\cdot s^+} \sqsubseteq \exists s^-. A \sqcap B_{r\cdot s^+}$ in
  $\T^*$ implies that $u'$ has an $s^-$-successor
  $u''\in (A\sqcap B_{r\cdot s^+})^G$. As each node has at most one incoming
  $s$-edge, $u = u''$ and $u\in (B_{r\cdot s^+})^G$. But $u$ has an outgoing
  $r$-edge, which contradicts the last concept inclusion in $\T_{\lnot Q}$.
  Thus, $P$ is not satisfied in $\T^*$. \qed
\end{example}

We are now ready to reduce containment modulo schema to unrestricted satisfiability modulo \HornALCIF TBox. Note that the guarantees on the resulting TBox in the statement below are sufficient to conclude Theorem~\ref{thm:containment} using Theorem~\ref{thm:satisfiability}.

\begin{theorem}
\label{thm:reduction}
Given a UC2RPQ $P$, an acyclic UC2RPQ $Q$, and a schema $S$, one can compute in \EXPTIME a UC2RPQ $\widehat P$ of polynomial size and a \HornALCIF TBox $\T$ using linearly many additional concept names  and polynomially many at-most constraints, 
such that $P\subseteq_S Q$ if and only if $\widehat P$ is (unrestrictedly) unsatisfiable modulo $\T$. 
\end{theorem}

Let us sketch the proof. Let $\T_S$ be the \HornALCIF TBox corresponding to $S$. Note that apart from the explicit restrictions captured in $\T_S$ the schema $S$ also ensures that only graphs with \emph{exactly} one label per node are considered. To ensure \emph{at most} one label from $\Gamma_S$ per node, we use the TBox $\widehat \T_S =\T_S \cup \{A\sqcap B \sqsubseteq \bot \mid A,B\in \Gamma_S, A\neq B\}$. The concept inclusion $\top \sqsubseteq \bigsqcup \Gamma_S$,  expressing that each node has \emph{at least} one label from $\Gamma_S$, is not Horn and cannot be used. Instead, we modify the query $P$. Assuming $\Gamma_S = \{A_1, A_2, \dots, A_n\}$, we include $(A_1+A_2+\dots+A_n)$ before and after each edge label used in an atom of $P$. Additionally, to ensure that $P$ uses only labels allowed by $S$, we substitute in $P$ each label not in $\Gamma_S \cup \Sigma_S^\pm$ by $\varnothing$. Letting $\widehat P$ be the resulting query, we have \[P \subseteq_S Q \quad \iff \quad  \widehat P \subseteq_{\widehat\T_S} Q\] \cameraready{}{ (see Lemma~\ref{lem:labels}).}
Because $Q$ is acyclic, by adapting the rolling-up technique~\cite{HorrocksT00} one can compute in \PTIME a \HornALCIF TBox $\T_{\lnot Q}$ over an extended set of concept names $\Gamma_S\cup\Gamma_{Q}$ such that 
\[\widehat P \subseteq_{\widehat \T_S} Q \;\; \iff \;\; \widehat P \text{ is finitely unsatisfiable modulo } \widehat\T_S \cup\T_{\lnot Q}\,. \]
\cameraready{}{ (see Lemma~\ref{lem:rollingup}).} Since $\widehat \T_S \cup \T_{\lnot Q}$ is a $\HornALCIF$ TBox, we can consider its completion  $\big(\widehat\T_S \cup \T_{\lnot Q}\big)^*$. As UC2RPQs are witnessed by finite subgraphs whenever they are satisfied, we can infer from Theorem~\ref{thm:completion} that  $\widehat P$ is finitely satisfiable modulo $\widehat \T_S \cup\T_{\lnot Q}$ iff $\widehat P$ is satisfiable modulo  $\big(\widehat\T_S\cup \T_{\lnot Q}\big)^*$\cameraready{.}{ (see Lemma~\ref{lem:finmod}).}

It remains to compute the completion. Reversing cycles does not introduce new concept names, but it may generate exponentially many concept inclusions. Identifying a finmod cycle involves deciding unrestricted entailment of \HornALCIF concept inclusions, which is decidable in \EXPTIME \cite{GiacomoL96}. However, since the input TBox might grow to an exponential size as more and more cycles are reversed, it is unlikely that the completion can be computed in \EXPTIME for every \HornALCIF TBox. Our key insight is that $\widehat\T_S\cup \T_{\lnot Q}$ enjoys a particular property, invariant under reversing cycles, that keeps the complexity under control. 

A concept inclusion (CI) of the form $K \sqsubseteq \exists R. K'$ or $K \sqsubseteq \exists^{\leq 1} R. K'$ is \emph{relevant} for a TBox $\T$ if the triple $(K, R, K')$ is satisfiable modulo $\T$; that is, some model $G$ of $\T$ contains nodes $u$ and $u'$ such that $u\in K^G$, $(u,u') \in R^G$, and $u'\in (K')^G$. 
We say that $\T$ is \emph{$S$-driven} if for each relevant CI in $\T$ of the form $K \sqsubseteq \exists R. K'$ (resp. $K \sqsubseteq \exists^{\leq 1} R. K'$),
$\T$ contains  $A\sqsubseteq \exists R. A'$ (resp. $A\sqsubseteq \exists^{\leq 1} R. A'$) for some $A,A' \in \Gamma_S$ such that $A \in K$, $A'\in K'$; here and later we blur the distinction between conjunctions of concept names and sets of labels.  Note that $\widehat\T_S\cup \T_{\lnot Q}$ is trivially $S$-driven, as all its  existential and at-most constraints are of the form $A\sqsubseteq \exists R. A'$ or $A\sqsubseteq \exists^{\leq 1} R. A'$.

\begin{lemma}
Every $S$-driven TBox $\T$ can be simplified in polynomial time so that it contains at most  $|\Sigma^\pm_S|\cdot|\Gamma_S|^2$ at-most constraints. 
\end{lemma}

From our results in  Section~\ref{sec:satisfiability} it follows that unrestricted  entailment for a \HornALCIF TBox $\T$ with $k$ concept names and $\ell$ at-most constraints can be solved in time $O\big(\mathrm{poly}(|\T|) \cdot 2^{\mathrm{poly}(k,\ell)}\big)$ \cameraready{.}{ (Corollary~\ref{cor:entailment}).}
Hence, it would suffice to show that by reversing a finmod cycle in an $S$-driven TBox, we obtain another $S$-driven TBox.
In fact, we prove something weaker, but sufficient to compute the completion in \EXPTIME, and conclude that it is $S$-driven.  

Let $K_1, R_1,  \dots, K_{n-1}, R_{n-1}, K_n$ be a finmod cycle in an $S$-driven \HornALCIF TBox $\T$.
Reversing it will extend $\T$ with CIs \[K_{i+1} \sqsubseteq \exists R^-_i.K_i \quad\text{and}\quad K_i \sqsubseteq \exists^{\leq 1}  R_i.K_{i+1}\] for $1 \leq i < n$. 
If all triples $(K_i,R_i,K_{i+1})$ are unsatisfiable wrt $\T$, then all CIs to be added are irrelevant for $\T$ and we are done.
Suppose that some  $(K_i,R_i,K_{i+1})$ is satisfiable. Then, in the model for $(K_i,R_i,K_{i+1})$ we can trace the finmod cycle forward, witnessing each triple. Hence, the whole cycle is satisfiable (all its triples are). Then, we can show that there are unique $A_1, A_2, \dots, A_n \in \Gamma_S$ such that $A_i \in K_i$ for all $i\leq n$, and
$A_1, R_1, \dots, A_{n-1}, R_{n-1}, A_n$
is a finmod cycle in $\T$\cameraready{.}{ (Lemma~\ref{lem:schema-driven}).}
By reversing it, we can add to $\T$ CIs \[A_{i+1} \sqsubseteq \exists R^-_i.A_{i} \quad \text{and} \quad A_{i} \sqsubseteq \exists^{\leq 1}  R_i.A_{i+1}\] 
for $1 \leq i < n$, which makes the resulting extension $S$-driven.

Based on the obtained invariant we can compute the completion $\big(\widehat\T_S\cup \T_{\lnot Q}\big)^*$ in \EXPTIME\cameraready{.}{ (Lemma~\ref{lem:completion}).} By reducing  $\big(\widehat\T_S\cup \T_{\lnot Q}\big)^*$ as described above, we obtain the desired TBox $\T$, thus completing the proof of Theorem~\ref{thm:reduction}.

\section{Satisfiability modulo TBox}
\label{sec:satisfiability}

The last missing piece is to solve the  unrestricted satisfiability of C2RPQs modulo \HornALCIF. \citeauthor{CalvaneseOS11} show that the problem is in \EXPTIME not only for \HornALCIF, but even for \ALCIF extended with additional features \cite{CalvaneseOS11}. This result is not directly applicable, because our reduction produces a TBox of exponential size. The following theorem gives the more precise complexity bounds that we need.

\begin{theorem}
\label{thm:satisfiability}
Unrestricted satisfiability of a C2RPQ $p$ modulo an \ALCIF TBox $\T$ using $k$ concept names and $\ell$ at-most constraints can be decided in time $O\big(\mathrm{poly}(|\T|)\cdot 2^{\mathrm{poly}\left(|p|,k,\ell\right)}\big)$. \end{theorem}

\citeauthor{CalvaneseOS11} solve the problem by first showing a simple
model property and then providing an algorithm testing existence of simple models. We rely on the same simple model property, but design a new algorithm with the desired complexity bounds. Yet, rather than diving into the details of the algorithm, we devote most of this section to the simple model property. We do it to show a connection to an elegant graph-theoretical notion that helps to simplify the reasoning considerably, at least for \ALCIF.
We begin by illustrating how simple models are obtained for queries satisfiable modulo schemas (rather than arbitrary TBoxes).
\begin{example}
  Take the schema $S$ in Figure~\ref{fig:sparse-witness} (its two types are
  represented with a blue square and a red circle),
  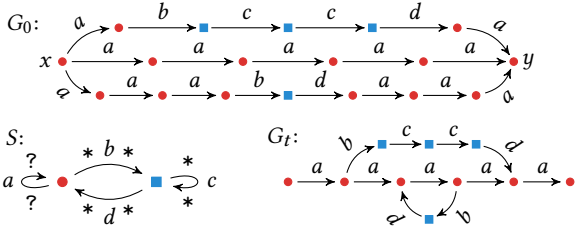
\begin{figure}
\begin{tikzpicture}[
  >=stealth',
  empty/.style={minimum size=0pt,inner sep=0pt, outer sep=0pt},
  punkt/.style={circle,minimum size=0.125cm,draw=solarized-red,fill=solarized-red,inner sep=0pt, outer sep=0.125cm},
  kwadrat/.style={rectangle,minimum size=0.125cm,draw=solarized-blue,fill=solarized-blue,inner sep=0pt, outer sep=0.125cm},
  punkcik/.style={circle,minimum size=0.1cm,draw=solarized-red,fill=solarized-red,inner sep=0pt, outer sep=0.075cm},
  kwadracik/.style={rectangle,minimum size=0.1cm,draw=solarized-blue,fill=solarized-blue,inner sep=0pt, outer sep=0.05cm},
  kropka/.style={rectangle,minimum size=0.0125cm,draw,fill,inner sep=0pt, outer sep=0.025cm},
  ]

  \path[use as bounding box] (-1,2.3) rectangle (7.25,-0.45);

  \begin{scope}
    \node[left] at (-0.4,0.6) {$S$:};
    \node[punkt] (A) at (0,0)   (A) {}; 
    \node[kwadrat] (B) at  (1.25,0) {}; 
    \draw (A) edge[->,loop left] node[left] {\sl a} node[pos=0.2,below] {\MAYBE} node[pos=0.8,above] {\MAYBE} (A);
    \draw (B) edge[->,loop right] node[right] {\sl c} node[pos=0.2,above] {\MANY} node[pos=0.8,below] {\MANY} (B);
    \draw (A) edge[->,bend left] node[above] {\sl b} node[pos=0.2,above] {\MANY} node[pos=0.8,above] {\MANY} (B);
    \draw (B) edge[->,bend left] node[below] {\sl d} node[pos=0.2,below] {\MANY} node[pos=0.8,below] {\MANY} (A);
\end{scope}

\begin{scope}[xshift=3cm]
  \node[left] at (0.25,0.6) {$G_t$:};
  \node[punkcik] (A0) at (0,0)    {} ;
  \node[punkcik] (A1) at (0.75,0) {} ;
  \node[punkcik] (A2) at (1.5,0)  {} ;
  \node[punkcik] (A3) at (2.25,0) {} ;
  \node[punkcik] (A4) at (3,0)    {} ;
  \node[punkcik] (A5) at (3.75,0) {} ;
  \node[kwadracik] (B2) at (1.25,0.5) {};
  \node[kwadracik] (B3) at (1.875,0.5) {};
  \node[kwadracik] (B4) at (2.5,0.5) {};

  \node[kwadracik] (B5) at (1.875,-0.5) {};
  
  \draw (A0) edge[->] node[above] {\sl a} (A1);
  \draw (A1) edge[->] node[above] {\sl a} (A2);
  \draw (A2) edge[->] node[above] {\sl a} (A3);
  \draw (A3) edge[->] node[above] {\sl a} (A4);
  \draw (A4) edge[->] node[above] {\sl a} (A5);

  \draw (A1) edge[->,bend left] node[above,sloped] {\sl b} (B2);
  \draw (B2) edge[->] node[above] {\sl c} (B3);
  \draw (B3) edge[->] node[above] {\sl c} (B4);
  \draw (B4) edge[->,bend left] node[above,sloped] {\sl d} (A4);

  \draw (A3) edge[->,bend left] node[below,sloped] {\sl b} (B5);
  \draw (B5) edge[->,bend left] node[below,sloped] {\sl d} (A2);
\end{scope}
\begin{scope}[yshift=1.6cm,xshift=0cm]
  \node[left] at (-0.2,0.5) {$G_0$:};
  
  \node[punkcik] (x) at (0,0) {};
  \node[punkcik,outer sep=1.5pt] (y) at (6,0) {};
  \node[left=-0.1cm of x] {$x$};
  \node[right=-0.1cm of y] {$y$};

  \node[punkcik] (a1) at (1.2,0) {};
  \node[punkcik] (a2) at (2.4,0) {};
  \node[punkcik] (a3) at (3.6,0) {};
  \node[punkcik] (a4) at (4.8,0) {};
  \draw (x) edge[->] node[above] {\sl a} (a1);
  \draw (a1) edge[->] node[above] {\sl a} (a2);
  \draw (a2) edge[->] node[above] {\sl a} (a3);
  \draw (a3) edge[->] node[above] {\sl a} (a4);
  \draw (a4) edge[->] node[above] {\sl a} (y);
  
  \node[punkcik]   (b1) at (0.75,0.45) {};
  \node[kwadracik] (b2) at (1.88,0.45) {};
  \node[kwadracik] (b3) at (3.00,0.45) {};
  \node[kwadracik] (b4) at (4.12,0.45) {};
  \node[punkcik]   (b5) at (5.25,0.45) {};

  \draw[bend angle=27] (x) edge[->,bend left] node[above,sloped] {\sl a} (b1);
  \draw (b1) edge[->] node[above] {\sl b} (b2);
  \draw (b2) edge[->] node[above] {\sl c} (b3);
  \draw (b3) edge[->] node[above] {\sl c} (b4);
  \draw (b4) edge[->] node[above] {\sl d} (b5);
  \draw[bend angle=27] (b5) edge[->, bend left] node[above,sloped] {\sl a} (y);

  \node[punkcik]   (c1) at (0.5,-0.45) {};
  \node[punkcik]   (c2) at (1.33,-0.45) {};
  \node[punkcik]   (c3) at (2.17,-0.45) {};
  \node[kwadracik] (c4) at (3,-0.45) {};
  \node[punkcik]   (c5) at (3.86,-0.45) {};
  \node[punkcik]   (c6) at (4.67,-0.45) {};
  \node[punkcik]   (c7) at (5.5,-0.45) {};
  
  \draw[bend angle=27] (x) edge[->,bend right] node[below,sloped,pos=0.4] {\sl a} (c1);
  \draw (c1) edge[->] node[above] {\sl a} (c2);
  \draw (c2) edge[->] node[above] {\sl a} (c3);
  \draw (c3) edge[->] node[above] {\sl b} (c4);
  \draw (c4) edge[->] node[above] {\sl d} (c5);
  \draw (c5) edge[->] node[above] {\sl a} (c6);
  \draw (c6) edge[->] node[above] {\sl a} (c7);
  \draw[bend angle=27] (c7) edge[->,bend right] node[below,sloped,pos=0.4] {\sl a} (y);
\end{scope}
\end{tikzpicture}
    \caption{\label{fig:sparse-witness} Simple witness for satisfiability. }
  \end{figure}
  and consider the following satisfiable (cyclic) query
  \[
    p(x,y) = 
    (a\cdot b\cdot c^+\cdot d\cdot a)(x,y) \land
    (a^*)(x,y) \land
    (a^* \cdot b\cdot d\cdot a^*)(x,y)\,.
  \]
  Since $p$ is satisfiable modulo $S$, we take any graph conforming to $S$ where
  $p$ is satisfied, and we choose any 3 paths witnessing each of the regular
  expressions of $p$. We construct the initial graph $G_0$ consisting of the 3
  paths joined at their ends: it might look like the one in Figure~\ref{fig:sparse-witness}. We observe that $S$ requires every red circle node
  to have at most one outgoing $a$-edge and at most one incoming $a$-edge (to
  and from a red circle node). The initial graph $G_0$ violates this requirement
  and to enforce it we exhaustively merge offending nodes. The final graph
  $G_t$ is a simple model of $p$ modulo $S$. \qed
\end{example}

We formalise simple models using a graph-theoretic notion of sparsity proposed by  \citeauthor{LeeS08} \cite{LeeS08}. We say that a connected graph $G$ with $n$ nodes and $m$ edges is \emph{$c$-sparse} if $m \leq n + c$. (In  \citeauthor{LeeS08}'s terminology this corresponds to $(1,-c)$-sparsity.)  Being $c$-sparse is preserved under adding and removing nodes of degree 1. By exhaustively removing nodes of degree 1 from a $c$-sparse graph $G$ we arrive at single node or a connected $c$-sparse graph $H$ in which all nodes have degree at least 2. Assuming $c\geq 1$, it is not hard to see that such a graph consists of at most $k=2c$ distinguished nodes connected by at most $l=3c$ simple paths disjoint modulo endpoints\cameraready{.}{ (see  Lemma~\ref{lem:sparse}).} We call such a graph a \emph{$(k,l)$-skeleton}, and we refer to the graph $H$ above as the \emph{skeleton of} $G$. 
Thus, a $c$-sparse graph consists of a $(2c,3c)$-skeleton and a number of attached trees; by attaching a tree to a graph we mean taking their disjoint union and adding a single edge between the root of the tree and some node of the graph. 

For the purpose of the simple model property we need to lift the notion of $c$-sparsity to infinite graphs. 
We call a (possibly infinite) graph \emph{$c$-sparse} if it consists of a finite connected $c$-sparse graph with finitely many finitely branching trees attached.

\begin{theorem}
\label{thm:sparsification}
A connected C2RPQ $p$ is satisfiable  in a possibly infinite model of an \ALCIF TBox $\T$ iff $p$ is satisfiable in a possibly infinite $|p|$-sparse model of $\T$.  
\end{theorem}

\begin{proof}
Let $c$ be the difference between the number of atoms and the number of variables of $p$. Because $p$ is connected, $c \geq -1$. By definition, $p$ understood as a graph with variables as nodes and atoms as edges is $c$-sparse.

We write $H \to H'$ to indicate that there is a \emph{homomorphism} from graph $H$ to graph $H'$; that is, a function $h$ mapping nodes of $H$ to nodes of $H'$ that preserves node labels and the existence of labelled edges between pairs of nodes. Let $G$ be a (possibly infinite) model of $p$ and $\T$. We construct a sequence of finite connected $c$-sparse graphs of strictly decreasing size
\[G_0 \to G_1 \to \dots \to G_t \to G\] such that $G_0 \models p$ and the homomorphism from $G_t$ to $G$ is injective over $R$-successors of every node, for each $R$.

To construct $G_0$ let us fix a match of $p$ in $G$ together with a (finite) witnessing path for each atom of $p$. We construct $G_0$ as follows. For each variable $x$ of $p$ we include a node $v_x$ whose set of labels is identical to that of the image of $x$ in $G$ under the fixed match. Next, for each atom of $p$ that connects variables $x$ and $y$ we add a simple path connecting $x$ and $y$ such that the sequence of edge labels and sets of node labels read off of this path is identical to that of the witnessing path of this atom in $G$. This graph can be seen as a specialization of $p$ where each regular expression is replaced by a single concrete word, except that we include full sets of labels of nodes, as they are encountered in the witnessing path in $G$. It follows immediately that $G_0 \models p$ and that $G_0 \to G$. To see that $G_0$ is $c$-sparse one can eliminate the internal nodes of the connecting paths one by one\cameraready{,}{, like in the proof of Lemma~\ref{lem:sparse},} until a graph isomorphic to $p$ remains. 

We define the remaining graphs $G_i$ inductively, maintaining an additional invariant $G_i \to G$. Suppose we already have $G_i$ together with a homomorphism $h_i \colon G_i \to G$ for some $i\geq 0$. If $h_i$ is injective over $R$-successors of each node of $G_i$, we are done. If not, there are two different $R$-successors $u_1$ and $u_2$ of a node $v$ in $G_i$ that are mapped to the same node $u'$ in $G$. It follows that $u_1$ and $u_2$ have the same sets of labels types. We let $G_{i+1}$ be the graph obtained from $G_i$ by merging $u_1$ and $u_2$ into a single node $u$. We include an $R'$-edge between $u$ and each $R'$-successor of $u_1$ or $u_2$. This decreases the number of nodes by one, and the number or edges by at least one. It follows that $G_{i+1}$ is $c$-sparse and $G_i \to G_{i+1} \to G$. 

Because the sizes of graphs $G_i$ are strictly decreasing, at some point we will arrive at a graph $G_t$ such that the homomorphism from $G_t$ to $G$ is injective over $R$-successors. 

The graph $G_t$ clearly satisfies $p$. It also satisfies all concept inclusions in $\T$ of the forms 
$K \sqsubseteq A_1 \sqcup A_2 \sqcup \dots \sqcup A_n$, $K \sqsubseteq \bot$, $K \sqsubseteq \forall R. K'$, $K \sqsubseteq \nexists R. K'$, and $K \sqsubseteq \exists^{\leq 1} R. K'$, because $h_i$ is injective over $R$-successors and $G\models \T$. On the other hand, $G_t$ is not guaranteed to satisfy concept inclusions of the form $K \sqsubseteq \exists R. K'$ in $\T$. In order to fix it, we exhaustively (ad infinitum) perform the following: whenever a node $v$ in $G_t$ is missing an $R$-successor with some set of labels, we add it and map it to some such $R$-successor $u'$ of the image of $v$ in $G$  ($u'$ exists because $G \models \T$). As $c \leq |p|$, the resulting (typically infinite) graph $\widehat G$ is $|p|$-sparse, and it satisfies $p$ and $\T$.
\end{proof}

The connectedness assumption in Theorem~\ref{thm:sparsification} is not restrictive, because a witnessing graph for $p$ can be obtained by taking the disjoint union of witnesses for its connected components. Hence, it remains to decide for a given connected $p$ if there exists a $|p|$-sparse graph $G$ that satisfies $p$ and $\T$. 
To get a finer control of the effect different parameters of the input have on the complexity, we side-step two-way alternating tree automata (2ATA) applied by \citeauthor{CalvaneseOS11} and develop a more direct algorithm.

Observe that if $p$ is satisfied in a $|p|$-sparse graph $G$, then $G$ contains a $(4|p|, 5|p|)$-skeleton $H'$, extending the skeleton of $G$, such that all variables of $p$ are mapped to distinguished nodes of $H'$. Indeed, $H'$ can be obtained by iteratively extending the skeleton of $G$. Suppose that some variable is mapped to a node $v$ that is not yet a distinguished node of $H'$. If $v$ already belongs to $H'$, then it is an internal node in a path between two distinguished nodes; we then split the path in two, turning $v$ into a distinguished node. If $v$ does not belong to $H'$, then it belongs to a tree attached to $H'$ at a node $u$. If $u$ is not a distinguished node of $H'$, we turn it into one, as above. Then, we add $v$ to $H'$ as a distinguished node, including the path between $u$ and $v$ into $H'$ as well. As we start from a $(2|p|,3|p|)$-skeleton and add at most two distinguished nodes and two paths for each variable of $p$, we end up with a $(4|p|,5|p|)$-skeleton.

Thus, the algorithm can guess a $(4|p|, 5|p|)$-skeleton $H'$ with each path represented by a single symbolic edge and check that it can be completed to a suitable graph $G$ by materializing symbolic edges into paths and attaching finitely many finitely branching trees
in such a way that $G$ is a model of $\T$ and there is a match of $p$ in $G$ that maps variables of $p$ to distinguished nodes of $H'$. 
This can be done within the required time bounds by means of a procedure that can be seen as a variant of type elimination or an emptiness test for an implicitly represented nondeterministic tree automaton\cameraready{.}{ (see Theorem~\ref{thm:sparse-witness}).}

\section{Discussion}
\label{sec:conclusions}

\paragraph{Summary} In this paper we have studied several static analysis problems for graph
transformations defined with Datalog-like rules that use acyclic C2RPQs. When
the source schema is given, we studied the \emph{equivalence} problem of two
given transformations, and the problem of \emph{target schema elicitation} for a
given transformation. If the output schema is also given, we have studied the
problem of \emph{type checking}. We have shown that the above problems can be
reduced to containment of C2RPQs in acyclic UC2RPQs modulo schema, a problem
that we have reduced to the unrestricted (finite or infinite) satisfiabilty of a
C2RPQ modulo \HornALCIF TBox using cycle reversing. 
For the latter problem we have presented an algorithm with sufficiently good complexity to accommodate the exponential blow-up introduced by cycle reversing,  thus allowing to solve in \EXPTIME all problems of interest.
We have also shown matching lower bounds by reducing query containment modulo schema to each of the static
analysis problems.

\paragraph{Finite containment modulo \HornALCIF TBox} In the course of the proof of the upper bound for containment modulo schema, we essentially solved (finite) containment modulo \HornALCIF TBox. Indeed, while the \EXPTIME upper bound relies on the special shape of the TBox expressing the schema, the method can be applied directly to any \HornALCIF TBox, at the cost of an exponential increase in complexity. Thus, we immediately get that finite containment of UC2RPQs in acyclic UC2RPQs modulo \HornALCIF TBoxes can be solved in \TWOEXPTIME. 
To the best of our knowledge this is the first result on finite containment of C2RPQs in the context of description logics. 
A related problem of finite entailment has been studied for various logics \cite{GogaczIM18,GogaczGIJM19,GogaczGGIM20,GutierrezGIM2022}, but
while for conjunctive queries the solutions carry over to finite containment, for C(2)RPQs these logics are too weak to allow this. Unrestricted containment of C2RPQs modulo \ALCIF TBoxes is known to be in \TWOEXPTIME~\cite{CalvaneseOS11}, but passing from unrestricted to finite structures is typically challenging for such problems. For example, finite entailment of CRPQs for a fundamental description logic \ALC{} has been solved only recently~\cite{GutierrezGIM2022}, 15 years after the unrestricted version~\cite{CalvaneseEO07}.

\paragraph{Extending queries} It is straightforward to extend our methods to two-way \emph{nested regular expressions} (NREs)~\cite{PeArGu10}. We also intend to investigate introducing \emph{negation} in filter expressions of NREs.
Eliminating the acyclicity assumption, on the other hand, is problematic. 
Containment of arbitrary C2RPQs is \EXPSPACE-complete~\cite{CGLV00}, and we have shown that it reduces to our problems of interest for transformation rules with cyclic queries. Hence, extending our \EXPTIME upper bounds to transformations allowing cyclic C2RPQs is highly unlikely. In fact, even establishing decidability would be hard. For acyclic queries we could use the rolling-up technique to reduce containment to satisfiability, which allowed us to apply the cycle reversing technique and pass from finite to unrestricted models. When cyclic queries are allowed, the rolling-up technique is inapplicable and we are left with containment of C2RPQs modulo constraints, which is a major open problem, not only for constraints expressed in description logics.  
The only positive results we are aware of do not go significantly beyond CQs extended with a binary reachability relation \cite{DeuTan2002}. 

\paragraph{Extending schemas} Extending the schema  formalism with disjunction is also challenging: the corresponding description logic would not be Horn any more and the transition to unrestricted models via cycle reversing would not be possible. 
Supporting multiple labels on nodes would not be a trivial extension either: we rely on the single label per node assumption in the reduction of the problems of interest to containment of UC2RPQs in \emph{acyclic} UC2RPQs, and in the \EXPTIME upper bound.
Supporting more general cardinality constraints, on the other hand, should be possible, but it might affect the complexity upper bounds. 

\paragraph{Extending the data model}
It is straightforward
to encode data values in our graph model, for instance, by using dedicated node labels to designate \emph{literal} nodes whose identifiers are their data
values. 
Then, one can apply methods similar to type checking to verify that transformations are well-behaved, and in particular, do not attempt to construct literal nodes from non-literal ones. 
However, the full consequences of allowing literal values in definitions of transformation rules need to be thoroughly investigated.

\medskip

\noindent Finally, we have considered equivalence of transformations based on equality of results but one could also consider a variant based on
isomorphism of results. This would be an entirely different problem, probably much harder.

\begin{acks}
This work was supported by Poland's National Science Centre grant 2018/30/E/ST6/00042. 
We would like to thank Sebastian Maneth, Mikaël Monet, Bruno Guillon, and Yazmin Ibáñez-Garc\'ia for their comments and discussions.
For the purposes of open access, the authors have applied a CC BY public copyright licence to any Author Accepted Manuscript version arising from this submission.
\end{acks}

\bibliographystyle{ACM-Reference-Format}
\bibliography{references}

\balance

\appendix

\clearpage
\section{Details on Queries}
\label{sec:app-preliminaries}
A \emph{two-way regular expression} is an expression defined with the following
grammar.
\[
  \varphi \coloncolonequals \varnothing \mid \epsilon \mid A \mid R \mid
  \varphi\cdot\varphi \mid \varphi + \varphi \mid \varphi^*,
\]
where $A\in\Gamma$ and $R\in\Sigma^\pm$. We define the semantics with the notion
of witnessing paths that we formalize next. Given a graph $G$, a \emph{path}
from $u_0$ to $u_k$ in $G$ is a sequence
$\pi=u_0 \cdot \ell_1 \cdot u_1 \cdot \ldots \cdot u_{k-1} \cdot \ell_{k} \cdot
u_k$ such that $u_0,\ldots,u_k$ are nodes of $G$,
$\ell_1,\ldots,\ell_k\in\Gamma\cup\Sigma^\pm$, and for every
$i\in\{1,\ldots,k\}$ the following conditions are satisfied:
\begin{enumerate}
\itemsep0pt
\item if $\ell_i\in\Gamma$, then $u_{i-1}=u_i$ and $u_i\in \ell_i^G$,
\item if $\ell_i\in\Sigma^\pm$, then $(u_{i-1},u_i)\in \ell_i^G$.
\end{enumerate}
The labeling of $\pi$ is $\ell_1\cdot\ldots\cdot \ell_n$. Given a two-way
regular expression $\varphi$ we define the corresponding binary relation on
nodes of the graph: $(u,v)\in[\varphi]_G$ iff there is a path from node $u$ to
node $v$ in $G$ whose labeling is recognized by $\varphi$.

Now, a \emph{conjunctive two-way regular path query} (C2RPQ) is a formula of the
form
\[
  q(\bar{x}) = \exists \bar{y}.
  \varphi_1 (z_1,z_1')\land\ldots\land
  \varphi_k (z_k,z_k'),
\]
where for every $i\in\{1,\ldots,k\}$ the formula $\varphi_i$ is a two-way
regular expression and $\bar{x}=\{z_1,z_1',\ldots,z_k,z_k'\}\minus\bar{y}$. A
C2RPQ is \emph{Boolean} if all of its variables are existentially quantified.

Evaluating a C2RPQ $q(\bar{x})$ over a graph $G$ yields a set $[q(\bar{x})]^G$
of tuples over $\bar{x}$ i.e., functions that assign nodes of $G$ to elements of
$\bar{x}$. Formally, $t\in [q(\bar{x})]^G$ iff there is a tuple $t'$ over
$\bar{y}$ such that the two tuples combined $t''=t\cup t'$ satisfy all atoms
i.e., $(t''(z_i),t''(z_i'))\in [\varphi_i]_G$ for every
$i\in\{1,\ldots,k\}$. When the query is Boolean, then it may have only a single
answer, the \emph{empty tuple} $()$ i.e., the unique function with the empty
domain. If indeed $()\in[q]^G$ we say that $q$ is satisfied in $G$ and denote it
by $G\models q$; otherwise, when $[q]^G=\emptyset$, we say that $q$ is not
satisfied in $G$ and we write $G\not\models q$.

For defining transformations we employ the subclass of acyclic C2RPQs. Formally,
for a query $q$ we construct its \emph{query multigraph} whose nodes are
variables and for every atom $\varphi(x,y)$ we add an edge $(x,y)$ unless the
atom is of the form $A(x,x)$, $\epsilon(x,x)$, or $\varnothing(x,x)$. $q$ is
\emph{acyclic} if its query multigraph is acyclic.

Finally, the semantics of  \emph{unions of conjunctive two-way regular path queries} (UC2RPQs),
represented  as sets of C2RPQs, is defined simply as:
\[[\{Q_1(\bar{x}),\ldots,Q_k(\bar{x})\}]^G= [Q_1(\bar{x})]^G \cup \ldots\cup
[Q_k(\bar{x})]^G\,.\] A UC2RPQ is \emph{acyclic} if all of its components are
acyclic. A \emph{Boolean} UC2RPQ consists of Boolean C2RPQs.

\section{Proofs for Transformations}
\label{sec:app-transformations}
We begin by introducing elements of useful terminology. Given any finite subsets
$\Gamma_0\subseteq\Gamma$ and $\Sigma_0\subseteq\Sigma$, we say that a schema
$S$ is over $\Gamma_0$ and $\Sigma_0$ if $\Gamma_S=\Gamma_0$ and
$\Sigma_S=\Sigma_0$. Analogously, we say that a \ALCIF TBox $\T$ is over
$\Gamma_0$ and $\Sigma_0$ if all base concept names and base rule names used in
$\T$ are from $\Gamma_0$ and $\Sigma_0$ respectively. Also, we say that a graph
$G$ is over $\Gamma_0$ and $\Sigma_0$ if $G$ does not use any node or edge label
outside of $\Gamma_0$ and $\Sigma_0$, and we extend this notion to families of
graphs in the canonical fashion: $\G$ is a family of graphs over $\Gamma_0$ and
$\Sigma_0$ if every graph in $G$ is over $\Gamma_0$ and $\Sigma_0$.  Finally, a
transformation $T$ is over $\Gamma_0$ and $\Sigma_0$ if all rules in $T$ use in
their heads node and edge labels in $\Gamma_0$ and $\Sigma_0$ respectively.

However, for a transformation we shall need to identify tighter sets of node
and edge labels when the input schema is known. As such, a transformation rule
$\rho\leftarrow q(\bar{x})$ is \emph{productive} modulo a schema $S$ if
$q(\bar{x})\nsubseteq_S \varnothing$. A transformation $T$ is \emph{trimmed}
modulo $S$ if 1) every rule in $T$ is productive modulo $S$, 2) for every
$A\in\Gamma_T$ there is an $A$-node rule in $T$, and 3) for every $r\in\Sigma_T$
there is a $r$-edge rule in $T$. Naturally, checking that a transformation is
trimmed can be Turing-reduced in polynomial time to testing query containment
modulo schema. Moreover, for a given schema $S$ we can trim a given
transformation $T$ by removing all unproductive rules and removing from
$\Gamma_T$ and $\Sigma_T$ any symbols that are not present in the head of any of
the remaining rules.

Next, an $\L_0$ TBox over $\Gamma_0$ and $\Sigma_0$ is a set of statements of
the forms
\begin{align*}
  &A \sqsubseteq \exists R. B, &
  &A \sqsubseteq \nexists R. B, &
  &A \sqsubseteq \exists^{\leq 1} R. B,
\end{align*}
where $A,B \in \Gamma_0$ and $R\in\Sigma_0^\pm$. $\T$ is \emph{coherent} iff 1)
$\T$ does not contains two contradictory rules $A\sqsubseteq \exists R. B$ and
$A\sqsubseteq \nexists R.B$ for any $A,B\in\Gamma$ and $R\in\Sigma^\pm$, and 2)
$\T$ contains $A\sqsubseteq \exists^{\leq 1} R.B$ whenever it contains
$A\sqsubseteq \nexists R. B$. 
Now, for a given schema $S$ the corresponding
$\L_0$ TBox $\T_S$ (over $\Gamma_S$ and $\Sigma_S$) is defined as follows.
\begin{align*}
  \T_S = {}
  & \{ A \sqsubseteq \exists R.B \mid A,B\in\Gamma_S,\ R\in\Sigma_S^\pm,\
    \delta_S(A,R,B) \in\{\ONE,\PLUS\} \}\\[-2pt]
  \cup\ & \{ A \sqsubseteq \exists^{\leq 1}
          R.B \mid A,B\in\Gamma_S,\ R\in\Sigma_S^\pm,\
          \delta_S(A,R,B) \in \{\ONE,\MAYBE,\NONE\} \}\\[-2pt]
  \cup\ & \{ A \sqsubseteq \nexists R.B \mid A,B\in\Gamma_S,\
          R\in\Sigma_S^\pm,\ \delta_S(A,R,B) = \NONE \}.
\end{align*}
It is easy to see that there is one-to-one correspondence between schemas and
coherent TBoxes. More precisely, given $\Gamma_0\subseteq\Gamma$ and
$\Sigma_0\subseteq\Sigma$, for any schema $S$ over $\Gamma_0$ and $\Sigma_0$,
$\T_S$ is a coherent TBox over $\Gamma_S$ and $\Sigma_S$, and for any coherent
TBox $\T$ over $\Gamma_0$ and $\Sigma_0$ there is a unique schema $S$ over
$\Gamma_0$ and $\Sigma_0$ such that $\T_S=\T$. Naturally, $\T_S$ also captures
the semantics of the cardinality constraints of $S$.
\begin{proposition}
  \label{prop:schema-as-tbox}
  For any schema $S$ and for any graph $G$, $G$ conforms to $S$ if and
  only if $G \models \T_S$, $G\models \top\sqsubseteq\bigsqcup \Gamma_S$, and
  $G\models A\sqcap B \sqsubseteq \bot$ for any $A,B\in\Gamma_S$.
\end{proposition}
\begin{proof}
  Straightforward since the \ALCIF formulas are translations of the conditions
  of conformance of a graph to a schema.
\end{proof}

We use the above result to reduce type checking to testing entailment of simple
\ALCIF statements. Recall that for a schema $S$ and a transformation $T$ we
define the entailment relation $(T,S)\models K \sqsubseteq K'$ as
$T(G)\models K \sqsubseteq K'$ for every $G\in L(S)$.
\begin{lemma}
  \label{lemma:transformations-type-checking}
  Given two schemas $S$ and $S'$ and a transformation $T$,
  $\{T(G)\mid G\in L(S)\}\subseteq L(S')$ if and only if
  $(T,S)\models \top\sqsubseteq \bigsqcup \Gamma_T$ and $(T,S)\models\T_{S'}$.
\end{lemma}
\begin{proof}
  Immediate consequence of Proposition~\ref{prop:schema-as-tbox} and the fact
  that transformations must use a single dedicated node constructor for each
  node label. This ensures that $(T,S)\models A\sqcap B \sqsubseteq \bot$ holds
  for any $A,B\in\Gamma_{S'}$.
\end{proof}
\noindent
Later we prove how to reduce entailment of statements to query
containment. Before, we address the problem of schema elicitation by observing
that the correspondence between schemas and their $\L_0$ TBoxes is tighter. We
first need to establish two auxiliary results. The first one characterizes the
containment of schemas, which is expressed as an extension of a syntactic
containment relation $\preccurlyeq$ on the symbols used to specify participation
constraints. More precisely, we define $\preccurlyeq$ as the transitive and
reflexive closure of the following assertions: $\NONE\preccurlyeq\MAYBE$,
$\ONE\preccurlyeq\MAYBE$, $\MAYBE\preccurlyeq\PLUS$, and
$\PLUS\preccurlyeq\MANY$.
\begin{proposition}
  \label{prop:schema-containment}
  Take finite $\Gamma_0\subseteq\Gamma$ and $\Sigma_0\subseteq\Sigma$.  Given
  two schemas $S_1$ and $S_2$ over $\Gamma_0$ and $\Sigma_0$,
  $L(S_1) \subseteq L(S_2)$ if and only if
  \[\delta_{S_1}(A,R,B)\preccurlyeq\delta_{S_2}(A,R,B)\] for every $A,B\in\Gamma_0$
  and $R\in\Sigma_0^\pm$.
\end{proposition}
\begin{proof}
  For the \emph{if} part, we take any $G$ that conforms to $S_1$ and we note
  first that every node of $G$ has exactly one label in $\Gamma_0$. Also, for
  any $A,B,\in\Gamma_0$ and any $R\in\Sigma_0^\pm$ we observe that
  \[\delta_{S_1}(A,R,B)\preccurlyeq\delta_{S_2}(A,R,B)\] implies that any $A$-node in
  $G$ whose number of $R$-successors with label $B$ satisfies the participation
  constraint $\delta_{S_1}(A,R,B)$ will also satisfy $\delta_{S_2}(A,R,B)$.
\end{proof}
\noindent
Next, we establish correspondence between $\L_0$ theories of sets of graphs and
their containment-minimal schemas.
\begin{proposition}
  \label{prop:minimal-schema}
  Take finite $\Gamma_0\subseteq\Gamma$ and $\Sigma_0\subseteq\Sigma$ and take
  any nonempty family $\G$ of graphs over $\Gamma_0$ and $\Sigma_0$ such that
  $\G\models \top\sqsubseteq \bigsqcup\Gamma_0$ and
  $\G\models A\sqcap B\sqsubseteq \bot$ for all $A,B\in\Gamma_0$. Let $\T$ be
  the set of all $\L_0$ statements over $\Gamma_0$ and $\Sigma_0$ that hold in
  every graph in $\G$. Then, $\T$ corresponds to the containment minimal schema
  $S$ over $\Gamma_0$ and $\Sigma_0$ such that $\G\subseteq L(S)$.
\end{proposition}
\begin{proof}
  We first argue that $\T$ is coherent. Indeed, should $\T$ contain two
  contradictory statements $A\sqsubseteq \exists R. B$ and
  $A\sqsubseteq \nexists R.B$, then no graph in $\G$ could satisfy $\T$ and we
  know that $\G$ is nonempty. Consequently, $\T$ corresponds to a schema that we
  denote $S^\circ=(\Gamma_0,\Sigma_0,\delta_{S^\circ})$. Naturally,
  $\G\subseteq L(S^\circ)$ because $\G\models \top\sqsubseteq \bigsqcup\Gamma_0$
  and $\G\models A\sqcap B\sqsubseteq \bot$.

  Now, take any schema $S$ over $\Gamma_0$ and $\Sigma_0$ such that
  $\G\subseteq L(S)$. We show that $L(S^\circ)\subseteq L(S)$ with a proof by
  contradiction. Suppose $L(S^\circ)\nsubseteq L(S)$. By
  Proposition~\ref{prop:schema-containment}, there are $A,B\in\Gamma_0$ and
  $R\in\Sigma_0^\pm$ such that
  $\delta_{S^\circ}(A,R,B)\not\preccurlyeq\delta_S(A,R,B)$. This means that
  $\T_S$ contains an $(A,R,B)$-constraint $\phi$ that $\T_{S^\circ}$ does not
  (by $(A,R,B)$-constraints we mean $A\sqsubseteq\exists R.B$,
  $A\sqsubseteq\exists^{\leq 1} R.B$, and $A\sqsubseteq\nexists R.B$). Since
  $\phi\not\in\T_{S^\circ}$ there is a graph $G\in\G$ such that
  $H\not\models\phi$, and consequently, $G$ does not conform to $S$. Thus
  $\G\nsubseteq L(S)$, a contradiction.
\end{proof}

\noindent
We obtain the following result allowing to solve the problem of schema
elicitation problem.
\begin{lemma}
  \label{lemma:transformations-schema-elicitation}
  Take a schema $S$ and a transformation $T$ that is trimmed modulo $S$ and such
  that $(T,S)\models \top\sqsubseteq \bigsqcup \Gamma_T$. Let $\T$ be the set of
  all $\L_0$ statements over $\Gamma_T$ and $\Sigma_T$ that are satisfied by
  every graph in the family $\{T(G) \mid G\in L(S)\}$. Then, $\T$ corresponds to
  the containment minimal schema over $\Gamma_T$ and $\Sigma_S$ that contains
  $\{T(G) \mid G\in L(S)\}$.
\end{lemma}
\begin{proof}
  The proof follows immediately from Proposition~\ref{prop:minimal-schema}
  except for the case when $T$ is empty. Then, however, $\Gamma_T$ and
  $\Sigma_T$ are empty too and so is $\T$. However, the schema that corresponds
  to $\T$ is also empty and it recognizes only empty graphs.  As such it is the
  containment minimal schema over $\Gamma_T$ and $\Sigma_T$ that contains
  $\{T(G) \mid G\in L(S)\}\subseteq\{\emptyset\}$.
\end{proof}
\noindent
To move to reducing entailment of statements to query containment we repeat the
definitions of the relevant queries but in this version we clearly indicate the
transformation in question. More precisely, For a transformation $T$,
$A,B\in\Gamma_T$, and $r\in\smash{\Sigma_T}$ we define:
\begin{align*}
  & Q_{A}^T(\bar{x})=\big\{q(\bar{x}) \bigm{|}
    A\big(f_A(\bar{x})\big)\leftarrow q(\bar{x})\in T\big\}, \\
  & Q_{A,r,B}^T(\bar{x},\bar{y})=\big\{q(\bar{x},\bar{y}) \bigm{|}
    r\big(f_A(\bar{x}),f_B(\bar{y})\big)\leftarrow q(\bar{x},\bar{y})\in T \big\},\\
  & Q_{A,r^-,B}^T(\bar{x},\bar{y})=\big\{q(\bar{y},\bar{x}) \bigm{|}
    r\big(f_B(\bar{y}),f_A(\bar{x})\big)\leftarrow q(\bar{y},\bar{x})\in T \big\}.
\end{align*}
Now, we prove that the entailment of $\top\sqsubseteq \bigsqcup \Gamma_T$ is
reduced to query containment.
\begin{lemma}
  \label{lemma:transformations-consistency}
  Given a schema $S$ and a transformation $T$,
  $(T,S)\models \top\sqsubseteq \bigsqcup \Gamma_T$ if and only if
  $\exists\bar{y}.Q_{A,R,B}^T(\bar{x}, \bar{y}) \subseteq_S Q_A^T(\bar{x})$ for
  every $A,B\in\Gamma_T$ and $R\in\Sigma_T^\pm$.
\end{lemma}
\begin{proof}
  For the \emph{if} direction, we take any graph $G\in L(S)$ and any element in
  $u\in\dom(T(G))$. This element has been introduced by node rule or by an
  edge rule, but only the latter is of concern. Thus, assume that $u=f(t)$ has
  been generated by the rule
  $R(f_A(\bar{x}),f_B(\bar{y}))\leftarrow Q(\bar{x},\bar{y})$ with the valuation
  $\bar{x}=t$ and $\bar{y}=t'$ . Since $(t,t')\in Q_{A,R,B}^T(G)$ and
  $\exists \bar{y}. Q_{A,R,B}^T(\bar{x},\bar{y}) \subseteq_S Q_A^T(\bar{x})$,
  $t\in Q_A^T(G)$, and therefore, there is a node rule
  $A(f_A(\bar{x}))\leftarrow Q'(\bar{x})$ such that $t\in
  Q'(G)$. Consequently, $u\in A^{T(G)}$.

  For the \emph{only if} direction, we take any $G\in L(S)$ and any answer
  $(t,t')\in Q_{A,R,B}^T(\bar{x},\bar{y})$ which implies that
  $(t,t')\in q(\bar{x},\bar{y})$ for some rule
  $R(f_A(\bar{x}),f_B(\bar{y}))\leftarrow q(\bar{x},\bar{y})$. Consequently,
  $T(G)$ contains the fact $R(f_A(t),f_B(t'))$. Since $T(G)$ satisfies the
  statement $\top\sqsubseteq\bigsqcup\Gamma_T$ and nodes constructed with $f_A$
  can only be part of node label assertions with $A$, $T(G)\models
  f_A(t)$. Therefore, there must be a rule
  $A(f_A(\bar{x})\leftarrow q'(\bar{x})$ that generated the fact $f_A(t)$ with
  the valuation $\bar{x}=t$. Consequently, $t\in[Q_A]^{T(G)}$.
\end{proof}

\begin{lemma}
  \label{lemma:transformations-entailment}
  Take a schema $S$ and a transformation $T$, such that
  $\Gamma_T\subseteq\Gamma_S$, $\Sigma_S\subseteq\Sigma_S$, and
  $(T,S)\models\top\sqsubseteq \bigsqcup \Gamma_T$. For any $A,B\in\Gamma_T$ and
  any $R\in\Sigma_T^\pm$ we have that
  \begin{align*}
    & (T,S)\models A \sqsubseteq \exists R.B\iff 
      Q_A(\bar{x}) \subseteq_S Q_{A,R,B}^T(\bar{x}),\\[-2pt]
    &(T,S)\models A \sqsubseteq \nexists R.B \iff 
      \exists \bar{y}.Q_A(\bar{x}) \!\land\! Q_{A,R,B}^T(\bar{x},\bar{y})
      \subseteq_S \varnothing,\\[-2pt]
    &(T,S)\models A \sqsubseteq \exists^{\leq 1} R.B\iff\\[-2pt]
    & \hspace{7.75ex}\exists \bar{x}. Q_A^T(\bar{x}) \!\land\!
      Q_{A,R,B}^T(\bar{x},\bar{y}) \!\land\!
      Q_{A,R,B}^T(\bar{x},\bar{z}) \subseteq_S
      \textstyle\bigwedge_i [\epsilon(y_i,z_i).
  \end{align*}
\end{lemma}
\begin{proof} We prove each of the 3 claims separately. 
  \begin{enumerate}
  \item For the \emph{if} part, we fix a graph $G\in L(S)$ and take any node
    $u=f_A(t)$ with label $A$ in $T(G)$. Thus, there is a node rule
    $A(f_A(\bar{x}))\leftarrow q(\bar{x})$ such that $t\in [q(\bar{x})]^G$ and
    consequently, $t\in [Q_A^T(\bar{x})]^G$. Since
    $Q_A^T(\bar{x}) \subseteq_S Q_{A,R,B}^T(\bar{x})$,
    $t\in [ Q_{A,R,B}^T(\bar{x})]^G$ and there exists rule
    $R(f_A(\bar{x}),f_B(\bar{y}))\leftarrow q'(\bar{x},\bar{y})$ such that
    $(t,t')\in[q'(\bar{x},\bar{y})]^G$. Consequently, $T(G)$ contains the edge
    $R(f_A(t),f_B(t'))$. Because $T(G)$ satisfies
    $\top\sqsubseteq\bigsqcup\Gamma_T$, there is also a rule
    $B(f_B(\bar{y}))\leftarrow q''(\bar{y})$ such that $t'\in[q''(\bar{y})]^G$,
    and hence the node $f_B(t')$ has label $B$ in $G$.

    For the \emph{only if} part, we fix a graph $G\in L(S)$ and take any
    $t\in[Q_A^T]^G$, which means that there is a node rule
    $A(f_A(\bar{x}))\leftarrow q(\bar{x})$ with
    $t\in[q(\bar{x})]^G$. Consequently, $A(f_A(t))$ belongs to $T(G)$. Since
    $\M_0(G)\models A \sqsubseteq \exists R.B$, $G$ has an edge $R(f_A(t),v)$
    and the node $v$ has label $B$. This edge must be generated by an edge rule
    $R(f_A(\bar{x}),f_B(\bar{y}))\leftarrow q'(\bar{x},\bar{y})$. Consequently,
    $t$ belongs to the answers to $\exists \bar{y}. q'(\bar{x},\bar{y})]$ which
    is contained in $Q_{A,R,B}(\bar{x})$ modulo $S$.
    
  \item The proof of this statement is by contradiction and it uses arguments
    that are analogous to those used in the proof of the above claim and we only
    outline it. We take a graph $G\in L(S)$ such that in $T(G)$ there is a node
    $f_A(t)$ with label $A$ and an $R$-edge to a node with with label $B$. This
    happens if and only if the intersection of $Q_A(\bar{x})$ and
    $\exists \bar{y}. Q_{A,R,B}(\bar{x},\bar{y})$ is non-empty.
    
  \item Similarly, the proof is by contradiction but uses argument analogous to
    those in the proof of the first claim and we only outline it. We take a
    graph $G\in L(S)$ such that $T(G)$ has an $A$-node $f_A(t)$ which has
    $R$-edges to two different $B$-nodes $f_B(t'_1)$ and $f_B(t'_2)$. This is
    possible if and only if the query
    $\exists \bar{x}. Q_A^T(\bar{x}) \land Q_{A,R,B}^T(\bar{x},\bar{y})$
    returns both $t'_1$ and $t'_2$, and consequently,
    $\exists \bar{x}. Q_A^T(\bar{x}) \land Q_{A,R,B}^T(\bar{x},\bar{y}) \land
    Q_{A,R,B}^T(\bar{x},\bar{z}) \subseteq_S \bigwedge_i \epsilon(y_i,z_i)$
    returns $(t'_1,t'_2)$.  Because node constructors are invective,
    $t'_1\neq t'_2$, and therefore, $(t'_1,t'_2)$ cannot be answer to
    $\bigwedge_i \epsilon(y_i,z_i)$. \qedhere
  \end{enumerate}
\end{proof}

For testing equivalence of two transformations we observe that since a
transformation is equivalent to its trimmed version, two transformations $T_1$
and $T_2$ are equivalent modulo $S$ if and only if they trimmed versions
$\trim_S(T_1)$ and $\trim_S(T_2)$ are equivalent modulo $S$. In the following
lemma, $Q_1\equiv_S Q_2$ is short for $Q_1\subseteq_S Q_2$ and
$Q_2\subseteq_S Q_1$.
\begin{lemma}
  \label{lemma:transforamtions-equivalence}
  Take a schema $S$ and two transformations $T_1$ and $T_2$ that are both trimmed
  modulo $S$. We have that $T_1\equiv_S T_2$ if and only if the following
  conditions are satisfied:
  \begin{enumerate}
  \item $\Gamma_{T_1}=\Gamma_{T_2}$ and $\Sigma_{T_1}=\Sigma_{T_2}$,
  \item $Q_A^{T_1}(\bar{x}) \equiv_S Q_A^{T_2}(\bar{x})$ for every
    $A\in \Gamma_{T_1}$,
  \item
    $Q_{A,R,B}^{T_1}(\bar{x},\bar{y}) \equiv_S Q_{A,R,B}^{T_2}(\bar{x},\bar{y})$
    for every $A,B\in \Gamma_{T_1}$, $R\in\Sigma_{T_1}$.
  \end{enumerate}
\end{lemma}
\begin{proof}
  The \emph{if} part is trivial. We prove the \emph{only if} part by proving the
  contraposition: we show that if one of the conditions (1), (2), and (3) is not
  satisfied, then $T_1\nequiv_S T_2$.

  If (1) is not satisfied, then one of the transformations has at least one rule
  $\rho$ that generate a node or an edge with a label that is not employed by
  the other transformations. Since both transformations are trimmed, there exists
  an input graph $G$ such that the rule $\rho$ produces objects on the
  output. But then $T_1(G)\neq T_2(G)$.

  If (2) is not satisfied, then there is an input graph $G$ such that one of the
  transformations generates a node that the other does not. Hence,
  $T_1(G)\neq T_2(G)$.

  If (3) is not satisfied, then analogously, there is an input graph $G$ such
  that one of the transformations generates an edge that the other does
  not. Hence, $T_1(G)\neq T_2(G)$.
\end{proof}

\section{Rolling up queries}

We next show how to reduce the non-satisfaction of an acyclic UC2RPQ $Q$ to the
satisfaction of a \HornALCIF TBox $\T_{\lnot Q}$. The TBox $Q$ is basically a
recursive program that defines a collection of sets (monadic relations) of
nodes. We illustrate this construction with the following example.
\begin{example}
  \label{ex:rollingup}
  We take the following Boolean query. 
  \[
    Q_0=\exists x_0,x_1,x_2,x_3.\
    (a\cdot b^*\cdot c)(x_2,x_1)\ \land \ (A)(x_3,x_1)\ \land\ (a^-)(x_1,x_0). 
  \]
  We construct a TBox that essentially simulates automata for the regular
  expressions, which are presented in Figure~\ref{fig:tree-like-query-NDAs}.
  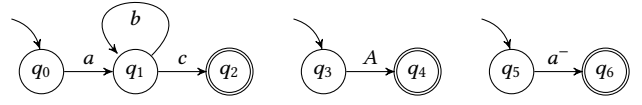
\begin{figure}[htb]
  \centering
  \begin{tikzpicture}[>=stealth']
    \path[use as bounding box] (-0,-0.2) rectangle (7.5,.75);
    \small 

    \begin{scope}[xshift=0cm]
    \node[circle,draw] (q5) at (0,0) {$q_0$};
    \node[circle,draw] (q6) at (1.25,0) {$q_1$} edge[<-] node[above] {$a$}(q5);
    \node[circle,draw,double] (q7) at (2.5,0) {$q_2$} edge[<-] node[above] {$c$}(q6);
    \draw (q6) edge[->,loop] node[below] {$b$} (q6);
    \node at (-.5,0.75) {} edge[bend left,->] (q5);
    \end{scope}

    \begin{scope}[xshift=3.75cm]
    \node[circle,draw] (q2) at (0,0) {$q_3$};
    \node[circle,draw,double] (q3) at (1.25,0) {$q_4$} edge[<-] node[above] {$A$}(q2);
    \node at (-.5,0.75) {} edge[bend left,->] (q2);
    \end{scope}

    \begin{scope}[xshift=6.25cm]
    \node[circle,draw] (q0) at (0,0) {$q_5$};
    \node[circle,draw,double] (q1) at (1.25,0) {$q_6$} edge[<-] node[above] {$a^-$}(q0);
    \node at (-.5,0.75) {} edge[bend left,->] (q0);
    \end{scope}
    
  \end{tikzpicture}
  \caption{\label{fig:tree-like-query-NDAs} Automata for regular expressions of $Q$.}
  \end{figure}

  \noindent
  The TBox $\T_{\lnot Q_0}$ consists of the following constraints.
    \begin{align*}
      & \top \sqsubseteq q_0\,,&
      & q_0 \sqsubseteq \forall a.q_1\,,&
      & q_1 \sqsubseteq \forall b.q_1\,,&
      & q_1 \sqsubseteq \forall c.q_2\,,&&\\
      & \top \sqsubseteq q_3\,,&
      & q_3 \sqcap A \sqsubseteq q_4\,,&
      & q_2 \sqcap q_4 \sqsubseteq q_5\,,&
      & q_5 \sqsubseteq \forall a^-.q_6\,,&
      & q_6 \sqsubseteq \bot\,.\tag*{\qed}
  \end{align*}
\end{example}
\noindent
$\T_{\lnot Q}$ introduces a set fresh node labels $\Gamma_Q$ and the
satisfaction $\T_{\lnot Q}$ is defined in terms of the existence of valuations
of symbols in $\Gamma_Q$. More precisely, given a graph $G$ over $\Gamma_0$ and
$\Sigma_0$ and a TBox $\T$ over $\Gamma_0\cup\Gamma_1$ and $\Sigma_0$, we say
that $G$ \emph{satisfies} $\T$ if and only if there is an interpretation
$\cdot^U:\Gamma_1\rightarrow \P(\dom(G))$ of symbols in $\Gamma_1$ such that
$G\cup U\models \T$.
\begin{lemma}
\label{lem:rollingup}
Given a Boolean acyclic UC2RPQs $Q$, one can compute in polynomial time a
\HornALCIF TBox $\T_{\lnot Q}$ and a reserved set of concept names $\Gamma_Q$
such that for every $G$ that does not use labels in $\Gamma_Q$, $G\not\models Q$
if and only if $G$ satisfies $\T_{\lnot Q}$.
\end{lemma}
\begin{proof}
We prove the lemma for queries that are Boolean C2RPQs that are acyclic and
connected. The claim extends to unions of Boolean acyclic C2RPQs in a
straightforward fashion: it suffices to take the union of the desired TBoxes of
all connected components of the union. Consequently, the query can be seen as a tree and
we assume that it is defined with the following grammar:
\[
  Q \coloncolonequals \varphi(Q,\ldots,Q),
\]
where $\varphi$ is a two-way regular expression over $\Sigma$ and $\Gamma$. For
instance, the query from Example~\ref{ex:rollingup} is represented as
$Q_0 = a^-(A,a\cdot b^*\cdot c)$. We express the semantics of such defined
queries as the set of all nodes that satisfy it.
\[
  [\varphi(Q_1,\ldots,Q_k)]^G = \{
  u \in \dom(G) \mid
  \exists v.\
  (v,u)\in[\varphi]_G,\
  v \in\textstyle\bigcap_i [Q_i]^G
  \}.
\]
Naturally, a graph $G$ satisfies $Q$ iff $[Q]^G\neq\emptyset$.

Now, fix an acyclic Boolean C2RPQ $Q$ and let $\Phi$ be the set of all two-way
regular expressions used in $Q$. For any $\varphi\in\Phi$ by
$N_\varphi=(K_\varphi,I_\varphi,\delta_\varphi,F_\varphi)$ we denote an
$\epsilon$-free NDA over the alphabet $\Sigma\cup\Gamma$ that recognizes
$\varphi$, where $K_\varphi$ is a finite set of states,
$I_\varphi\subseteq K_\varphi$ is the set of initial states,
$F_\varphi\subseteq K_\varphi$ is the set of final states, and
$\delta_\varphi\subseteq K_\varphi\times(\Sigma\cup\Gamma)\times K_\varphi$ is
the transition table. We assume that the size of $N_\varphi$ is polynomial in
the size of the expression $\varphi$ (such automaton can be obtained for
instance with the standard Glushkov technique). We also assume that the sets of
states are pair-wise disjoint.

The set of additional node labels consists of the states of automata:
$\Gamma_Q=\bigcup_\varphi K_\varphi$. The constructed TBox consists of two
subsets of rules: $\T_{\lnot Q}=\T_1\cup\T_0$. The set $\T_1$ encodes
transitions of the automata that simulate their execution.
\begin{enumerate}
\item[(1)] For every $\varphi$ and every $(q,R,q')\in\delta_\varphi$ such that
  $R\in\Gamma^\pm$, $\T_1$ contains
  $q \sqsubseteq \forall R . q'$;
\item[(2)] For every $\varphi$ and every $(q,A,q')\in\delta_\varphi$ such that
  $A\in\Sigma$, $\T_1$ contains $q \sqcap A \sqsubseteq q'$;
\item[(3)] For every node $\varphi$ of $Q$ with children
  $\varphi_1,\ldots,\varphi_k$, every $q\in I_\varphi$, $\T_1$ contains
  $\textstyle\bigsqcap\{q'\mid q' \in F_{\varphi_i},\ 1\leq i \leq k\}
  \sqsubseteq q$.  Note that when $\varphi$ is a leaf of $Q$, then $\T_1$
  contains $\top \sqsubseteq q$ for every $q\in I_\varphi$.
\end{enumerate}
The set $\T_0$ contains denial rules that ensure lack of valid run.
\begin{enumerate}
\item[(4)] For every $q\in F_\varphi$ of the root $\varphi$ of $Q$, 
  $\T_0$ contains $q \sqsubseteq \bot$;
\end{enumerate}
\noindent
Now, we fix a graph $G$ whose node labels do not use any symbol in $\Gamma_Q$. We
first argue that there is a unique minimal interpretation
$U_0:\Gamma_Q\rightarrow\P(\dom(G))$ such that $G\cup U_0\models\T_1$. Indeed, since
the rules are Horn-like, an intersection of two models of $\T_1$ is also a model
of $\T_1$.

Next, we prove the main claim with an inductive argument which requires defining
subqueries of $Q$. For $\varphi\in\Phi$ and $q\in K_\varphi$ by $Q_q$ we denote
the query $\psi(Q_1,\ldots,Q_k)$, where $Q_1,\ldots,Q_k$ are children of
$\varphi$ in $Q$ and $\psi$ is the two-way regular expression corresponding to
the automaton $M_{\varphi,q}=(K_\varphi,I_q,\delta_\varphi,\{q\})$ (essentially,
we make $q$ the only final state).  We claim that for any $\varphi\in\Phi$, any
$q\in K_\varphi$, and any $u\in N_G$ we have
\[
  u \in [Q_q]^G \iff u\in q^{U_0}.
\]
In essence, the unary predicate $q$ identifies all nodes at which the subquery
$Q_q$ is satisfied. We prove the above claim with double induction: firstly over
the height of the subquery $Q_q=\psi(Q_1,\ldots,Q_k)$, and secondly, over the
length of the witnessing path for $(v,u)\in[\psi]^G$ such that
$v\in\bigcap_i [Q_i]^G$.

If we let $I_\varphi=\{q_1,\ldots,q_k\}$, then $Q$ is equivalent to the union of
$Q_{q_1}\cup \ldots \cup Q_{q_k}$. Consequently, $Q$ is satisfied at a node
$u\in N_G$ iff $u\in {q_i}^{U_0}$ for some $i\in\{1,\ldots,k\}$. As such, $Q$ is
not satisfied at any node of $G$ if and only if
$U_0\models q_i \sqsubseteq \bot$ for every $i\in\{1,\ldots,k\}$ i.e.,
$U_0\models\T_0$. We finish the proof by observing that if the minimal model
$U_0$ does not satisfy $\T_0$, then none of supersets of $U_0$ does.
\end{proof}

\section{Proofs for Containment}
\label{sec:app-containment}

\begin{lemma}
\label{lem:boolean-containment}
Given a schema $S$, a UC2RPQ $P(\bar{x})$, and an acyclic UC2RPQ $Q(\bar{x})$, one can compute in polynomial time a schema $S^\circ$, a Boolean UC2RPQ $P^\circ$, and a Boolean acyclic UC2RPQ $Q^\circ$ such that $P(\bar{x})\subseteq_S Q(\bar{x})$ iff $P^\circ \subseteq_{S^\circ} Q^\circ$.
\end{lemma}

\begin{proof}
  Let $\bar x = (x_1, x_2, \dots, x_n)$ and let
  $\Gamma_S=\{A_1,\ldots,A_k\}$. We take a fresh node labels
  $X_1,\ldots,X_n\not\in\Sigma_S$ and fresh edge labels
  $r_1, r_2, \dots, r_n\not\in\Sigma_S$. The schema $S^\circ$ is obtained from
  $S$ as follows:
  \begin{align*}
    & \Gamma_{S^\circ} = \Gamma_S \cup \{A_0\},\\
    & \Sigma_{S^\circ} = \Sigma_S \cup \{r_1,\ldots,r_n\},\\
    & \delta_{S^\circ}(A,R,B) =
      \begin{cases}
        \delta_S(A,R,B) &\text{if $A,B\in\Gamma_S$ and $R\in\Sigma_S^\pm$},\\
        \MAYBE & \text{if $A=X_i$, $R\in\{r_i,r_i^-\}$, and $B\in\Gamma_S$},\\
        \NONE &\text{otherwise.}
      \end{cases}
  \end{align*} 
  Now, the queries $P^\circ$ and $Q^\circ$ are obtained from $P(\bar{x})$ and
  $Q(\bar{x})$ by quantifying existentially $x_1, x_2, \dots, x_n$ and also
  adding atoms $\exists y. (X_i\cdot r_i)(y,x_i)$ for every
  $i\in\{1,\ldots,n\}$. It is routine to check that
  $P(\bar{x})\subseteq_S Q(\bar{x})$ if and only if
  $P^\circ \subseteq_{S^\circ} Q^\circ$. There are two key facts. Firstly, 2RPQs
  in $P$ and $Q$ do not use labels $r_1, r_2, \dots, r_n$ (nor wildcards) and
  consequently cannot traverse edges with such labels. Secondly, the schema
  $S^\circ$ ensures that the original regular expression can be witnessed only
  by paths that begin and end in nodes with labels in $\Sigma_S$ only.
\end{proof}

\begin{corollary}
  \label{cor:boolean-rpq-containement}
  Given a schema $S$, two unary acyclic 2RPQs $p(x)$ and $q(x)$, one can compute
  in polynomial time a schema $S^\circ$ and Boolean 2RPQs $p^\circ$ and
  $q^\circ$ such that $p(x)\subseteq_S q(x)$ iff
  $p^\circ \subseteq_{S^\circ} p^\circ$.
\end{corollary}
\begin{proof}
  The construction of $S^\circ$ is as in Lemma~\ref{lem:boolean-containment} and
  the construction of Boolean RPQs depends on the form of the unary RPQ: 1) if
  $p(x_1)=\exists x_2. \varphi(x_1,x_2)$, then $p^\circ = r_1\cdot \varphi$ and
  2) if $p(x_1)=\exists x_2. \varphi(y,x)$, then $p^\circ = \varphi\cdot r_1^-$;
  $q^\circ$ is constructed in the same way.
\end{proof}

\begin{lemma}\label{lem:labels}
$P \subseteq_S Q$ iff $\widehat P \subseteq_{\widehat\T_S} Q$.

\end{lemma}

\begin{proof}
Each finite graph falsifying the left-hand side condition falsifies the right-hand side condition as well. For the converse, let $G$ be a finite graph falsifying the right-hand side condition. Without loss of generality we can assume that only labels from $\Gamma_S\cup\Sigma_S$ are used  in $G$. Let $G'$ be obtained by dropping all nodes without a label, as well as edges incident with these nodes. Because all concept inclusions in $\widehat\T_S$ that require a witnessing neighbour specify the label of this neighbour, they are not affected by this modification. Other concept inclusions are always preserved when passing to a subgraph. It follows that $G'$ conforms to $S$. The RPQs in $\widehat P$ can only traverse nodes with a label from $\Gamma_S$, so $\widehat P$ is still satisfied in $G'$. Then, $P$ is satisfied as well. $Q$ is not satisfied in $G'$, because $G'$ is a subgraph of $G$.  
\end{proof}

\begin{lemma}
\label{lem:finmod}
$\widehat P$ is finitely satisfiable modulo $\widehat \T_S \cup\T_{\lnot\Q}$ 
iff $\widehat P$ is satisfiable modulo $\big(\widehat\T_S\cup\T_{\lnot\Q}\big)^*$. 
\end{lemma}

\begin{proof}
Suppose that $\widehat P$ is satisfied in a finite model $G$ of $\widehat \T_S \cup\T_{\lnot Q}$. By Theorem~\ref{thm:completion}, there is a (possibly infinite) model of $\big(\widehat \T_S \cup\T_{\lnot Q}\big)^*$ containing $G$ as a subgraph. This model obviously satisfies $\widehat P$.

Conversely, suppose that there is a possibly infinite graph $G$ satisfying $\widehat P$ and $\left(\widehat \T_S \cup\T_{\lnot Q}\right)^*$. Let $p$ be the disjunct of $\widehat P$ that is satisfied in $G$. Let $H$ be the image of $p$ in $G$, including a finite witnessing path for each RPQ.  Note that $H$ is finite. By Theorem~\ref{thm:completion}, there is a finite model of $\widehat \T_S \cup\T_{\lnot\Q}$ containing $H$ as a substructure. This models satisfies $\widehat P$ as well. 
\end{proof}

\begin{lemma}
Every $S$-driven TBox $\T$ can be simplified in polynomial time so that it contains at most  $|\Sigma^\pm_S|\cdot|\Gamma_S|^2$ 
at-most constraints. 
\end{lemma}

\begin{proof}
To achieve this, for each such CI  of the form $K \sqsubseteq \exists^{\leq 1} R. K'$ in $\T$ we do one of the following.
\begin{itemize}
    \item If $\T$ contains $A \sqsubseteq \exists^{\leq 1} R. A'$ for some $A,A'\in \Gamma_S$ such that $A\in K$ and $A'
    \in K'$, then simply remove $K \sqsubseteq \exists^{\leq 1} R. K'$ from $\T$. This is correct because $A \sqsubseteq \exists^{\leq 1} R. A' \models  K \sqsubseteq \exists^{\leq 1} R. K'$.
    \item Otherwise, because $\T$ is $S$-driven, it follows that the triple $(K, R, K')$ is not satisfiable modulo $\T$. That is, $\T \models K \sqsubseteq \not\exists R.K'$. Since $K \sqsubseteq \not\exists R.K' \models K \sqsubseteq \exists^{\leq 1} R. K'$, we can safely replace $K \sqsubseteq \exists^{\leq 1} R. K'$ with $K \sqsubseteq \not\exists R.K'$ in $\T$.
\end{itemize}
The resulting TBox $\T'$ is equivalent to $\T$ and it only contains at-most constraints involving single concept names from $\Gamma_S$. The number of those is clearly bounded by $|\Sigma^\pm_S|\cdot|\Gamma_S|^2$. 
\end{proof}

\begin{lemma}
\label{lem:schema-driven}
Let $\T$ be an $S$-driven \HornALCIF TBox that was obtained from $\widehat \T_S \cup \T_{\lnot Q}$ by reversing some finmod cycles. 
For every satisfiable finmod cycle \[K_1, R_1, \dots,  K_{n-1}, R_{n-1}, K_n\] in $\T$ there exist unique $A_1, A_2, \dots, A_n \in \Gamma_S$ such that $A_i\in K_i$ for all $i\leq n$, and \[A_1, R_1, \dots,  A_{n-1}, R_{n-1}, A_n\] is a finmod cycle in $\T$ 
\end{lemma}

\begin{proof}
Since all triples in $K_1, R_1, \dots,  K_{n-1}, R_{n-1}, K_n$ are  satisfiable, all CIs
$K_i \sqsubseteq \exists R_i.K_{i+1}$ and $K_{i+1} \sqsubseteq \exists^{\leq 1} R_i^-.K_i$ are relevant for $\T$.
We cannot simply apply the fact that $\T$ is $S$-driven, because these CIs need not belong to $\T$: they are only entailed by $\T$. The proof will proceed in several steps.  

The first step is to see that each $K_i$ contains a label from $\Gamma_S$. Towards contradiction, suppose it does not. We construct a graph witnessing that $\T$ does not entail $K_i \sqsubseteq \exists R_i.K_{i+1}$, which is a contradiction.
Let $T_i$ be the tree-shaped graph obtained by unravelling some model of $\T$ witnessing that $(K_i,R_i,K_{i+1})$ is satisfiable, from a node $u$ satisfying  $K_i$. Clearly, $T_i$ is also a model of $\T$, its root $u$ satisfies $K_i$ and has an $R_i$-successor $u'$ satisfying $K_{i+1}$.
We construct $G$ as the graph with a single node $u_0$ whose labels are copied from the root $u$ of $T_i$ but with any letter from $\Gamma_S$ dropped. To see that $G \not\models K_i \sqsubseteq \exists R_i.K_{i+1}$, note that as $u\in (K_i)^{T_i}$ and $K_i$ contains no labels from $\Gamma_S$, also $u_0\in (K_i)^G$; but clearly $u_0$ has no $R_{i}$-successors at all.
Let us check that $G\models\T$. 
\begin{itemize}
\item New CIs of the form $K \sqsubseteq A$ are not introduced by reversing cycles, so it suffices to look at ones from $\widehat \T_S \cup \T_{\lnot Q}$. There, such CIs are only present in $\T_{\lnot Q}$ and always satisfy $A\notin\Gamma_S$ (see the proof of Lemma~\ref{lem:rollingup}). Hence, as they were satisfied in $T_i$ and $G$ was obtained by dropping labels from $\Gamma_S$, they still hold in $G$.
\item CIs of the form $K \sqsubseteq \bot$ in $\T$ were satisfied in $T_i$ and they cannot be violated by dropping labels (recall that $K$ does not use negation). 
\item All CIs of the forms $K \sqsubseteq \forall R.K'$, $K \sqsubseteq \not\exists R.K'$, and $K \sqsubseteq \exists^{\leq 1}R.K$ are trivially satisfied in $G$. 
\item Consider a CI of the form $K \sqsubseteq \exists R.K'$ from $\T$. Suppose that $u_0 \in K^G$. Then also $u\in K^{T_i}$. This means that the CI was ``fired'' in $T_i$, which implies that $(K, R, K')$ is satisfiable modulo $\T$ and $K \sqsubseteq \exists R.K'$ is relevant for $\T$. As $\T$ is $S$-driven, it follows in particular that $K$ contains a label from $\Gamma_S$. But this contradicts the fact that $u_0 \in K^G$. Hence, $K \sqsubseteq \exists R.K'$ is trivially satisfied in $G$. 
\end{itemize}
Thus we have shown that $G\models\T$. This concludes the first step. 

Now, as all $K_i$ contain a label from $\Gamma_S$ and all triples $(K_i, R_i, K_{i+1})$ are satisfiable modulo $\T$, it follows that for each $i$ there exists exactly one label $A_i \in \Gamma_S$ such that $A_i \in K_i$. It remains to show that $A_{i} \sqsubseteq \exists R_i.A_{i+1}$ and $A_{i+1} \sqsubseteq \exists^{\leq 1} R^-_i.A_{i}$.

Let us begin with $A_{i} \sqsubseteq \exists R_i.A_{i+1}$. Consider graph $G$ obtained from $T_i$ (same as above) by removing all subtrees rooted at $R_i$-successors of the root that satisfy $K_{i+1}$. Clearly, $G \not\models K_{i} \sqsubseteq \exists R_i.K_{i+1}$. As $\T \models K_{i} \sqsubseteq \exists R_i.K_{i+1}$, it follows that $G \not\models \T$. Then, some CI of the form $K \sqsubseteq \exists R. K'$ from $\T$ is violated in $G$, because CIs of other forms are preserved when passing to a subgraph. In particular, it must be the case that the root of $G$ satisfies $K$. But then also the root of $T_i$ satisfies $K$ and since $T_i \models \T$, the root $u$ of $T_i$ has an $R$-successor $u'$ that satisfies $K'$. This means that $K \sqsubseteq \exists R. K'$ is relevant for $\T$. Because $\T$ is $S$-driven, it must contain $A \sqsubseteq \exists R. A'$ for some $A, A' \in \Gamma_S$ such that $A\in K$, $A'\in K'$. As the root of $G$ satisfies both $K$ and $K_i$, and we know that $A\in K$ and $A_i\in K_i$ and that labels from $\Gamma_S$ are exclusive, it follows that $A=A_i$. We claim that also $R=R_i$ and $A'=A_{i+1}$. If $R \neq R_i$, then $u'$ is not an $R_i$-successor of the root in $T_i$, and it has not been removed in $G$. That would imply that $G$ actually does satisfy $K \sqsubseteq \exists R. K'$.
Since we know this is not the case, we conclude that $R = R_i$. Similarly, suppose that $A'\neq A_{i+1}$. Because $u'$ satisfies $K'$ and $A' \in K'$, it must have label $A'$. But then $u'$ cannot have label $A_{i+1}$, which means it cannot satisfy $K_{i+1}$, and has not been removed in $G$. This yields a contradiction just like before and we can conclude that $A'= A_{i+1}$. Wrapping up, we have seen that $A \sqsubseteq \exists R. A'$ belongs to $\T$ and that  $A = A_i$, $R=R_i$, and $A'=A_{i+1}$. This means that $A_i \sqsubseteq \exists R_i. A_{i+1}$ belongs to $\T$.

Finally, let us see that $A_{i+1} \sqsubseteq \exists^{\leq 1} R^-_i.A_{i}$ belongs to $\T$. Consider the model $T_i$ but reorganize it so that the root $u$ satisfies $K_{i+1}$ and has an $R^-_{i}$-successor $u'$ satisfying $K_i$. Let $G$ be the graph obtained from $T_i$ by duplicating the whole subtree rooted at $u'$, and adding an $R^-_{i}$-edge from $u$ to the root $u''$ of the copy. Clearly $G \not\models K_{i+1} \sqsubseteq \exists^{\leq 1} R^-_i.K_{i}$ and since $\T \models K_{i+1} \sqsubseteq \exists^{\leq 1} R^-_i.K_{i}$, we conclude that  $G \not\models \T$. It follows immediately that $G$ violates some CI of the form $K \sqsubseteq \exists^{\leq 1} R. K'$ from $\T$, as CIs of other forms are not affected by the modification turning $T_i$ to $G$. Similarly, it must hold that $R = R^-_{i}$, and that $u$ satisfies $K$ and $u'$ and $u''$ satisfy $K'$. It follows that $K \sqsubseteq \exists^{\leq 1} R. K'$ is relevant, $A_{i+1} \in K$, $A_i \in K'$, and $ A_{i+1} \sqsubseteq \exists^{\leq 1} R^-_{i}.A_i$ belongs to $\T$.
\end{proof}

\begin{lemma}
\label{lem:completion}
For $\T = \widehat \T_S\cup\T_{\lnot Q}$, the completion $\T^*$ can be computed in \EXPTIME.
\end{lemma}
\begin{proof}
Construct a graph $G_\T$ over all possible intersections $K$ of concept names used in $\T$, including an edge with label $R\in \Sigma^\pm$ from $K$ to $K'$ iff 
\[\T \models K \sqsubseteq \exists R.K' \quad\text{and}\quad \T \models K' \sqsubseteq \exists^{\leq 1} R^-.K\,.\]
$G_\T$ has exponential size and can be constructed in \EXPTIME, because CI entailment by \HornALCIF TBoxes can be tested in exponential time \cite{GiacomoL96}. 
Repeat the following until the graph stops changing. Pick an $R$-edge from $K$ to $K'$ such that there is no $R^-$-edge from $K'$ to $K$. Check if there exists a path from $K'$ to $K$ in $G_\T$. If so, the identified path combined with the $R$-edge from $K$ to $K'$ constitutes a finmod cycle 
\[K_1, R_1, \dots, K_{n-1}, R_{n-1}, K_n\] 
in $\T$.
Add to $G_\T$ an $R_i^-$-edge from $K_{i+1}$ to $K_i$ for all $i<n$ and extend $\T$ with the corresponding concept inclusions. Note that this includes an $R^-$-edge from $K'$ to $K$ and concept inclusions 
\[K' \sqsubseteq \exists R^-.K \quad\text{and}\quad \T \models K \sqsubseteq \exists^{\leq 1} R.K'\,.\] 
Moreover, if there are unique $A_1,A_2, \dots, A_n \in \Gamma_S$ such that $A_i \in K_i$ for $i\leq n$, check if 
\[A_1, R_1, \dots, A_{n-1}, R_{n-1}, A_n\] is a cycle in $G_\T$. If so, add to $G$ an $R_i^-$ edge from $A_{i+1}$ to $A_{i}$, and the corresponding CIs to $\T$. By Lemma~\ref{lem:schema-driven}, this ensures that the extended $\T$ is $S$-driven. We can now reduce it and recompute $G_\T$ based on the updated $\T$. Using the complexity bounds for CI entailment given in Corollary~\ref{cor:entailment}, we conclude that this can be done in \EXPTIME. Note that we are indeed relying on the more precise complexity bounds here, because at later iterations of the cycle reversing procedure the TBox might well contain exponentially many concept inclusions. However, it has still only the original concept names and, after reducing, only a polynomial number of at-most restrictions. 
\end{proof}

\section{Proofs for Satisfiability}
\label{sec:app-satisfiability}

\subsection{Introductory lemmas}

We begin by showing the two lemmas mentioned in the body of the paper. 

\begin{lemma} \label{lem:sparse}
For $c\geq 1$, if a finite connected $c$-sparse graph has only nodes of degree at least 2, then it is $(2c,3c)$-skeleton.
\end{lemma}

\begin{proof}
Let $G$ be a finite connected $c$-sparse graph without nodes of degree 0 or 1.  We claim that $G$ consists of at most $2c$ nodes connected by at most $3c$ paths disjoint modulo endpoints. If $G$ is empty, we are done.  Otherwise, we eliminate vertices of degree 2 that are incident with two different edges by merging these edges into a single edge. This process results in a $c$-sparse multigraph $G_0$, whose edges represent simple paths in $G$. This graph is either a single node with a loop  or all its nodes have degree at least $3$. In the first case it follows that $G$ is a single cycle, and thus a $(1,1)$-skeleton. In the second case, assuming that $G_0$ has $n$ nodes and $m$ edges, we have $3n/2 \leq m \leq n + c$. It follows that $c>0$, $n \leq 2c$,  $m\leq 3 c$.
\end{proof}

\begin{lemma}
\label{lem:skeleton}
If $p$ is satisfied in a $|p|$-sparse graph $G$, then $G$ contains a $(4|p|, 5|p|)$-skeleton $H$, extending the skeleton of $G$, such that all variables of $p$ are mapped to distinguished nodes of $H$ and $G$ can be obtained by attaching finitely many finitely branching trees to $H$. 
\end{lemma}
\begin{proof} The skeleton $H_0$ of $G$ is a $(2|p|, 3|p|)$-skeleton. Consider a match of $p$ in $G$. Some variables of $p$ might well be matched to nodes on the paths connecting the distinguished nodes of $H_0$ or in the attached trees. We define $H$ as follows. First, we add to $H$ as distinguished nodes all images of variables of $p$ that lie on the paths connecting distinguished nodes of $H_0$. Next, for each attached tree $T$ that contains an image of a variable of $p$, we add to $H$ as distinguished nodes all the images of variables of $p$ that belong to $T$ together with all their least common ancestors in $T$, as well as the node of $H$ to which the root of $T$ is connected. All ancestors (in $T$) of these nodes are added to $H$ as ordinary nodes. The skeleton $H$ thus obtained has the required properties. 
\end{proof}

\subsection{The main result}

The goal of this section is to prove the following theorem. 

\begin{theorem} \label{thm:sparse-witness}
Given a C2RPQ $p$ and an \ALCIF TBox $\T$ using $k$ concept names and $\ell$ at-most constraints, one can decide in time $O\big(\mathrm{poly}(|\T|)\cdot 2^{\mathrm{poly}\left(|p|,k,\ell\right)}\big)$ if there exists a $|p|$-sparse graph that satisfies $p$ and $\T$.
\end{theorem}

\newcommand{\ett}{\ensuremath{\rotatebox{90}{\scalebox{.8}[1.6]{${\curvearrowright}$}}}}
\newcommand{\ess}{\ensuremath{\rotatebox{90}{\scalebox{.8}[-1.6]{${\curvearrowright}$}}}}
\newcommand{\ets}{\ensuremath{\leftarrow}}
\newcommand{\est}{\ensuremath{\rightarrow}}
\newcommand{\eabove}{\ensuremath{\raisebox{0ex}{\scalebox{.8}[1.6]{${\curvearrowright}$}}}}
\newcommand{\ebelow}{\ensuremath{\raisebox{1ex}{\scalebox{.8}[-1.6]{${\curvearrowright}$}}}}
\newcommand{\eup}{\ensuremath{\uparrow}}
\newcommand{\edown}{\ensuremath{\downarrow}}

\newcommand{\dir}{\ensuremath{\textrm{dir}}}
\newcommand{\ft}{\ensuremath{f_{\textrm{labs}}}}
\newcommand{\fpz}{\ensuremath{f_{\textrm{vars}}}}

\newcommand{\fpn}{\ensuremath{\delta_{\textrm{node}}}}
\newcommand{\fpe}{\ensuremath{\delta_{\textrm{edge}}}}

\newcommand{\fws}{\ensuremath{\beta_{\textrm{src}}}}
\newcommand{\fwt}{\ensuremath{\beta_{\textrm{tgt}}}}

The proof of Theorem~\ref{thm:sparse-witness} is not very hard, but it combines several components and requires developing some machinery. Let us begin with a road map. 

Relying on Lemma~\ref{lem:skeleton}, we guess a $(4|p|, 5|p|)$-skeleton $H$. The distinguished nodes of $H$ are represented explicitly, together with all their labels,  but each of the connecting paths is represented by a single \emph{symbolic edge}. Note that there might be multiple symbolic edges between the same pair of distinguished nodes, representing different paths.  We need to check that $H$ can be completed to a graph $G$ by materializing the symbolic edges into paths and attaching finitely many finitely branching trees in such a way that $G$ is a model of $\T$ and there is a match of $p$ in $G$ that maps variables of $p$ to distinguished nodes of $H$. 

To achieve this, we guess an annotation of $H$ that summarizes how the witnessing paths of $p$ can traverse the parts of $G$ missing from $H$, and which witnesses of distinguished nodes required by $\T$ these parts provide (Section \ref{ssec:annotated}). We then check if these promises of the annotation are sufficient to guarantee that $p$ and $\T$ are satisfied (Section~\ref{ssec:sufficient}). Finally, we verify that the promises of the annotation can be fulfilled: we check if we can attach trees to the distinguished nodes and expand the symbolic edges into finite paths with attached trees in a way that matches the promises of the annotation and respects the TBox $\T$ (Section~\ref{ssec:implementing}). 

\subsection{Annotated skeleta}
\label{ssec:annotated}

Let $\Gamma_p$, $\Sigma_p$, $\Gamma_\T$, $\Sigma_\T$ be the sets of edge and node labels used in $p$ and $\T$, respectively. In what follows we only consider graphs and skeleta using only node labels from $\Gamma_p \cup \Gamma_\T$ and edge labels from $\Sigma_p \cup \Sigma_\T$.

Let $\Phi$ be the set of two-way regular expressions used in $p$. For each $\varphi \in \Phi$ we fix an equivalent linear size non-deterministic automaton $\A_\varphi$ over the alphabet $\Gamma_p\cup\Sigma^\pm_p$ with states $K_\varphi$, initial states $I_\varphi \subseteq K_\varphi$, and final states $F_\varphi\subseteq K_\varphi$. We assume that all $K_\varphi$ are pairwise disjoint and let $\delta = \bigcup_{\varphi\in\Phi} \delta_{\varphi}$.

An \emph{annotation} of skeleton $H$ is given by the following functions.
\begin{itemize}
\item $\fws$ and $\fwt$ record  information about the source and target of the paths represented by each symbolic edge:
they both map each symbolic edge $e$ to $\big(\Sigma_p\cup\Sigma_\T\big)^\pm \times 2^{\Gamma_p\cup \Gamma_\T}$. 
\item $\fpn$ records how the witnessing paths for $p$ may loop in the subtrees attached to the distinguished nodes.
Thus, $\fpn$ maps  every distinguished node to a subset of $\bigcup_{\varphi\in\Phi}  K_\varphi \times K_\varphi$.
\item  $\fpe$ records how the witnessing paths for $p$ progress along paths (and the trees attached to them) represented by the symbolic edges in the skeleton.
Thus, $\fpe$ maps every edge $e$ to a subset of 
$\bigcup_{\varphi\in\Phi}  K_\varphi \times K_\varphi \times \{\ess,\ett,\ets,\est\}$. If $e$ is an edge from $u$ to $v$, then $(s,s',\est) \in \fpe(e)$ indicates that some path enters (the part of the model summarized by) the edge $e$ from $u$ in state $s$, and exits at node $v$ in state $s'$. Similarly,  $(s,s',\ett) \in \fpe(e)$ indicates a loop: some path enters $e$ from $v$ in state $s$, and exits at the same node $v$ in state $s'$, etc.
\end{itemize}

\subsection{Verifying annotated skeleta}
\label{ssec:sufficient}

An annotation of $H$ is \emph{sufficient for TBox} $\T$ if the witnesses recorded by $\fws$ and $\fwt$ respect $\T$; that is, for each distinguished node $u$ of $H$ the graph $G_u$ defined below satisfies the TBox $\T_0$ obtained from $\T$ by dropping all concept inclusions of the form $A \sqsubseteq \exists R.B$.
To construct $G_u$ we begin from $u$ with labels inherited from $H$, and then for each symbolic edge $e$ incident with $u$ we add an $R$-successor $v_e$ of $u$ with label set $\Lambda$, where $(R,\Lambda) = \fws(e)$ if $u$ is the source of $e$ and $(R,\Lambda) = \fwt(e)$ if $u$ is the target of $e$.

An annotation is \emph{sufficient for C2RPQ} $p$ if there exists a function $\eta$ mapping variables of $p$ to distinguished nodes of $H$ such that for each atom $\varphi(x,y)$ of $p$, there exists a finite witnessing sequence $s_0 u_0 s_1 u_1 \dots s_k u_k$ of states and distinguished nodes of $H$ satisfying the following conditions.
\begin{itemize}
\item The witnessing sequence begins and ends correctly; that is, $s_0 \in I_\varphi$, $s_k \in F_\varphi$, $u_0 =\eta(x)$, $u_{k} =\eta(y)$.
\item Each transition step along a symbolic edge (or subtree attached to a distinguished node) updates the state as expected: for each $i< k$ one of the following holds:
\begin{itemize}
    \item $(s_i, s_{i+1},\est) \in \fpe(e)$ for some edge $e$ from $u_i$ to $u_{i+1}$;
    \item $(s_i, s_{i+1},\ets) \in \fpe(e)$ for some edge $e$ from $u_{i+1}$ to $u_{i}$;
    \item $(s_i, s_{i+1},\ess) \in \fpe(e)$ for some edge $e$ from $u_{i}$ to some $u$, and $u_i=u_{i+1}$;
    \item $(s_i, s_{i+1},\ett) \in \fpe(e)$ for some edge $e$ from some $u$ to $u_{i}$, and $u_i=u_{i+1}$;
    \item $(s_i, s_{i+1}) \in \fpn(u_i)$ and $u_i=u_{i+1}$. 
\end{itemize}
\end{itemize}
We point out that the witnessing sequence may traverse a symbolic edge  multiple times. In consequence, each tuple in $\fpe(e)$ must be ``realised'' by the single path represented by $e$ (and the attached trees).

\begin{proposition}
One can decide if a given annotated skeleton is sufficient for $p$ and $\T$ in $\PTIME$.
\end{proposition}
\begin{proof}
To check that the annotated skeleton is sufficient for $\T$ it is enough to examine the graphs $G_u$ for each distinguished node $u$ of the skeleton.

Checking that the annotated skeleton is sufficient for $p$ amounts to guessing the function $\eta$ and for each atom $\varphi(x,y)$ running a reachability test in the product graph whose nodes combine distinguished nodes of the skeleton with states from $K_\varphi$, where edges are defined according to the symbolic edges in the skeleton and the triples from $\fpe$. In the reachability test we check if there exists a path beginning in $\{\eta(x)\}\times I_\varphi$ and ending in $\{\eta(y)\}\times F_\varphi$.
\end{proof}

\subsection{Implementing annotated skeleta}
\label{ssec:implementing}

Consider an annotated skeleton $\mathcal{H}=\big(H, \fws, \fwt, \fpe, \fpn\big)$. 
We say that a graph $G$ \emph{implements} $\mathcal{H}$ if $G$ is obtained from $H$ by replacing each symbolic edge $e$ with a path $\pi_e$ connecting the endpoints of $e$ and by attaching finitely many finitely branching trees in a way consistent with the annotations, in the following sense.
\begin{itemize}
\item 
For each symbolic edge $e$ from $u$ to $u'$, the subgraph $G_e$ of $G$ that consists of $\pi_e$ and all trees attached to the internal nodes of $\pi_e$ is correctly summarized in the annotations:
\begin{itemize}
    \item for each $(s,s',d)\in \fpe(e)$ with $s,s'\in K_\varphi$ there is a path in $G_e$ with endpoints $(u,u)$ if $d=\ess$, $(u,u')$ if $d=\,\est\,$, $(u',u')$ if $d=\ett\,$, and $(u',u)$ if $d=\,\ets\,$, on which $\A_\varphi$ moves from state $s$ to state $s'$; 
    \item if $\fws(e) = (R_1, \Lambda_1)$ and $\fws(e) = (R_2, \Lambda_2)$, then the first edge of $\pi_e$ is an $R$-edge, the last edge of $\pi_e$ is an $R_2^-$-edge, the second node on $\pi_e$ has the labels set $\Lambda_1$, and the penultimate node on $\pi_e$ has label set $\Lambda_2$.
\end{itemize}

\item For each distinguished node $u$, the trees attached to $u$ are summarized correctly in the annotations: for each $(s,s')\in \fpn(u)$ with $s,s'\in K_\varphi$ there is a tree $T_{u}^{s,s'}$ attached to $u$ and a path that starts and ends in $u$ and otherwise only visits nodes of $T_{u}^{s,s'}$, on which $\A_\varphi$ moves from state $s$ to state $s'$.

\item $G$ is a model of $\T$.
\end{itemize}

Note that all the missing pieces of the graph are essentially trees (finitely branching, but typically infinite). Indeed, each $T_u^{s,s'}$ simply is a tree, but also $G_e$ can be viewed as a tree: its  root is the source of $e$, the root has exactly one child,  the path $\pi_e$ constitutes a special finite branch ending in the target of $e$ which is a leaf in this tree. Importantly, each $(s,s')\in\fpn(u)$ is witnessed by a finite subgraph of $T^{s,s'}_u$, and each triple $(s,s',d)\in \fpe(e)$ is witnessed by a finite subgraph of $G_e$. The algorithm to check if there exist such $T^{s,s'}_u$ and $G_e$ can be seen as an emptiness test for tree automaton, or as a variant of type elimination.

We first define \emph{types}, which can also be viewed as states of a tree automaton.
We assign to each node of the tree a type that records the following information:
\begin{itemize}
\item a subset of $\Gamma_p\cup\Gamma_\T$, representing the labels of the current node;
\item an element of $\Sigma_p^\pm \cup \Sigma_\T^\pm$ and a subset of $\Gamma_p\cup\Gamma_\T$, representing the label on the edge to the parent and the parent's label set;
\item with $\ell$ the number of at-most restrictions in $\T$, a list of $t\leq\ell+1$ elements of $\Sigma_p^\pm \cup \Sigma_\T^\pm$ and subsets of $\Gamma_p\cup\Gamma_\T$, representing labels on the edges to $t$ children of the current node and the children's label sets;  
\item a Boolean flag indicating whether the current node belongs to the special path (not used for $T_u^{s,s'}$ at all);
\item a subset of $\bigcup_{\varphi\in \Phi} K_\varphi\times K_\varphi \times \{\ebelow,\eabove,\edown,\eup\}$ recording the progress on witnessing $\fpe$ or $\fpn$: 
\begin{itemize}
\item $(s,s',\ebelow)$ indicates  that from state $s$ in the current node we can navigate the current subtree and return to the current node in state $s'$,
\item $(s,s',\eabove)$ indicates  that from state $s$ in the current node we can navigate outside of the current subtree and return to the current node in state $s'$,
\item  $(s,s',\edown)$ indicates that from state $s$ in the current node, we can reach the target node of $e$ in state $s'$,
\item $(s,s',\eup)$ indicates that from state $s$ in target node of $e$ we can reach the current node in state $s'$.
\end{itemize}
Actually, all four kinds of triples are required along the special path, but in the remaining nodes we only need the triples of the form $(s,s',\ebelow)$.
\end{itemize}

By a \emph{pre-type} we shall understand a type with the boolean flag and the progress information dropped; that is, a tuple \[(\Lambda, R', \Lambda', R_1, \Lambda_1,  \dots, R_t, \Lambda_t )\] with $\Lambda, \Lambda', \Lambda_1, \dots, \Lambda_t \subseteq \Gamma_p\cup\Gamma_\T$,  and $R', R_1, \dots, R_t \in \Sigma_p^\pm \cup\Sigma_\T^\pm$, and  $0 \leq t \leq \ell+1$.
In what follows we blur the distinction between conjunctions $K$ of concept names and sets $\Lambda$ of labels, as usual, and write $K \subseteq \Lambda$.

A pre-type $(\Lambda, R', \Lambda', R_1, \Lambda_1, \dots, R_t, \Lambda_t )$ is  \emph{compatible} with $\T$ iff there exists a graph $G$ such that 
\begin{itemize}
    \item there are pairwise different nodes $u, u', u_1, \dots, u_t$ with label sets $\Lambda, \Lambda', \Lambda_1, \dots, \Lambda_t$;
    \item there is an $R'$-edge from $u$ to $u'$ and an $R_i$-edge from $u$ to $u_i$ for all $i\leq t$, and no other edges are incident with $u'$; 
    \item for each $K \sqsubseteq \exists^{\leq 1} R.K'$ in $\T$ with $K\subseteq \Lambda$, every $R$-successor of $u$ that satisfies $K'$ belongs to $\{u', u_1, \dots, u_t\}$; and
    \item $G$ satisfies  $\T$ except that CIs of the form $K \sqsubseteq \exists R.K'$ are not required to be satisfied for $u'$.
\end{itemize}
Note that unlike in the notion of satisfiability used in Appendix~\ref{sec:app-containment}, the witnessing nodes cannot have additional labels, not listed in  $\Lambda, \Lambda', \Lambda_1, \dots, \Lambda_t$. 

\begin{lemma}
\label{lem:compatible}
Given $\T$ and $p$ one can compute the set of pre-types compatible with $\T$ within the time bound stated in Theorem~\ref{thm:sparse-witness}
\end{lemma}
\begin{proof}
Each pre-type $(\Lambda, R', \Lambda', R_1, \Lambda_1, \dots, R_t, \Lambda_t )$ can be interpreted as a star-shaped graph consisting of nodes $u,u', u_1, \dots, u_n$ with label sets $\Lambda, \Lambda', \Lambda_1, \dots, \Lambda_t$ such that $u'$ is an $R'$-successor of $u$, $u_i$ is an $R_i$-successor of $u$ for all $i\leq t$, and there are no other edges. Thus we can speak of a pre-type satisfying a concept inclusion, etc. 

We say a pre-type $(\Lambda, R', \Lambda', R_1, \Lambda_1, \dots, R_t, \Lambda_t )$ 
is \emph{repeatable} if there is no at-most restriction $K\sqsubseteq \exists^{\leq 1} R.K'$ in $\T$ such that $K\subseteq \Lambda'$, $R=(R')^-$, and $K'\subseteq \Lambda$.

A pre-type $(\Lambda, R', \Lambda', R_1, \Lambda_1, \dots, R_t, \Lambda_t )$ is said to be \emph{compatible with $\T$ modulo} a set $\Theta$ of pre-types if 
\begin{itemize}
    \item $\Theta$ contains a pre-type  $(\Lambda_i, R_i^-, \Lambda, \dots)$ for each $i\leq t$;
    \item the pre-type satisfies all CIs in $\T$ not of the form $K\sqsubseteq\exists R.K'$;
    \item for each concept inclusion $K \sqsubseteq \exists R. K'$ in $\T$ with $K\subseteq \Lambda$, at least one of the following holds:
    \begin{itemize}
        \item $R=R'$ and $K' \subseteq \Lambda'$, or
        \item $R=R_i$ and $K' \subseteq \Lambda_i$ for some $1\leq i\leq t$, or
        \item $R=R_0$ and $K' \subseteq \Lambda_0$ for some repeatable  $(\Lambda_0, R_0^-, \Lambda, \dots)$ from  $\Theta$.
    \end{itemize}
\end{itemize} 

Now, to compute the set of pre-types compatible with $\T$, we start with the set $\Theta=\Theta_0$ of all pre-types, and exhaustively remove those pre-types that are not compatible with $\T$ modulo $\Theta$. This algorithm terminates after at most 
\[|\Theta_0| = \sum_{t=0}^{\ell+1} \big|\Sigma_p^\pm\cup\Sigma^\pm_\T\big|^{t+1} \cdot \Big(2^{|\Gamma_p \cup\Gamma_\T|}\Big)^{t+2}\] iterations. Each iteration takes time polynomial in $|\Theta|$ and $|\T|$. 

The result is the maximum set $\Theta$ of pre-types such that each pre-type from $\Theta$ is compatible with $\T$ modulo $\Theta$. Each pre-type compatible with $\T$ will belong to this set, because the graph witnessing the triple can be used to argue that the triple will not be removed at any iteration. Conversely, each triple from $\Theta$ is compatible with $\T$, because one can construct a witnessing tree-shaped graph top-down, using the witnesses justifying the presence of pre-types in $\Theta$ in the last iteration of the algorithm. 
\end{proof}

\begin{lemma}
The existence of a graph implementing a given annotated skeleton is decidable within the time bound from Theorem~\ref{thm:sparse-witness}.
\end{lemma}

\begin{proof}
We call a type $(\Lambda, R', \Lambda', R_1, \Lambda_1, \dots, R_t, \Lambda_t, b, \Delta)$ \emph{compatible with $\T$} if the underlying pre-type $(\Lambda, R', \Lambda', R_1, \Lambda_1, \dots, R_t, \Lambda_t)$ is compatible with $\T$. \emph{Repeatable types} are defined analogously, based on the underlying pre-types.  Clearly, Lemma~\ref{lem:compatible} suffices to precompute the set of types compatible with  $\T$. Our task is to check if from these types one can construct the witnessing $G_e$ and $T_u^{s,s'}$. We will build them bottom-up, guaranteeing that each promise related to $p$ is fulfilled in a finite fragment.

A type $(\Lambda, R', \Lambda', R_1, \Lambda_1, \dots, R_t, \Lambda_t, b, \Delta)$ is \emph{compatible with $p$ modulo} a set $\Theta$ of types if there exists types $(\Lambda_i, R_i^-, \Lambda, \dots, b_i, \Delta_i) \in \Theta$ for $ 1 \leq i \leq t$ such that
\begin{itemize}

\item if $b=0$, then $b_i = 0$ for all $1 \leq i \leq t$, else $t\geq 1$, $b_1 = 1$, and $b_i = 0$ for all $1 < i \leq t$;

\item for each $(s,s',\ebelow) \in \Delta$,
\begin{itemize}
    \item $(s, A, s') \in \delta$ for some $A\in\Lambda$, or
    \item $(s,R_i,s_1)\in \delta$, $(s_1,s_2,\ebelow)\in \Delta^*_i$, and $(s_2,R_i^-,s') \in \delta$ for some $s_1,s_2$ and $1 \leq i \leq t$, or
    \item $(s,R_0,s_1)\in\delta$, $(s_1,s_2,\ebelow)\in \Delta^*_0$, and $(s_2,R_0^-,s') \in \delta$  for some $s_1,s_2$ and repeatable $(\Lambda_0, R_0^-, \Lambda, \dots, 0, \Delta_0) \in \Theta$, 
\end{itemize}
where $\Delta_i^*$ is the set of all $\big(s,s',\ebelow\big)$ such that there are states $s=s_1, s_2, \dots, s_m=s'$ with $\big(s_j,s_{j+1},\ebelow\big)\in \Delta_i$ for all $j < m$;

\item if $b=1$, then for each $(s,s',\eup)\in \Delta$, there are  $s_1,s_2$ such that $(s,s_1,\eup) \in \Delta_1$, 
$(s_1,s_2,\ebelow) \in\Delta_1^*$, and $(s_2,R_1^-,s') \in \delta$;

\item if $b=1$, then for each $(s,s',\edown)\in \Delta$, there are  $s_1,s_2$ such that $(s,R_1,s_1) \in \delta$, 
$(s_1,s_2,\ebelow) \in\Delta_1^*$, and $(s_2,s',\edown) \in \Delta_1$;

\item if $b=1$, then for each $(s,s',\eabove)\in \Delta_1$, there are $s_1,\dots, s_m$ such that $(s,R_1^-,s_1), (s_m,R_1,s') \in \delta$ and for all $j<m$, either  $(s_j,s_{j+1},\ebelow) \in \Delta_2^* \cup \dots \cup\Delta_m^*$, or $(s_j,s_{j+1},\eabove) \in \Delta$, or $(s_j,A,S_{j+1}) \in \delta$ for some  $A\in\Lambda$.
\end{itemize}

Let us first see how to decide the existence of $G_e$ for a  given symbolic edge $e$. The algorithm begins with the set $\Theta$ of all ``initial types'', which are
\begin{itemize}
    \item types $(\Lambda, R', \Lambda', b,\Delta)$ such that 
    \begin{itemize}
        \item $\Lambda$ is the label set of the target of $e$,
        \item $(R',\Lambda') = \fws(e)$,
        \item $b=1$,
        \item $\Delta$ consists of all $(s,s',\eabove)$ such that $(s,s',\ett) \in \fpe(e)$, as well as all $(s,s,\eup)$ and $(s,s,\edown)$;
    \end{itemize}
    \item types $(\Lambda, R', \Lambda',\dots, b,\Delta)$ compatible with $\T$ such that 
    \begin{itemize}
        \item $b=0$,
        \item $\Delta = \emptyset$.
    \end{itemize}
\end{itemize}
Then, we exhaustively extend $\Theta$ with types that are compatible with $\T$ and compatible with $p$ modulo $\Theta$.
When no more types can be added, the graph $G_e$ exists iff  $\Theta$ contains a type $(\Lambda, R', \Lambda', \dots,  b, \Delta)$ such that
\begin{itemize}
    \item $\Lambda$ is the label set of the source of the symbolic edge $e$;
    \item $\big((R')^-,\Lambda\big) = \fws(e)$;
    \item $b=1$;
    \item $\Delta$ contains no triples of the form $(s,s',\eabove)$;
    \item for each $(s,s',\ebelow)\in \fpe(e)$ there are states $s_1,s_2$ such that $\big(s, (R')^-, s_1\big) \in \delta$, $(s_1,s_2,\ebelow)\in\Delta^*$, and $\big(s_2, (R')^-, s'\big) \in \delta$;
    \item for each $(s,s',\edown)\in \fpe(e)$ there are states $s_1, s_2$ such that $\big(s, (R')^-, s_1\big) \in \delta$, $(s_1,s_2,\ebelow)\in\Delta^*$, and $(s_2,s',\edown)\in\Delta$;
    \item for each $(s,s',\eup)\in \fpe(e)$ there are states $s_1, s_2$ such that $(s, s_1, \eup) \in \Delta$, $(s_1,s_2,\ebelow)\in\Delta^*$, and $(s_2, R', s')\in\delta$.
\end{itemize}

This number of iterations of the algorithm is bounded by the number of all types, 
\[\sum_{t=0}^{\ell+1} \big|\Sigma_p^\pm\cup\Sigma^\pm_\T\big|^{t+1} \cdot \Big(2^{|\Gamma_p \cup\Gamma_\T|}\Big)^{t+2} \cdot 2 \cdot \left(2^{\left|\bigcup_{\varphi \in \Phi} K_\varphi\times K_\varphi \times \{\ebelow, \eabove, \edown, \eup\}\right|}\right)^{t}.\] 
Each iteration takes time polynomial in $|\Theta|^\ell$ and $|\T|$. 
 The promised complexity bounds follow.

Deciding the existence of the witnessing trees for a node $u$ of the annotated skeleton is very similar. We can reuse the set $\Theta$ computed for any symbolic edge $e$.
The only delicate issue is that we need to account for  $\fws(e')$ for all edges $e'$ outgoing from $u$ and $\fwt(e'')$ for all edges $e''$ incoming to $u$. Essentially, we check if there exists a type $(\Lambda, R_1, \Lambda_1, \ldots,R_t, b, \Delta)$ -- note the missing $R'$ and $\Lambda'$ -- with $b=0$ and $t \leq \ell+\deg(u)$, compatible with $\T$ and compatible with $p$ modulo $\Theta$, except that for $i=1,2, \dots, \deg(u)$, the components $R_i, \Lambda_i$ must be as specified by $\fws(e')$ and $\fwt(e'')$ for outgoing $e'$ and incoming $e''$, and their corresponding types must be  $(\Lambda_i, R_i^-, \Lambda, 0, \emptyset)$, not required to  belong to $\Theta$. This can be done in time polynomial in $|\Theta|^\ell$, $\T$, and $\mathcal{H}$.
\end{proof}

\begin{corollary}
\label{cor:entailment}
Unrestricted entailment of concept inclusions by an \ALCIF TBox $\T$ using $k$ concept names and $\ell$ at-most constraints can be decided in time $O\big(\mathrm{poly}(|\T|)\cdot 2^{\mathrm{poly}\left(k,\ell\right)}\big)$. 
\end{corollary}

\begin{proof}
The result holds in full generality, but we only sketch the arguments for the two kinds of concept inclusions we need to compute the completion. For existential constraints, note that \[\T \models  A_1  \sqcap \dots \sqcap A_n \sqsubseteq \exists R. K'\] iff the query \[\exists x. (A_1 \cdot \ldots \cdot A_n \cdot B) (x,x) \] is unsatisfiable modulo the TBox \[\T \cup \big\{K' \sqsubseteq \forall R^-. B'\,,\; B\sqcap B' \sqsubseteq \bot\big\}\,,\] where $B$ and $B'$ are fresh concept names. 
For at-most constraints,  \[\T \models  A_1 \sqcap  \dots \sqcap A_n \sqsubseteq \exists^{\leq 1} R. A'_1 \sqcap  \dots \sqcap A'_m\] iff the query 
\begin{align*}
\exists x, y, z. (A_1 \cdot \ldots \cdot A_n) (x,x) 
\land & (R\cdot A'_1\cdot \ldots \cdot A'_m\cdot B)(x,y)\land \\ \land & (R\cdot A'_1\cdot \ldots \cdot A'_m\cdot B')(x,z)
\end{align*} is unsatisfiable modulo the TBox \[\T \sqcup \big\{B \sqcap B' \sqsubseteq \bot\big\}\] where $B$ and $B'$ are fresh concept names.
\end{proof}

\section{Proof of Hardness}
\label{sec:app-hardness}

\newcommand{\Config}{\mathit{Config}}
\newcommand{\Pos}{\mathit{Pos}}
\newcommand{\pos}{\mathit{pos}}
\newcommand{\Accept}{\mathit{Accept}}
\newcommand{\Start}{\mathit{Start}}
\newcommand{\Symbol}{\mathit{Symbol}}
\newcommand{\State}{\mathit{State}}

\begin{theorem}
  \label{theorem:query-containment-hardness}
  Testing containment of Boolean 2RPQs modulo schema is \EXPTIME-hard.
\end{theorem}
We present a reduction of the acceptance problem of an alternating Turing
machine with a polynomial bound on space. We begin by defining a special variant
of alternating Turing machines. We also present a number of conceptual tools
used in the reduction.

\paragraph{Alternating Turing machines}
We consider a variant of alternating Turing machine with the following particularities:
\begin{itemize}
\item there is a single distinguished initial state that the machine never reenters;
\item there are two special states $q_\mathit{yes}$ and $q_\mathit{no}$ that are
  final (no transition allowed to follow)\footnote{The state $q_\mathit{no}$ is
    not necessary for the purposes of our reduction but we include it for the
    sake of completeness of this variant of ATM};
\item the transition table has exactly two transitions for any non-final state
  and any symbol;
\item there exists 3 special symbols: $\square$ for empty tape space, $\rhd$ for
  left tape boundary, and $\lhd$ for right tape boundary; we only assume that
  the input word does not use those symbols and the transition table handles the
  boundary symbols appropriately.
\end{itemize}
It's relatively easy to see that any alternating Turing machine with
polynomially bounded space can converted to the variant above. 

Formally, an \emph{alternating Turing machine} (ATM) is a tuple
$M=(A,K,q_0,\delta_1,\delta_2)$, where $A$ is a finite alphabet, $K$ is a finite
set of states with two distinguished final states $q_\mathit{yes}$ and
$q_\mathit{no}$ and partitioned into three pair-wise disjoint subsets
$ K=K_\forall\cup K_\exists\cup\{q_\mathit{yes},q_\mathit{no}\}$, $q_0\in K$ is
a distinguished initial state, and
$\delta_i: (K\setminus\{q_\mathit{yes},q_\mathit{no}\})\times A\rightarrow (
K\setminus\{q_0\}) \times A \times \{\mathord{-1},\mathord{+1}\}$ are two
transition tables such that $\delta_i(q,x)=(q',y,d)$ satisfies the following two
conditions:
\begin{enumerate}
\item if $x=\rhd$, then $y=\rhd$ and $d=\mathord{+1}$ and 
\item if $x=\lhd$, then $y=\lhd$ and $d=\mathord{-1}$.
\end{enumerate}
A \emph{configuration} of $M$ is a string of the form
$\rhd \cdot w \cdot q \cdot v \cdot \lhd$, where $q\in K$ and $w,v\in\Sigma^*$
Applying a transition
$(q',z,d)\in  K\times A \times\{\mathord{-1},\mathord{+1}\}$ to
the configuration $\rhd \cdot w \cdot x \cdot q \cdot y \cdot v \cdot \lhd$
yields:
\begin{enumerate}
\item $\rhd \cdot w \cdot q' \cdot x \cdot z \cdot v \cdot \lhd$ if
  $d=\mathord{-1}$
\item $\rhd \cdot w \cdot x \cdot z \cdot q' \cdot v \cdot \lhd$ if
  $d=\mathord{+1}$
\end{enumerate}
We consider ATMs with polynomially bounded space, a class of Turing machines
that defines the class $\ASPACE$ known to coincide with $\EXPTIME$. Recall that
a binary tree is a finite prefix-closed subset $T\subseteq \{1,2\}^*$ and a
labeled-tree is a function $\lambda$ that assigns a label to every element (node) of a
tree.

Given an ATM $M$ and a polynomial $\mathit{poly}(n)$, a \emph{run} of $M$
w.r.t.\ $\mathit{poly}$ on an input
$w\in (\Sigma\setminus\{\rhd,\lhd,\square\})^*$ is a binary tree $\lambda$ whose
nodes are labeled with configurations of $M$ such that:
\begin{enumerate}
\item the root node is labeled with
  $\lambda(\varepsilon)=\rhd\cdot q_0 \cdot w \cdot
  \square^{\mathit{poly}(|w|)-|w|}\cdot\lhd$
\item for non-leaf node $n\in\dom(\lambda)$ let
  $\lambda(n)=\rhd\cdot w \cdot q \cdot x \cdot v \cdot\lhd$; for every
  $i\in\{1,2\}$ if $n$ has a child $n\cdot i$, then the configuration
  $\lambda(n\cdot i)$ is obtained by applying the transition $\delta_i(q,x)$ to
  the configuration $\lambda(n)$. Also, if $q\in  K_\forall$, then $n$ has both
  children $n\cdot1$ and $n\cdot2$ and if $q\in  K_\exists$, then $n$ has
  precisely one child,
\item for every leaf node $n\in\dom(\lambda)$ the configuration $\lambda(n)$ uses a
  final state $q_{\mathit{yes}}$ or $q_{\mathit{no}}$.
\end{enumerate}
A run is \emph{accepting} if and only if all its leaves  use the state
$q_{\mathit{yes}}$. The ATM $M$ (with space bound $\mathit{poly}$) accepts a
word $w$, in symbols $M(w)=\mathit{yes}$ if and only if there is an accepting run of $M$
w.r.t.\ $\mathit{poly}$ on $w$. 

\paragraph{Reduction outline}
We present a reduction of the problem of word acceptance by an ATM with
polynomial bound on space to the complement of the problem of containment of
Boolean 2RPQs in the presence of schema. We point out that the class of
\ASPACE-complete problems is closed under complement, and consequently, this
reduction proves that the query containment problem is \EXPTIME-hard.

More precisely, for an ATM $M$, whose space is bounded by $\mathit{poly}(n)$,
and an input word $w$ we construct a schema $S$ and two Boolean 2RPQs $p$
and $q$ such that
\[
  M(w) = \mathit{yes} \iff p \nsubseteq_S q \iff
  \exists G \in L(S).\ G \models p
  \land G \not\models q\,.
\]
In the sequel, we refer to $p$ as the positive query and to $q$ as the
negative query. Naturally, we present a reduction that is polynomial i.e., the
combined size of $p$, $q$, and $S$ is bounded by polynomial in the size of $M$
and $w$.

The reduction constructs a schema $S$ and queries $p$ and $q$ for which the
counter-example of $p\subseteq_S q$ represents an accepting run of $M$ on
$w$. Before we present the reduction in detail, we introduce 3 conceptual
devices that we use in the reduction: nesting queries, encoding disjunction, and
enforcing tree structure.

\paragraph{Nesting queries}
The reduction employs a relatively large and complex queries and throughout the
reduction we employ \emph{nesting of regular path queries} that is expanded as
follows:
\[
  p[q] = p\cdot q \cdot q^-
\]
with the inverse operator being extended to regular path queries in the standard
fashion.
\begin{align*}
  \varnothing^- = {} & \varnothing\,, &
  \epsilon^-  = {} &\epsilon\,,&
  A^- = {} & A\,,\\
  (\varphi_1 \cdot \varphi_2)^- = {} & \varphi_2^- \cdot \varphi_1^-\,,&
  (\varphi_1 + \varphi_2)^- = {}& \varphi_1^- + \varphi_2^-\,,&
  (\varphi^*)^- = {} & (\varphi^-)^*\,.&
\end{align*}
We point out that, in general, this definition is not equivalent to the standard
meaning of nesting of regular expressions but in our reduction nested queries
are evaluated at nodes for which the schema ensures the intended meaning.

\paragraph{Encoding disjunction}
The first conceptual device allows us to express disjunction in schemas, which
we illustrate on the following example. Take two node labels $A$ and $B$ and
suppose we wish to require $A$-nodes to have either one outgoing $a$-edge or one
outgoing $b$-edge to a node with label $B$. The schema formalism allows us to
make the following restriction.
\[
  A \rightarrow a:B^\MAYBE, b:B^\MAYBE\,.
\]
Alone, it is insufficient as it allows nodes that do not fulfill the disjunctive
requirement: a $A$-node that has no outgoing edge or has both outgoing edges. We
remove those cases with the help of a positive and a negative query. Namely, we
define
\[
  p = A[(a+b)]
  \qquad
  \text{and}
  \qquad
  q = A[a][b]
\]
and we observe that in a graph that conforms to the above schema any node with
label $A$ that satisfies $p$ and does not satisfy $q$ has precisely one
outgoing edge.

\paragraph{Enforcing tree structure}
In our reduction we aim at constructing a tree-shaped counter examples and we
use the positive query to diligently enforce disjunction in every node. In
essence, the positive query will traverse the counter-example and impose
satisfaction of a relevant query in every node. We present this device on an
example where we define rooted binary trees. The general shape of the tree
follows the schema in Figure~\ref{fig:tree-schema}.
\begin{figure}[htb]
\begin{tikzpicture}[
  >=stealth',
  punkt/.style={circle,minimum size=0.125cm,draw,fill,inner sep=0pt, outer sep=0.125cm},
  ] 
  \begin{scope}
    \path (0,0) node[punkt] (Node) {} node[left=0.1cm]{\sl Node};
    \path (2,0) node[punkt] (Leaf) {} node[right=0.1cm]{\sl Leaf};
    \draw (Node)
    edge[->,loop above]
    node[above] {$a_1$}
    node[pos=0.2,left] {\small\MAYBE}
    node[pos=0.8,right] {\small\MAYBE}
    (Node);

    \draw (Node)
    edge[<-,loop below]
    node[below] {$a_2$}
    node[pos=0.2,right] {\small\MAYBE}
    node[pos=0.8,left] {\small\MAYBE}
    (Node);

    \draw (Node)
    edge[->,bend left]
    node[above] {$a_1$}
    node[pos=0.2,above] {\small\MAYBE}
    node[pos=0.8,above] {\small\MAYBE}
    (Leaf);

    \draw (Node)
    edge[->,bend right]
    node[below] {$a_2$}
    node[pos=0.2,below] {\small\MAYBE}
    node[pos=0.8,below] {\small\MAYBE}
    (Leaf);

  \end{scope}
\end{tikzpicture}
\caption{\label{fig:tree-schema}Example schema for modeling trees.}
\end{figure}
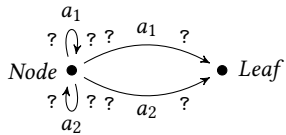

Naturally, the schema alone is insufficient to capture the right structure of
the tree.
Consequently, additional requirements are imposed with the help of the following
negative Boolean query
\[
  q =
  \mathit{Node}[a_1\cdot\mathit{Node}][a_1\cdot\mathit{Leaf}]
  +
  \mathit{Node}[a_2\cdot\mathit{Node}][a_2\cdot\mathit{Leaf}]
  +
  [a_1^-][a_2^-]
\]
that ensures that an inner node does not have two outgoing edges with the same
label and that no node has two incoming edges. We point out that when $q$ is
not satisfied at a node, schema ensures that it has at most one incoming
edge. To enforce the correct tree structure we define the following unary query 
\begin{multline*}
  p_\mathit{Tree}(x) =
  \big(((\mathit{Node}[a_1][a_2] \cdot a_1)^*\cdot\mathit{Leaf}\cdot(a_2^-)^*
  \cdot 
  a_1^-\cdot a_2)^*\cdot{} \\
  \mathit{Leaf} \cdot (a_2^-)^*\big)(x,x)\,.
\end{multline*}
The key observation here is that $a_1^-$ is always followed by $a_2$ and the
query can move up the tree only after a leaf has been reached. This ensures a
proper traversal of the structure, with every node satisfying the pattern
$\mathit{Node}[a_1][a_2]$. Consequently, for any connected graph $G$ that
conforms to the above schema, satisfies $p$, and does not satisfy $q$, $G$ is a
binary tree.

\paragraph{The input of the reduction}
We fix an ATM $M=(A,K,q_0,\delta_1,\delta_2)$ whose space is bounded by
$\mathit{poly}(n)$ and we fix an input word
$w\in(A\setminus\{\rhd,\lhd,\square\})^*$. We let $n=|w|$,
$m=\mathit{poly}(|w|)$, and assume that $A=\{a_1,\ldots,a_k\}$ and that
$K=\{q_0,q_1,\ldots,q_\ell\}$. Throughout the description of the reduction,
unless we say otherwise, we use $a,b$ to range over symbols in $A$, we use $q,p$
to range over states in $K$, and we use $i,j$ to range over tape positions
$\{1,\ldots,m\}$.

\paragraph{The schema}
We construct a schema $S$ whose signature is
\begin{align*}
  & \Sigma_S=\{\Config,\Pos,\mathit{Symb},\mathit{St}\}\,,\\
  & \Gamma_S=\{\forall_1,\forall_2,\exists_1,\exists_2,\pos_1,\ldots,\pos_m\}
    \cup \{a_1,\ldots,a_k\} \cup \{q_0,\ldots,q_\ell\}\,.
\end{align*}
In essence, $\Config$-nodes represent configurations and $\Pos$-nodes represent
tape cells. The edges labeled with $\{\forall_1,\forall_2,\exists_1,\exists_2\}$
are transition edges that connect configurations. The schema $S$ is presented in
Figure~\ref{fig:reduction-schema}.
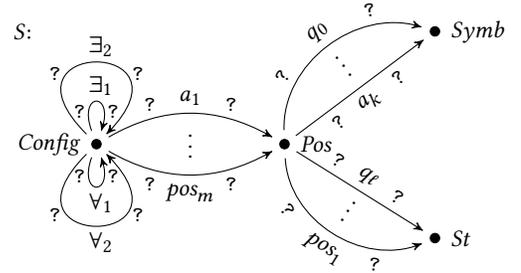
\begin{figure}[htb]
\begin{tikzpicture}[
  >=stealth',
  punkt/.style={circle,minimum size=0.125cm,draw,fill,inner sep=0pt, outer sep=0.125cm},
  ] 
  \begin{scope}
    \node[left] at (-0.75,1.5) {$S$:};
    \path (0,0) node[punkt] (Config) {} node[left=0.1cm]{$\Config$};
    \path (2.5,0) node[punkt] (Pos) {} node[right=0.1cm]{$\Pos$};

    \path (4.5,1.5) node[punkt] (Symb) {} node[right=0.1cm]{\it Symb};
    \path (4.5,-1.25) node[punkt] (St) {} node[right=0.1cm]{\it St};
    
    \draw[->] (Node)
    .. controls +(110:0.75) and +(70:0.75) ..
    node[above=-1pt] {$\exists_{\mathrlap{1}}$}
    node[pos=0.2,left=-2pt] {\small\MAYBE}
    node[pos=0.8,right=-2pt] {\small\MAYBE}
    (Node);

    \draw[->] (Node)
    .. controls +(135:2) and +(45:2) ..
    node[above=-1pt] {$\exists_{\mathrlap{2}}$}
    node[pos=0.3,left=-1pt] {\small\MAYBE}
    node[pos=0.7,right=-1pt] {\small\MAYBE}
    (Node);

    \draw[->] (Node)
    .. controls +(-110:0.75) and +(-70:0.75) ..
    node[below=-1pt] {$\forall_{\mathrlap{1}}$}
    node[pos=0.2,left=-2pt] {\small\MAYBE}
    node[pos=0.8,right=-2pt] {\small\MAYBE}
    (Node);

    \draw[->] (Node)
    .. controls +(-135:2) and +(-45:2) ..
    node[below=-1pt] {$\forall_{\mathrlap{2}}$}
    node[pos=0.3,left=-1pt] {\small\MAYBE}
    node[pos=0.7,right=-1pt] {\small\MAYBE}
    (Node);
    
    \draw[] (Config)
    edge[->,bend right]
    node[below] {$\pos_m$}
    node[pos=0.25,below] {\small\MAYBE}
    node[pos=0.75,below] {\small\MAYBE}
    (Pos);
    
    \path (Config) -- node[below=2pt,sloped]{\smash{\vdots}} (Pos);
    
    \draw[] (Config)
    edge[->,bend left]
    node[above] {$a_1$}
    node[pos=0.25,above] {\small\MAYBE}
    node[pos=0.75,above] {\small\MAYBE}
    (Pos);

    \draw[bend angle=0] (Pos)
    edge[->,bend right]
    node[below,sloped] {$a_k$}
    node[pos=0.25,below,sloped] {\small\MAYBE}
    node[pos=0.75,below,sloped] {\small\MAYBE}
    (Symb);
    
    \path (Pos)
    .. controls +(55:0.75) and +(-165:0.75) ..
    node[sloped]{\smash{\vdots}}
    (Symb);
    
    \draw[bend angle=55] (Pos)
    edge[->,bend left]
    node[above,sloped] {$q_0$}
    node[pos=0.25,above,sloped] {\small\MAYBE}
    node[pos=0.75,above,sloped] {\small\MAYBE}
    (Symb);
    
    \draw[bend angle=0] (Pos)
    edge[->,bend right]
    node[above,sloped] {$q_\ell$}
    node[pos=0.25,above,sloped] {\small\MAYBE}
    node[pos=0.75,above,sloped] {\small\MAYBE}
    (St);
    
    \path (Pos)
    .. controls +(-65:0.75) and +(185:0.75) ..
    node[below,sloped]{\smash{\vdots}}
    (St);
    
    \draw[bend angle=55] (Pos)
    edge[->,bend right]
    node[below,sloped] {$\pos_1$}
    node[pos=0.25,below,sloped] {\small\MAYBE}
    node[pos=0.75,below,sloped] {\small\MAYBE}
    (St);

  \end{scope}
\end{tikzpicture}
\caption{\label{fig:reduction-schema}Schema for the reduction.}
\end{figure}
We introduce macros that illustrate the intended meaning of the remaining edge
labels. The first macro checks that the symbol at position $i$ on the tape is
$a$.
\[
  \Symbol_{i,a} = \Config\left[\pos_i\cdot a\right]\,.
\]
The next one checks that the configuration is a given state $q$ with the head at
a given position $i$.
\[
  \State_{i,q} = \Config\left[pos_i\cdot q\right]\,.
\]
Finally, we also introduce a macro that asserts the state of a configuration
without any constraint on the position of the head.
\[
  \State_{q} = \Config\big[\textstyle\bigplus_i pos_i\cdot q\big]\,.
\]
And analogously, a macro that asserts heads position only
\[
  \mathit{Head}_i = \Config\big[\textstyle\bigplus_q pos_i\cdot q\big]\,.
\]
\paragraph{The negative query}
We define a number of queries that detect violations of good structure of a run;
their union will be used as the negative query. First, we identify
configurations that has two different symbols at a position of the tape.
\[
  q_\mathit{TwoSymbols} =
  \Config\left[\textstyle\bigplus_i\bigplus_{a\neq b} \Symbol_{i,a} \cdot \Symbol_{i,b}\right]\,.
\]
Similarly, we identify configurations with two different heads.
\[
  q_\mathit{TwoHeads} =
  \Config\left[
    \textstyle\bigplus_{i\neq j \lor p \neq q} \State_{i,q} \cdot \State_{j,p}
  \right]\,.
\]
Next, we identify configurations with outgoing transition edges that do not fit
their state.
\[
  q_{\mathit{BadTransitionEdges}} = 
  \Config\left[
  \begin{aligned}
    & \textstyle\bigplus_{q\in K_\forall} \State_q[\exists_1 + \exists_2] + {} \\ 
    & \textstyle\bigplus_{q\in K_\exists} \State_q[\forall_1 + \forall_2] + {} \\
    & \State_{q_\mathit{yes}}[\forall_1 + \forall_2 + \exists_1 + \exists_2] + {} \\
    & \State_{q_\mathit{no}}[\forall_1 + \forall_2 + \exists_1 + \exists_2] \\
  \end{aligned}
\right]\,.
\]
Additionally, identify configurations with existential states that have both
existential outgoing edges (the definition of a run requires precisely one).
\[
  q_{\mathit{TwoExistentialEdges}} =
  \textstyle\bigplus_{q\in K_\exists} \State_q[\exists_1][\exists_2]\,.
\]
The initial configuration, which is the only configuration with state $q_0$,
should be the root of the run and as such it should not have any incoming transition edges. 
\[
  q_{\mathit{BadTreeRoot}} = \State_{q_0}[\exists_1^- + \exists_2^- + \forall_1^- + \forall_2^-]
  \,.
\]
To make sure that the run is a tree, no configuration should have two incoming
transitions (note that the schema forbids more than one incoming edge with the
same label).
\[
  q_{\mathit{BadTreeNode}} = \Config\left[
    \begin{aligned}
      &[\exists_1^-][\exists_2^-] +
      [\exists_1^-][\forall_1^-] +
      [\exists_1^-][\forall_2^-] +{} \\
      &[\exists_2^-][\forall_1^-] +
      [\exists_2^-][\forall_2^-] +
      [\forall_1^-][\forall_2^-]
    \end{aligned}
    \right]\,.
\]
Similar requirements apply to tape: we do not allow tape positions that are
used by two different configurations.
\begin{align*}
  q_{\mathit{BadTape}} =
  & \textstyle\bigplus_{i\neq j} \Pos[\pos_i^-][\pos_j^-] + {}\\
  & \textstyle\bigplus_{p\neq q} \mathit{St}[p^-][q^-] + {}\\
  & \textstyle\bigplus_{a\neq b} \mathit{Symb}[a^-][b^-]\,.
\end{align*}
Finally, we construct the union of the above queries.
\begin{align*}
  q_M = {} & q_\mathit{TwoSymbols} + q_\mathit{TwoHeads} +
           q_{\mathit{BadTransitionEdges}} + {} \\
         & q_{\mathit{TwoExistentialEdges}} +
           q_{\mathit{BadTreeRoot}} + q_{\mathit{BadTreeNode}} +
           q_{\mathit{BadTape}}\,.
\end{align*}
\paragraph{The positive query}
We first construct a query that ensures that a configuration is valid and then
we design a path query that traverses the tree and ensures that each of its
configurations is valid. A valid configuration satisfies the following
queries. It has a head at some position.
\[
  p_{\mathit{Head}} = \Config\left[\textstyle\bigplus_i \mathit{Head}_i\right]\,.
\]
Every position has a symbol.
\[
  p_{\mathit{Tape}} =
  \Config\left[\textstyle\bigplus_{a} Symbol_{1,a}\right]
  \ldots
  \left[\textstyle\bigplus_{a} Symbol_{m,a}\right]\,.
\]
The configuration has the required outgoing transitions and only final states
are accepted in leaves.
\[
  p_{\mathit{Transition}} =
  \Config\left[
    \begin{aligned}
      & \textstyle\bigplus_{q\in K_\forall}\State_q [\forall_1][\forall_2] +{}\\
      & \textstyle\bigplus_{q\in K_\exists}\State_q [\exists_1 + \exists_2] +{}\\
      & \State_{q_\mathit{yes}} + \State_{q_\mathit{no}}
    \end{aligned}
  \right]\,.
\]
The positive query ensuring that transitions are executed properly is more
difficult to define and we decompose it into several macros. First, we define a
macro $\mathit{Move}_{i,q,a}$ that verifies that that a configuration in state
$q$ at position $i$ with symbol $a\in\Sigma$ has the right children
configurations. We define this macro depending on the type of state:

\medskip
\noindent
(1) For $q\in\{q_\mathit{yes},q_\mathit{no}\}$ no children are necessary (the
negative query $q_\mathit{BadTransitionsEdges}$ forbids any)
\begin{align*}
\mathit{Move}_{i,q,a} = {}
&\State_{q}\cdot\Symbol_{i,a}\,.
\intertext{
(2) For $q\in K_\exists$ we check that one of the transitions is implemented
  (the negative query $q_\mathit{TwoExistentialEdges}$ forbids more than
  one). We let $\delta_1(q,a) = (q_1,b_1,d_1)$ and $\delta_2(q,a) = (q_2,b_2,d_2)$.
}
\mathit{Move}_{i,q,a}
= {} & [\State_{i,q}\cdot\Symbol_{i,a}\cdot\exists_1\cdot\State_{i+d_1,q_1}\cdot\Symbol_{i,b_1}]\\
+ {}\ & [\State_{i,q}\cdot\Symbol_{i,a}\cdot\exists_2\cdot\State_{i+d_2,q_2}\cdot\Symbol_{i,b_2}]\,.
\intertext{
(3) For $q\in K_\forall$ both transitions must be implemented. Again we let
$\delta_1(q,a) = (q_1,b_1,d_1)$ and $\delta_2(q,a) = (q_2,b_2,d_2)$.
}
\mathit{Move}_{i,q,a}
= {} & [\State_{i,q}\cdot\Symbol_{i,a}\cdot\forall_1\cdot\State_{i+d_1,q_1}\cdot\Symbol_{i,b_1}]\\
\cdot {}\ \ & [\State_{i,q}\cdot\Symbol_{i,a}\cdot\forall_2\cdot\State_{i+d_2,q_2}\cdot\Symbol_{i,b_2}]\,.
\end{align*}
Now, a transition is executed correctly if the following positive query holds at
a configuration node. 
\[
  p_{\mathit{Execution}}=
  \Config\left[\textstyle\bigplus_{i,q,a} \mathit{Move}_{i,q,a}\right]\,.
\]
To handle the tape we need to make sure that 1) the tape of the initial
configuration contains precisely the input word and 2) that symbols at the
positions without head are copied correctly. For the first, we define the
following macro.
\[
  \mathit{InitTape} = \Symbol_{1,w_1}\cdot\ldots\Symbol_{n,w_n}
  \cdot\Symbol_{n+1,\square}\cdot\ldots\cdot\Symbol_{m,\square}\,.
\]
The next macro verifies that the symbol at a position $i$ of
the tape is a correct copy of its preceding configuration.
\[
  \mathit{PosCopy}_{i} =
  \left[\textstyle\bigplus_{a} \Symbol_{i,a}
  (\exists_1+\exists_2+\forall_1+\forall_2)^-\Symbol_{i,a}\right]\,.
\]
Naturally, when the head in the preceding configuration is at position
$i$, then we must only verify that symbols at positions other
than $i$ are copied.
\begin{align*}
  \mathit{TapeCopy} = 
  \textstyle\bigplus_{i}\big(
  & [(\exists_1+\exists_2+\forall_1+\forall_2)^-\mathit{Head}_i] \cdot{} \\
  & \mathit{PosCopy}_{1}\cdot\ldots\cdot\mathit{PosCopy}_{i-1} \cdot {} \\
  & \mathit{PosCopy}_{i+1}\cdot\ldots\cdot\mathit{PosCopy}_{m}\big)\,.
\end{align*}
Finally, the query that verifies the correctness of the tape follows.
\[
  p_\mathit{TapeCopy} = \Config\left[
    \State_{1,q_0}\cdot\mathit{InitTape} + \mathit{TapeCopy}
    \right]\,.
\]
Now, we take the conjunction of the queries that verify local correctness of a
configuration. 
\[
  p_\mathit{Config} =
  p_\mathit{Head}
  \cdot
  p_\mathit{Tape}
  \cdot
  p_\mathit{Transition}
  \cdot
  p_\mathit{Execution}
  \cdot
  p_\mathit{TapeCopy}\,.
\]
Additionally, we define a configuration that is a leaf (accepting)
\[
  p_\Accept = p_\Config \cdot \State_{q_\mathit{yes}}\,.
\]
And, the initial configuration
\[
  p_\Start = p_\Config \cdot \State_{q_0}\,.
\]
Finally, we define the positive query, based on the ideas of enforcing tree
structure in $p_\mathit{Tree}$. It traverses the counter-example and ensures
that it contains only good configurations.
\begin{align*}
  & p_{M,w} = {}
  p_\Start {} \cdot{} \\
  & \hspace{1ex}\big((p_\Config \cdot (\forall_1 + \exists_1 + \exists_2))^* \cdot
    p_\Accept \cdot (\exists_1^- + \exists_2^- + \forall_2^-)^*
    \cdot \forall_1^-\cdot\forall_2 \big)^*
    \cdot{}\\
  & \hspace{1.8ex}(p_\Config \cdot (\forall_1 + \exists_1 + \exists_2))^* \cdot
    p_\Accept \cdot
    (\exists_1^- + \exists_2^- + \forall_2^-)^* \cdot p_\Start\,.
\end{align*}
Before stating the main proof we present in Figure~\ref{fig:positive-query-automaton} a
conceptual automaton that corresponds to the above Boolean 2RPQ.
\begin{figure}[htb]
\begin{center}
\begin{tikzpicture}[>=latex]

  \node (start) at (-0.5,0.75) {};
  \node[circle,draw] (q0) at (0,0) {$q_0$};
  \node[circle,draw] (q1) at (1.75,0) {$q_1$};
  \node[circle,draw] (q2) at (4.5,0) {$q_2$};
  \node[circle,draw,double] (q3) at (6.25,0) {$q_3$};

  \draw[semithick] (start) edge[->,bend left] (q0);
  \draw[semithick] (q0) edge[->] node[above]{$p_\Start$} (q1);
  \draw[semithick,bend angle=15] (q1) edge[->, bend left] node[above]{$p_\Accept$} (q2);
  \draw[semithick,bend angle=15] (q2) edge[->,bend left] node[below]{$\forall_1^-\cdot\forall_2$} (q1);
  \draw[semithick] (q2) edge[->] node[above]{$p_\Start$} (q3);
  \draw[semithick] (q1) edge[->,loop] node[above=3pt,pos=0.6] {$p_\Config\cdot(\forall_1+\exists_1+\exists_2)$} (q1);
  \draw[semithick] (q2) edge[->,loop] node[above=3.5pt,pos=0.4] {$\forall_2^-+\exists_1^-+\exists_2^-$} (q2);
\end{tikzpicture}
\end{center}
\caption{\label{fig:positive-query-automaton} Conceptual automaton of the positive query $p_{M,w}$.}
\end{figure}
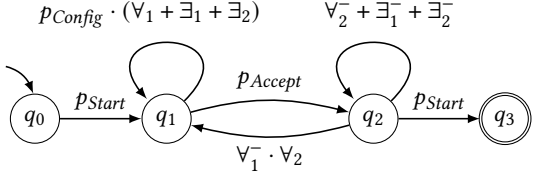
In the proof below, we refer to $p_{i,j}$ as the query defined with the above
automaton whose initial state is $q_i$ and final state is $q_j$. The main claim
follows.
\begin{claim}
  $p_{M,w}\not\subseteq_S q_M$ if and only if $M(w)=\mathit{yes}$.
\end{claim}
\begin{proof}
  For the \emph{if} direction, we take the accepting run $\lambda$ and construct
  the corresponding graph $G$ as follows. The nodes and their labels are as follows.
  \begin{align*}
    \Config^G& = \{c_n \mid n\in\dom(\lambda)\},\\
    \Pos^G& =    \{t_{n,i} \mid n\in\dom(\lambda), 1 \leq i \leq M\}, \\
    \mathit{St}^G & = \{s_n \mid n\in\dom(\lambda)\}, \\
    \mathit{Symb}^G & = \{e_{n,i} \mid n\in\dom(\lambda), 1 \leq i \leq M\}\,.
  \end{align*}
  The edges of $G$ are:
  \begin{enumerate}
    \item $(c_n,\pos_i,t_{n,i})$ for every $n\in\dom(\lambda)$ and
    $i \in\{1,\ldots,M\}$,
    \item $(t_{n,i},q,s_n)$ for every $n\in\dom(\lambda)$ where $q$ is the state of
    configuration $\lambda(n)$;
    \item $(t_{n,i},a,e_{i,n})$ for every $n\in\dom(\lambda)$ and
    $i \in\{1,\ldots,M\}$ where $a$ is the symbol at position $i$ of the tape of
    configuration $\lambda(n)$;
  \item $(c_n,\forall_1, c_{n\cdot 1})$ and $(c_n,\forall_2,c_{n\cdot 2})$ for every
    $n\in\dom(\lambda)$ such that the configuration $\lambda(n)$ is at state
    $q\in K_\forall$;
  \item $(c_n,\exists_j,c_{n\cdot j})$ for every $n\in\dom(\lambda)$ such that
    the configuration $\lambda(n)$ is at state $q\in K_\exists$ and $n$ has a
    child $n\cdot j$ in $\lambda$ for some $j\in\{1,2\}$.
  \end{enumerate}
  It is easy to show that $G$ satisfies the schema $S$, does not satisfy $q$,
  all $\Config$-nodes satisfy $p_\Config$, the root node satisfies $p_\Start$
  and every leaf node satisfies $p_\Accept$.

  With a simple induction, on the height of a node $n\in\dom(\lambda)$, we prove
  that for any $n\in\dom(\lambda)$ the node $c_n$ satisfies the query
  $p_{1,2}$. This shows that the root node $c_\varepsilon$ satisfies the query
  $p_{0,3}=p$.

  For the \emph{only if} direction, we take any $G$ that satisfies $S$,
  satisfies $p$, and does not satisfy $q$. W.l.o.g. we can assume that $G$ is
  connected; otherwise we take any connected component that satisfies $p$. We
  show that $G$ is a tree encoding an accepting run of $M$ on $w$. Note that $q$
  is a Boolean RPQ, and thus a single two-way regular expression. Thus, in the
  sequel we analyze its witnessing paths in $G$ but $p$ should not be confused
  with a binary query; a Boolean RPQ ask the existence of a witnessing path
  without the need to report its endings.

  Take any pair of nodes $u_0$ and $v_0$ such that there is a path
  from $u_0$ to $v_0$ that witnesses $q$ (which is a regular expression). Since
  $G$ does not have a node with two incoming edges ($q_\mathit{BadTreeNode}$ and
  $q_\mathit{BadTape}$ are not satisfied at any node), $u_0$ and $v_0$ are the
  same node. Consequently there is a path from $u_0$ to $u_0$ that witnesses
  $p_{1,2}$ and we show with an induction on the length of the path from $u_0$
  to any reachable $\Config$-node $v$ that there is a path form $v$ to $v$ that
  witnesses $p_{1,2}$, and consequently, $v$ satisfies $p_\Config$. This implies
  that $G$ has the form of a tree, all of its $\Config$-nodes satisfy
  $p_\Config$ and all its leaves satisfy $p_\Accept$. Moreover, we can construct
  an accepting run $\lambda$ from $G$ that shows that $M(w)=\mathit{yes}$.
\end{proof}
\noindent
Finally, we observe that the sizes of $S$, $p$, and $q$ are polynomial in the
size of $M$ and $w$, which proves the main claim.

\noindent
The hardness of containment in the presence of schema implies hardness of
the static analysis problems we study.  
\begin{lemma}
  \label{lemma:static-analysis-hardness}
  Type checking, equivalence, and schema elicitation are \EXPTIME-hard.
\end{lemma}
\begin{proof}
  We reduce the containment of unary 2RPQs in the presence of schema to the
  problems of interest. Note that by
  Theorem~\ref{theorem:query-containment-hardness} and
  Corollary~\ref{cor:boolean-rpq-containement}, containment of unary acyclic
  2RPQs is \EXPTIME-hard. We take any schema $S$ and two unary 2RPQs $p(x)$ and
  $q(x)$. In all reductions $S$ is the input schema and we assume a single unary
  constructor $\F=\{f_A\}$.

  \smallskip
  \noindent
  We begin by showing that testing $(T,S)\models \bigsqcap \Gamma_T$ is
  \EXPTIME-hard. We take the transformation $T$ defined with the following
  rules.
  \[
    A(f_A(x)) \leftarrow q(x)
    \quad\text{and}\quad
    a(f_A(x),f_A(x)) \leftarrow p(x)\,.
  \]
  We observe that $(T,S)\models \bigsqcap \Gamma_T$ if and only if $p(x)\subseteq_S q(x)$.

  \smallskip
  \noindent
  For equivalence, we define the following two transformations.
  \begin{align*}
    T_1:{}& A(f_A(x)) \leftarrow q(x)\,.\\
    T_2:{}& A(f_A(x)) \leftarrow q(x)\,,\quad A(f_A(x)) \leftarrow p(x)\,.
  \end{align*}
  We observe that $T_1\equiv_S T_2$ if and only if $p(x)\subseteq_S q(x)$.

  \smallskip
  \noindent
  For type checking we define the following transformation and output schema 
  \begin{align*}
    T:{}& A(f_A(x)) \leftarrow p(x)\,,&
    &A(f_A(x)) \leftarrow q(x)\,,\\
    &a(f_A(x),f_A(x)) \leftarrow q(x)\,.\\
    S':{}& A \rightarrow a:A^\ONE\,.
  \end{align*}
  We observe that that $T(S)\subseteq S'$ if and only if $p(x)\subseteq_S q(x)$.

  \smallskip
  \noindent
  To prove that schema elicitation is also \EXPTIME-hard, we take the previous
  transformation $T$, the input schema $S$, and show that $p(x)\subseteq_S q(x)$
  if and only if the $\subseteq$-minimal schema that captures the output graphs is
  precisely $S'$. We observe that deciding equivalence of two schemas is easily
  accomplished in polynomial time and therefore any algorithm for schema
  elicitation must require exponential time. 
\end{proof}

\end{document}